%% file: main.tex
\documentclass[11pt]{article}

\input{header}
\usepackage[margin=1in]{geometry}
\title{Adaptive Off-Policy Inference for M-Estimators\\ Under Model Misspecification}

\begin{document}

\author{
James Leiner\textsuperscript{1}, Robin Dunn\textsuperscript{2}, Aaditya Ramdas\textsuperscript{1,3} 
}

\maketitle
\begin{tabular}{l}
\textsuperscript{1}Department of Statistics and Data Science, Carnegie Mellon University \\
\textsuperscript{2}Novartis Pharmaceuticals Corporation, Advanced Methodology and Data Science \\
\textsuperscript{3}Machine Learning Department, Carnegie Mellon University \\
\end{tabular}

\begin{abstract}
\input{abstract}
\end{abstract}
\input{body}

\begin{appendix}
\input{appendix}

\end{appendix}

\section*{Acknowledgments}
\input{acknowledgements}

\bibliographystyle{chicago}
\bibliography{ref}

\end{document}

%% file: abstract.tex
When data are collected adaptively, such as in bandit algorithms, classical statistical approaches such as ordinary least squares and $M$-estimation will often fail to achieve asymptotic normality. Although recent lines of work have modified the classical approaches to ensure valid inference on adaptively collected data, most of these works assume that the model is correctly specified. The misspecified setting poses unique challenges because the parameter of interest itself may not be well-defined over a non-stationary distribution of rewards. We therefore tackle the problem of \emph{off-policy} inference in adaptive settings, where we uniquely define a projected solution over a stationary evaluation policy. Our method provides valid inference for $M$-estimators that use adaptively collected bandit data with a possibly misspecified working model.  A key ingredient in our approach is the use of flexible approaches to stabilize the variance induced by adaptive data collection. A major novelty is that the procedure enables the construction of valid confidence sets even in settings where treatment policies are unstable and non-converging, such as when there is no unique optimal arm and standard bandit algorithms are used. Empirical results on semi-synthetic datasets constructed from the Osteoarthritis Initiative demonstrate that the method maintains type I error control, while existing methods for inference in adaptive settings do not cover in the misspecified case. 

%% file: body.tex
\section{Introduction}
Adaptive data collection has become a common practice in modern data analysis due to the rise of reinforcement learning and sequential data collections in contexts such as personalized healthcare \citep{info:doi/10.2196/jmir.7994}, web recommendations \citep{afsar2022reinforcement}, and clinical trials \citep{https://doi.org/10.1002/sim.3720,info:doi/10.2196/18477}. In contrast to classical settings, where i.i.d.\ data are observed, adaptively collected data allows an analyst to interact with the decision-making algorithm in various ways, leading to sequentially dependent data. A canonical example is the (contextual) bandit problem where an analyst is allowed to make a choice of treatment (or ``arm'') at each time step and then observes a reward. Bandit algorithms attempt to maximize the total expected reward over time. These algorithms tend to result in data collection policies that explore treatment options in the early rounds and tend to shift to a greedier strategy that chooses the treatment with the highest estimated reward in the later rounds. When inference is conducted on data collected by bandit algorithms, the greedy search process can result in unstable variances \citep{pmlr-v80-deshpande18a,NEURIPS2020_6fd86e0a} and biased estimates of the average reward in each treatment arm \citep{NEURIPS2019_65b1e92c} that make naive approaches to inference invalid. 

There is a rich literature for adapting to the difficulties that arise for inference on adaptively collected data. Typical approaches focus on limiting the variance in each data collection using techniques such as propensity score truncation \citep{cook2023semiparametric} or post-hoc reweighting of the data when constructing an estimator \citep{bibault2021postcontext,syrgkanis2024postreinforcementlearninginference}. These approaches allow for assumption-lean inference in semiparametric settings, but they limit the study to relatively simple functionals such as the average reward for a fixed treatment or contrasts between average rewards for fixed arms. 

An alternative set of approaches assumes a linear model linking treatment ($X_{t} \in \R^{d}$) and rewards ($Y_{t} \in \R$),
$$ Y_{t} = \theta^{T}X_{t} + \epsilon_{t},$$ where $\epsilon_t$ is i.i.d.\ random noise. 
It is striking that even in this simple setting, estimators such as ordinary least squares fail to achieve asymptotic normality when $X_{t}$ depends on previous rounds $\{X_{i},Y_{i}\}_{i=1}^{t-1}$. A classical result \citep{1982laiwei} demonstrates that a necessary condition for the OLS estimate of $\theta$ to be asymptotically normal is when there exists a deterministic sequence of positive definite matrices $\{B_{T}\}_{T=1}^{\infty}$ such that $B_{T}^{-1}\sum_{t=1}^{T}X_{t}X_{t}^{T} \overset{p}{\to} I_{d}$. More recent literature \citep{NEURIPS2020_6fd86e0a,pmlr-v80-deshpande18a,khamaru2021near} makes the point that this assumption typically does not apply for common bandit algorithms when the margin (that is, the difference in mean rewards between arms) is zero. \cite{NEURIPS2020_6fd86e0a} proposes a solution through batching, by fixing the sampling rule for each batch and letting the number of observations within each batch go to infinity. This means that even if the sampling rule does not concentrate \emph{across time}, it is fixed within each batch, allowing the empirical covariance matrix within each batch to concentrate. An alternative set of approaches \citep{pmlr-v80-deshpande18a,khamaru2021near} uses an online debiasing approach, where the finite sample bias is controlled through regularization. Other approaches generalize this this to the generalized partial linear model \cite{10.1214/24-AOS2485} and for estimating equations \citep{NEURIPS2023_a399456a}.

We note that across all of these approaches, \emph{correct specification} of the working model is crucial to ensuring valid inference. Correct model specification is an unlikely assumption to hold in practical settings. In the non-adaptive setting, it is common to instead conduct inference on a projected solution, 
\begin{equation} \label{eqn:projected_linear}
\theta^{\star}:=\argmin_{\theta}\mathbb{E}\left[\left(Y - \theta^{T}X\right)^{2}\right]
\end{equation}
It is well known that under mild regularity conditions, the ordinary least squares estimate with sandwich estimators of the variance will cover the projected solution \citep{white1982maximum}, but no analogous result exists in the adaptive setting.

Moving beyond the linear case, \cite{zhang2021mestimators} considers statistical inference on $M$-estimators in the contextual bandit problem, allowing for more complex models to describe the reward structure, such as generalized linear models. However, this work again requires that the conditional mean of the model is correctly specified in the working model. However, in adaptive settings, selecting a target corresponding to a projected solution analogous to \cref{eqn:projected_linear} is not straightforward because the data-generating mechanism is non-stationary over time and dependent on user choices. We therefore define a target with respect to a fixed evaluation policy, which yields an off-policy projection parameter as an inferential target that remains interpretable under persistent adaptivity; formal details are deferred to \cref{sec:problem}.

An alternative approach \citep{zhang2023statistical} performs inference for $Z$-estimators and allows for a misspecified model, but only under a finite amount of adaptivity. In this framework, $n$ individuals are tracked over $T$ time periods, with $T$ fixed and $n$ growing to $\infty$. The advantage of this approach is that it allows for non-stationary reward and context distributions over time, but it relies on the number of individuals enrolled in the trial to tend to $\infty$ and does not allow adaptive decisions to be made separately for each individual in the trial.

\subsection{Our Contributions}

In this work, we propose a methodology for inferring parameters in $M$-estimation problems under a potentially misspecified working model and non-vanishing levels of adaptivity (i.e.\ without assuming the adaptivity eventually vanishes). The main result, \cref{thm:clt}, provides a Central Limit Theorem that enables the construction of confidence intervals when the variance (which varies over time and may not converge) of the score function can be estimated consistently. 

The problem of estimating the variance of the score function is nontrivial because using naive empirical estimates when applying \cref{thm:clt} is only valid when the variance is a fixed deterministic quantity. In many adaptive settings, such as bandit problems where expected rewards between arms are comparable, action selection probabilities will be non-convergent random quantities, leading to unstable variances \citep{NEURIPS2020_6fd86e0a}. To solve this problem, we propose flexible methods to construct valid time-varying plug-in estimates for the variances at each time step.

In a work that is concurrent to ours,  \cite{guo2025statisticalinferencemisspecifiedcontextual} also provides a procedure for inference on $Z$-estimators trained on adaptively collected data via inverse propensity weighting, but only under the assumption that treatment policies converge as $T \rightarrow \infty$. Since many common bandit algorithms fail to converge under model misspecification, the authors provide a sufficient set of conditions to guarantee policy convergence (e.g., continuous policies that are sufficiently smooth). We take a different approach --- rather than restricting ourselves to only particular classes of policy that are guaranteed to converge, we allow the analyst to assign treatments in relatively arbitrary ways and then adjust for the adaptivity post hoc by estimating the conditional variance \emph{separately} at each time step. In cases where the policy converges, these machine-learning-based approaches can be replaced with simple empirical estimates of the variance, and our procedure can be simplified considerably. 

We note that before now, the problem of estimating projection parameters under non-finite amounts of adaptivity was not solved even for simple linear models. Of course, this is merely a special case of $M$-estimation so our proposed methodology can be applied straightforwardly. We discuss this particular application as a running example throughout the remainder of the paper.

Our primary results are written under the assumption that the number of time steps tends to infinity, with only a single observation at each time step. However, this can be easily extended to the batched setting, where multiple observations are observed at each time step.

\subsection{Paper Outline}
In \cref{sec:problem}, we define the problem, most crucially introducing definitions of target parameters inspired by the off-policy evaluation literature that remain tractable in the misspecified setting. In \cref{sec:results}, we present a CLT that enables inference for $M$-estimators in the adaptive setting given accurate (time-varying) estimates of the variance of the score function. \cref{sec:covariance_estimation} discusses practical strategies for estimating variance using flexible machine learning approaches. \cref{sec:empirical_results} presents empirical results on semi-synthetic datasets constructed from the Osteoarthritis Initiative, a publicly available longitudinal dataset provided by the NIH which tracks health outcomes of patients with osteoarthritis. The results verify the validity of the procedure and the failure of existing approaches to provide confidence sets with valid coverage. We provide concluding thoughts in \cref{sec:conclusion}, including potential limitations of this approach and possible directions for future work. All proofs for stated theorems are included in the Appendix.

\subsection{Notation}
We introduce some shorthand notation used throughout the paper. We define $[n] := \{1,...,n\}$ for a positive integer $n$. We denote $e_{j}$ as the $j$-th standard basis vector in $\R^{d}$. For a function $f_{\theta} : \R^{d} \rightarrow \R$, we refer to $\dot{f}_{\theta}$ as the first derivative of $f$ with respect to $\theta$, $\ddot{f}_{\theta}$ as the Hessian matrix with respect to $\theta$, $\dddot{f}_{\theta}$ as the third derivative with respect to $\theta$, and so on. For $M\in \R^{d_1 \times d_2 \times d_3}$, we write $\norm{M}_{1} = \sum_{i=1}^{d_{1}}\sum_{j=1}^{d_2}\sum_{k=1}^{d_3}|a_{i,j,k}|$.
For two matrices $A,B$, $A \succeq B$ means that $A-B$ is positive semidefinite.

When taking expectations, $\E_{\mathcal{P},\pi}$ denotes the expectation under the assumption that $(X_t,A_t,Y_t) \sim  p \left(y | x, a\right)\pi(a|x) p(x)$. When no subscript is denoted, the expectation is over the entire joint distribution of $\{X_t,A_t,Y_t\}_{t=1}^{T}$. 

\section{Problem Setup} \label{sec:problem}

We assume a stochastic bandit environment in which features $X_{t} \in \mathbb{R}$ are observed in each round $t\in [T]$. After observing a feature, an analyst chooses an action $A_{t}$ from a finite set of $K$ actions $\mathcal{A}:= \{1,...,K\}$. After selecting an action, an outcome $Y_{t}$ is observed. We denote $Y_{t}(a)$, where $a \in \mathcal{A}$, to be the counterfactual result had the analyst chosen $a$ in round $t$, regardless of the actual action taken. We assume that the joint distribution of features and potential outcomes is independent and identically distributed across time. 

\begin{assumption} \label{assumption:potential_outcomes}
$\{(X_{t}, Y_{t}(1), \ldots, Y_t(K))\}_{t=1}^{T} \overset{\mathrm{iid}}{\sim}  \mathcal{P}$ .
\end{assumption}

Although potential outcomes are distributed i.i.d., we allow the analyst to choose the treatment adaptively based on a pooled history $\mathcal{H}_{t-1} := \{ (X_{i},A_{i},Y_{i}) \}_{i=1}^{t-1}$. Formally, we say $\pi_{t}(a|X_{t}) := \mathbb{P} (A_{t} =a | X_{t}, \mathcal{H}_{t-1}) \in \sigma(\mathcal{H}_{t-1}) $ for all $a \in \mathcal{A}$. We can summarize the conditional density (given $\mathcal{H}_{t-1}$) over actions, features, and outcomes at time step $t$ as
\begin{equation} \label{eqn:distribution}
    (X_t, A_t,Y_t) \sim p \left(y | x, a\right)\pi_{t}(a|x)  p(x),
\end{equation}
where densities $p(x)$ and $p(y|x,a)$ are invariant over time and defined by \cref{assumption:potential_outcomes}, but $\pi_{t}(a|x)$ is adaptive over time, because the analyst has the option to change the assignment probabilities in reaction to $\mathcal{H}_{t-1}$. 

Our goal is to construct confidence regions for a parameter $\theta^{\star}$, which we define as the expected maximizer of some function $m_{\theta}$ in some parameter space $\Theta$. For example, $m_{\theta}(x,a,y):= -(y -\theta^{T}(x,a))^{2}$ corresponds to an ordinary least squares regression, and $m_{\theta}(x,a,y):= -y(x,a)^{T}\theta + \psi ((x,a)^{T} \theta)$, where $\psi$ denotes the convex log-partition function, corresponds more generally to GLMs \citep{agresti2015foundations}. Some nuance is in order to determine the distribution over which the expected loss is minimized. Because the distribution of actions evolves over time, the distribution at each time step may be substantially different so the solution to  $\argmax_{\theta \in \Theta} \E\left[m_{\theta}(X_{t},A_{t},Y_{t}) \right]$ will also vary substantially across $t \in [T]$. 
In order to ensure there is a unique and stable inferential target, we tackle the problem of \emph{off-policy learning} and seek inference under a hypothetical policy $\pi_{e}(a|x)$ that is distinct from the policy used during treatment and invariant over time. 
\begin{equation} \label{eqn:theta_star}
\theta^{\star} := \argmax_{\theta \in \Theta} \mathbb{E}_{\mathcal{P}, A\sim\pi_{e}}\left[ m_{\theta}(X,A,Y) \right].
\end{equation}

The policy $\pi_{e}(a |x)$ may be a fixed and known distribution of interest. For example, $\pi_e(A_t = a| X_t) = 1 / |\mathcal{A}| \text{ for all } a \in \mathcal{A}$ and $t \in [T]$ would represent a hypothetical setting in which treatments are chosen uniformly within the population. Alternatively, $\pi_{e}(a | x)$ might be a fixed but unknown quantity to be estimated later, such as a hypothetical optimal policy or a limiting quantity such as $\lim_{t \rightarrow \infty} \mathbb{P}(A_{t}=a| X_{t}, \mathcal{H}_{t-1})$ (which we do not assume to exist in general in our paper). An appeal to Slutsky's theorem will allow us to substitute estimated versions of optimal policies after data collection to form our estimators, so long as the estimates are asymptotically independent of the history as $T \rightarrow \infty$. Of course, the interpretation of $\theta^{\star}$ can vary dramatically depending on the choice of $\pi_{e}$, as a simple example will elucidate. 

\begin{figure} 
    \centering
    \includegraphics[width=0.8\linewidth]{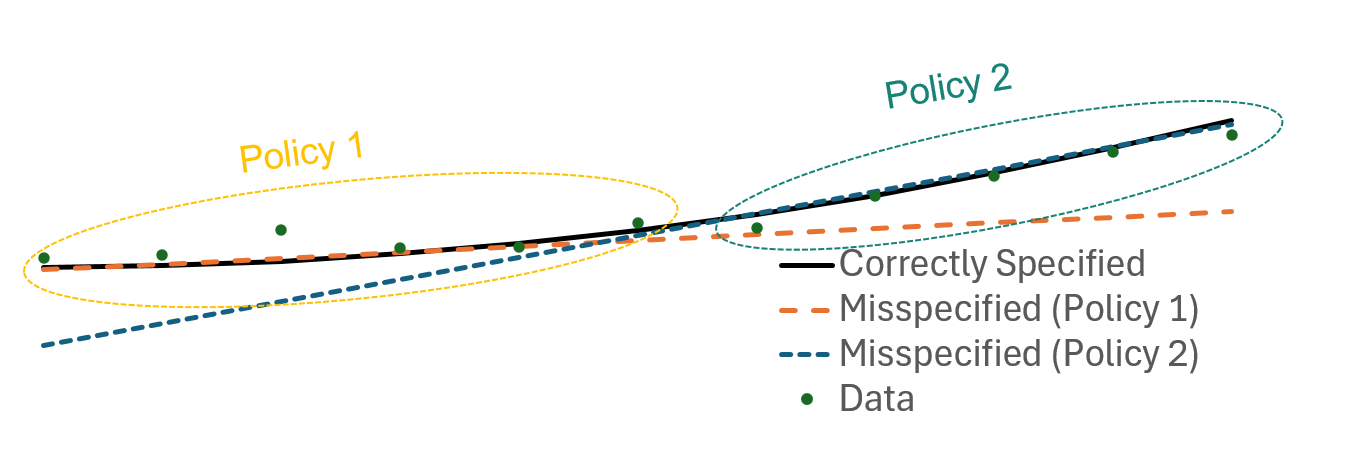}
    \caption{Illustration for \cref{example1}. Under the misspecified linear model $m_{\theta} = -(Y_{t} - \theta_{0} - \theta_{1}A_{t})^{2}$, $(\theta_{0}^{\star},\theta_{1}^{\star}) = (-0.2, 4)$ for the first policy and $(\theta_{0}^{\star},\theta_{1}^{\star}) = (-5.3, 15)$ for the second policy. Intuitively, the misspecified model will try to estimate a secant line at different points of the quadratic function, resulting in very different interpretations for the target parameter. As such, the choice of evaluation policy is critical for properly interpreting the target parameter.}
    \label{fig:example_misspecified}
\end{figure}

\begin{example}[Effect of $\pi_{e}$ on $\theta^{\star}$] \label{example1}
Consider data generated as $Y_{t} \sim N(6A_{t}^{2},1)$ under two treatment regimes, where $\mathcal{A} = \{0, 0.1, 0.2, 0.3, 0.4, 0.5, 0.6, 0.7,0.8,0.9,1\}$. The first policy is uniform over $A_t \in \{0,0.1,0.2,0.3,0.4,0.5\}$. The second policy is uniform over $A_t \in \{0.6,0.7,0.8,0.9,1.0\}$. Under the correctly specified model $m_{\theta}(A_{t},Y_{t}) := -(Y_{t} - \theta_{0} - \theta_{1}A_{t}^{2})^{2}$, $(\theta_{0}^{\star},\theta_{1}^{\star}) = (0, 6)$ for both policies. On the other hand, under an incorrectly specified linear model, $\theta_{0}^{\star}$ and $\theta_{1}^{\star}$ can vary markedly as seen in \cref{fig:example_misspecified}. 

\end{example}

Although this framework is robust enough to accommodate parameter estimation for relatively complex parametric models that describe the relationship between rewards, features, and treatments, it can also be used to target simple functionals such as average treatment effect. 

\begin{remark}[Average Treatment Effect]
We can use this framework to target the average treatment effect by using a vector of indicator variables as the chosen model with a square loss function. We can permit any evaluation policy that satisfies the standard causal assumption of $Y_{t}(a) \perp \!\!\! \perp A_{t}$ under $\pi_{e}$, and such that $\pi_{e}(A_{t}=a) > 0$. 

Formally, we let $\theta = (\theta_1, ..., \theta_{|\mathcal{A}|})$ and consider setting  $m_{\theta} = -(Y_{t} - \sum_{a=1}^{K} 1_{A_{t} = a} \theta_{a})^{2}$. First order conditions imply that for all $t$,
$$0= \mathbb{E}_{\mathcal{P}, \pi_{e}}\left[ Y_{t} - \sum_{a=1}^{K} 1_{A_{t} = a} \theta^{\star}_{a}\right].$$
Noting that $Y_{t} = Y_{t}(a) 1_{A_{t} = a}$ and rearranging terms yields, 
\begin{align*}
 \theta^{\star}_{a} = \frac{\mathbb{E}_{\mathcal{P}, \pi_{e}}\left[Y_{t}(a) 1_{A_{t} = a} \right]  }{ \mathbb{E}_{\mathcal{P}, \pi_{e}}\left[ 1_{A_{t} = a} \right]  } =  \frac{ 
  \mathbb{E}_{\mathcal{P}}\left[Y_{t}(a)\right] \mathbb{E}_{\mathcal{P}, \pi_{e}}\left[ 1_{A_{t} = a} \right]  }{ \mathbb{E}_{\mathcal{P}, \pi_{e}}\left[ 1_{A_{t} = a} \right]  }  = \mathbb{E}_{\mathcal{P}}\left[Y_{t}(a)\right]= \mathbb{E}_{\mathcal{P}}\left[Y(a)\right].
\end{align*}
To infer the average effect of treatment $j$ versus treatment $k$, we can then target $\eta^{T} \theta^{\star}$ for some contrast vector $\eta:= e_{j} - e_{k}$ (recall $e_j$ refers to the $j$-th standard basis vector in $\R^{K})$. For example $\eta = (-1,1)$ in the setting of $K=2$ would recover the  average treatment effect $ \mathbb{E}\left[Y(1) - Y(0) \right].$
\end{remark}

\paragraph{Choice of estimator} Similar to \cite{zhang2021mestimators}, one option is to consider estimators of the form
\begin{equation*}
\hat{\theta}_{0} :=  \argmax_{\theta} \sum_{t=1}^{T}  w_{t} m_{\theta}(X_{t},A_{t},Y_{t}),
\end{equation*}
where $w_{t} \in \sigma(\mathcal{H}_{t-1},X_{t})$. Note that when the model is correctly specified, we have at every $t$ that
\begin{equation}\label{eqn:theta_star_mis}
  \theta^{\star} = \argmax_{\theta \in \Theta} \E \left[ m_{\theta} (X_{t},A_{t},Y_{t}) |A_{t}, X_{t} \right] \text{ for all } A_{t} \in \mathcal{A}, X_{t} \in \mathbb{R}.  
\end{equation} 
When this is true conditionally, it will also be true marginally for \emph{any choice} of policy. Therefore, under no model misspecification, we have at every time step $t$ that 
\begin{align*}
    \theta^{\star} 
= \argmax_{\theta \in \Theta}\mathbb{E}_{\mathcal{P}, \pi_{e}}\left[ m_{\theta}(X_{t},A_{t},Y_{t}) \right]  
= \argmax_{\theta \in \Theta}\mathbb{E}_{\mathcal{P}, \pi_{t}}\left[ m_{\theta}(X_{t},A_{t},Y_{t}) |\mathcal{H}_{t-1} \right],
\end{align*}
removing the need to consider the choice of \emph{policy} when defining the target parameter. Assuming that the maximum occurs at a critical point of $m_{\theta}$, this implies that
$$0 = \E\left[\dot{m}_{\theta^{\star}}(X_{t},A_{t},Y_{t}) |X_{t},A_{t}\right] = \E\left[w_{t} \dot{m}_{\theta^{\star}}(X_{t},A_{t},Y_{t}) |X_{t},A_{t},\mathcal{H}_{t-1}\right],$$
for all $X_{t},A_{t}$ and all $w_{t} \in \sigma(X_{t},A_{t},\mathcal{H}_{t-1})$.
Thus, $w_{t}$ can be arbitrarily chosen based on $X_t$, $A_t$, and $\mathcal{H}_{t-1}$ without violating the first-order conditions. For this reason, previous work often uses $w_t$ as a free parameter to control the variance without having to worry that this type of re-weighting will change the target away from $\theta^{\star}$. As an example, \cite{zhang2021mestimators} chooses $w_{t} = \sqrt{\frac{\pi_{e}(A_{t} | X_{t})}{ \mathbb{P} (A_{t}  | X_{t}, \mathcal{H}_{t-1}) }}$, which allows the variance of the score function to converge to a value independent of history.  Other examples include online debiasing approaches \citep{pmlr-v80-deshpande18a,khamaru2021near} which use regularization to choose $w_{t}$ so that the variance of $\hat{\theta}_{T}$ is controlled.

Under model misspecification, however, the choice of weights has an impact on the target parameter that is being covered because at $\theta^{\star}$ we will require that at each time step $t$,
$$ 0 = \E_{\mathcal{P},\pi_{e}}\left[ \dot{m}_{\theta^{\star}}(X_{t},A_{t},Y_{t})\right] =\E_{\mathcal{P},\pi_{t}}\left[ w_{t}\E_{\mathcal{P}}\left[\dot{m}_{\theta^{\star}}(X_{t},A_{t},Y_{t}) | X_{t},A_{t}, \mathcal{H}_{t-1}\right]|\mathcal{H}_{t-1}\right],$$
which is \emph{only} true for the specific choice of $w_{t} = \frac{\pi_{e}(A_{t} | X_{t})}{ \mathbb{P} (A_{t}  | X_{t}, \mathcal{H}_{t-1}) }$. Thus, we are constrained to define an estimator such as
\begin{equation} \label{eqn:hat-theta-naive} 
\hat{\theta}_{0} =  \argmax_{\theta} \sum_{t=1}^{T}  \frac{\pi_{e}(A_{t} | X_{t})}{ \mathbb{P} (A_{t}  | X_{t}, \mathcal{H}_{t-1}) } m_{\theta}(X_{t},A_{t},Y_{t}).
\end{equation}
Using inverse propensity weights ensures that $\hat{\theta}_{0}$ will be unbiased for $\theta^{\star}$, but we must now consider alternative methods to stabilize the variance of the estimator, which we will discuss in the next section. 

\paragraph{MAIPWM-Estimator} As an alternative to \cref{eqn:hat-theta-naive}, we can potentially improve power through the use of a predictive model that we update over time, $f_{t} : \R  \times \mathcal{A}  \rightarrow \R$. $f_{t}$ is trained on $\mathcal{H}_{t-1}$ (along with potentially other information known to the experimenter a priori before the experiment). $f_{t}$ then takes in $(X_{t},A_{t})$ as an input and outputs an estimate for $  \E\left[ Y_{t} | X_{t},A_{t}\right]$. An alternative definition to  \cref{eqn:hat-theta-naive} is what we will term the \textbf{misspecified augmented inverse propensity weighted $M$-estimator (MAIPWM-Estimator)} 
\begin{equation} \label{eqn:hat_theta}
\tilde{\theta}_{T} = \argmax_{\theta \in \Theta}  \sum_{t=1}^{T} \sum_{a=1}^{K} \pi_{e}(A_{t} =a | X_{t}) \left( m_{\theta}(a,X_{t},f_{t}(a,X_{t}))   + \mathbbm{1}_{A_{t} = a} \frac{ m_{\theta}(A_{t},X_{t},Y_{t})  -  m_{\theta}(X_{t},A_{t},f_{t}(X_{t},A_{t})) }{ \mathbb{P} \left(A_{t} =a | X_{t},\mathcal{H}_{t-1} \right)  }\right).
\end{equation}
The estimator improves power when the model $f_{t}$ is a good estimator of $\E[Y_{t} | X_t,A_t]$
by augmenting observations with predictions from a model. This approach is routinely used when estimating simple functionals like the average treatment effect in causal inference \citep{10.1111/ectj.12097} and off-policy evaluation \citep{doi:10.1073/pnas.2014602118} and is referred to as the \emph{augmented inverse probability weighted} estimator in these settings. More recently, \cite{zrnic2024active} introduces a similar modification of the score function of $M$-estimators with predictions derived from black box machine learning models, but limits their findings to the non-causal, non-adaptive observational setting. 

A common approach to finding the maximum is to find the solution defined by the root of the \emph{score equation}, 
\begin{equation} \label{eqn:score}
s_{t,\theta} = \sum_{a=1}^{K} \pi_{e}(A_{t} =a | X_{t}) \left( \dot{m}_{\theta}(X_{t},a,f_{t}(X_{t},a))   + \mathbbm{1}_{A_{t}=a}\frac{ \left( \dot{m}_{\theta}(X_{t},a,Y_{t})  -  \dot{m}_{\theta}(X_{t},a,f_{t}(X_{t},a) \right)}{ \mathbb{P} \left(A_{t}| X_{t}, \mathcal{H}_{t-1} \right)  }\right).
\end{equation}
In order to limit the variance of this quantity, we will instead consider a second estimator $\hat{\theta}_{T}$ defined as the (approximate) root of $\sum_{t=1}^{T} \Sigma_{t}^{-1} s_{t,\theta}$ for specially constructed matrices $\Sigma_t \in \R^{d}$ designed to normalize the variance of $s_{t,\theta}$. However, the construction of $\Sigma_t$ will require consistent estimates of $\theta^{\star}$ so our procedure will involve a two-step process where $\tilde{\theta}_{T}$ is used to construct the stabilizing matrices $\Sigma_t$ and then a central limit theorem is proved for $\hat{\theta}_{T}$.

This approach will be most useful when an individual has access to powerful machine learning methods for prediction but wants to perform inference on a simpler parametric model that is more interpretable. If an analyst prefers not to manage a predictive model during the course of the experiments, substituting any constant for $f_{t}$ will recover \cref{eqn:hat-theta-naive}. Therefore, we will prove our main results with respect to $\tilde{\theta}_{T}$ and $\hat{\theta}_{T}$, but similar guarantees will exist for $\hat{\theta}_{0}$ when no predictive model is being updated as a special case.  

\section{Main Results} \label{sec:results}
Our approach will be to construct two sets of weights based on different filtrations. As before, $w_{t} \in \sigma(X_{t},\mathcal{H}_{t-1})$ is restricted to being the inverse propensity weights $\frac{\pi_{e}(A_{t} | X_{t})}{ \mathbb{P} (A_{t}  | X_{t}, \mathcal{H}_{t-1}) }$ to ensure that $\hat{\theta}_{T}$ is unbiased for $\theta^{\star}$. Similarly, we consider $\Sigma_{t} \in \sigma(\mathcal{H}_{t-1})$ that will allow us to adjust the second moment of $s_{t,\theta}$, while preserving the first moment of the score function in $\theta^{\star}$ to be $0$. By assumption, $\E\left[s_{t,\theta^{\star}}|\mathcal{H}_{t-1}\right] =0$. Define a martingale difference sequence (MDS) $Z_{t} := \Sigma_{t}^{-1/2} s_{t,\theta^{\star}}$. Then,
$$\E[Z_{t} | \mathcal{H}_{t-1}] = \E[\Sigma_{t}^{-1/2} s_{t,\theta^{\star}}| \mathcal{H}_{t-1}] =\Sigma_{t}^{-1/2}\E[s_{t,\theta^{\star}}| \mathcal{H}_{t-1}] = 0,$$
by the fact that $\Sigma_{t}^{-1/2}\in \sigma(\mathcal{H}_{t-1})$. A natural choice is to choose $\Sigma_{t}$ as 
$$ V_{t,\theta^{\star}} = \mathbb{E}_{\mathcal{P},\pi_{t}}\left[ s_{t,\theta^{\star}}s_{t,\theta^{\star}}^{T} |\mathcal{H}_{t-1}\right].$$ 

By re-weighting the score function to stabilize the variance but \emph{requiring} that the weights be determined from a coarser filtration than what was used for constructing $w_{t}$, we are able to preserve the first-order conditions that will be stated in \cref{assumption:maximum}. Of course,  $V_{t,\theta^{\star}}$ is unknown in practice and will need to be replaced with an estimate $\hat{V}_{t}$. We proceed with the remainder of this section assuming the existence of such an estimate but note that it is a non-trivial task as we only have a single observation corresponding to each $t$. The key idea for constructing an estimator is to leverage the potential outcomes framework of \cref{assumption:potential_outcomes}, which we discuss in \cref{sec:covariance_estimation}. 

Before stating the main result, we introduce several assumptions that will be required to ensure asymptotic normality. 

\begin{assumption} \label{assumption:differentiability}
The first three derivatives of $m_{\theta}(x,a,y)$ with respect to $\theta$ exist for every $\theta \in \Theta$, every $a \in \mathcal{A}$ and every $(x,y)$ in the joint support of $\mathcal{P}$.
\end{assumption}

The existence of the first two derivatives allows us to identify the quantity of interest as the expected maximizer of $m_{\theta}$ and estimate it by evaluating the critical points of this function, which we discuss in more detail in \cref{assumption:maximum}. The existence of a third derivative will allow us to use the Taylor expansion around the score function to form a confidence set covering $\theta^{\star}$. 

\begin{assumption} \label{assumption:maximum}
Let $\theta^{\star}$ be defined as in \cref{eqn:theta_star} and $\Theta \subset \R^{d}$ be a bounded parameter space such that \cref{assumption:differentiability} holds. We further assume that at every $t \in [T]$,
\begin{enumerate}
\item $\E_{\mathcal{P},\pi_{e}}\left[ \dot{m}_{\theta^{\star}}(X_{t},A_{t},Y_{t})\right] =0,$  
\item $-\E_{\mathcal{P},\pi_{e}}\left[ \ddot{m}_{\theta^{\star}}(X_{t},A_{t},Y_{t}) \right] \succeq H  $ for some positive definite matrix $H$,
\item For any $\epsilon> 0 $, there exists constants $\delta_1,\delta_2 > 0 $ such that 
\begin{itemize}
    \item $\inf_{\norm{\theta- \theta^{\star}} \ge \epsilon} \left\{ \mathbb{E}_{\pi_{e}}\left[ m_{\theta^{\star}}(X_{t},A_t,Y_t) - m_{\theta}(X_{t},A_t,Y_t) \right]  \right\} \ge \delta_1$,
    \item  $\inf_{\norm{\theta- \theta^{\star}} \ge \epsilon} \norm{\mathbb{E}_{\pi_{e}}\left[\dot{m}_{\theta}(X_{t},A_t,Y_t)\right]}_{1}  \ge \delta_2$.
\end{itemize}
\item The first four moments of $m_{\theta}(X_t,A_t,Y_t)$,  $\dot{m}_{\theta}(X_t,A_t,Y_t)$, and $\ddot{m}_{\theta}(X_t,A_t,Y_t)$ are bounded with respect to $\mathcal{P},\pi_{e}$. 
\end{enumerate}
\end{assumption}
The first two conditions of this assumption ensure that the quantity of interest $\theta^{\star}$ is the solution to a maximization problem and can be found by evaluating the critical points of $m_{\theta}$. The third condition mirrors classical assumptions to demonstrate the consistency of $Z$ estimators and $M$ estimators \citep{van2000asymptotic}. Having a well-separated solution ensures the uniqueness of $\theta^{\star}$ and allows estimators constructed from finite samples to converge appropriately. Note that we assume that $\theta^{\star}$ is well separated \emph{both} as a maximizer of $m_{\theta}$ and as a root of the score equation because we consider estimators defined in both of these ways within our procedure. A special case of the third condition is when $m_{\theta}$ is a continuously differentiable concave function with the maximum obtained at $\theta^{\star}$. The fourth condition is used to ensure that Lindeberg conditions are achieved when invoking the martingale central limit theorem and is analogous to conditions in classical proofs of the normality of $M$-estimators \citep{van2000asymptotic}.

\begin{assumption}[Bounded importance ratios] \label{assumption:bounded_importance}
There exist a constant $C_{1}> 0$ such that  $ \frac{1}{ \mathbb{P} \left(A_{t}=a| X_{t}, \mathcal{H}_{t-1} \right) } \le C_{1}$ for all $a\in \mathcal{A}$. 
\end{assumption}
Bounding the weights is important to ensure that the variance of the score function is bounded, so that the martingale law of large numbers and central limit theorems can be applied. It is possible to weaken this so that the $\mathbb{P} \left(A_{t}=a| X_{t}, \mathcal{H}_{t-1} \right)$ are allowed to converge to either $0$ or $1$ at appropriately slow rates, but we leave this generalization to future work.

\begin{assumption}[Finite Bracketing Numbers] \label{assumption:bracket3}
Define the set of functions
$ \mathcal{M}_{\Theta} =  \{ m_{\theta}(X_t,A_{t},Y_{t}) : \theta \in \Theta  \}$ and
$ \dot{\mathcal{M}}_{\Theta} =  \{c^{T}\dot{m}_{\theta}(X_t,A_{t},Y_{t}) : \theta \in \Theta , \norm{c} \le 1 \}$. We assume that for all $\epsilon>0$, the bracketing numbers  $N_{[]}(\epsilon, \mathcal{M}_{\Theta}, L_{2}(\mathcal{P},\pi_{e})) < \infty$ and $N_{[]}(\epsilon, \dot{\mathcal{M}}_{\Theta}, L_{2}(\mathcal{P},\pi_{e})) < \infty$.

\end{assumption}
\cref{assumption:bracket3} limits the complexity of the function class $\{ m_{\theta} : \theta \in \Theta\}$ so that a martingale law of large numbers can be applied which is required when constructing an argument for the consistency of $\hat{\theta}_{T}$ and $\tilde{\theta}_{T}$. We note that \cite{zhang2021mestimators} provides an intuitive sufficient condition to check to ensure that \cref{assumption:bracket3} is true. They show that whenever $m_{\theta}$ is Lipschitz in the sense that 
$$|m_{\theta}(X_{t},A_{t},Y_{t}) - m_{\theta'}(X_{t},A_{t},Y_{t})| \le h(X_{t},A_{t},Y_{t}) \norm{\theta - \theta'}$$
for some function $h$ such that $\E_{\mathcal{P},\pi_e}[h(X_t,A_t,Y_t)^{2}] < m $ for some constant $m$, then $\mathcal{M}_{\theta}$ will have finite bracketing numbers for all $\epsilon > 0 $. 
We place a similar assumption on the complexity of the model classes. 
\begin{assumption}[Model Class Complexity] \label{assumption:bracket_complexity}
The predictive model $f_{t} \in \sigma(\mathcal{H}_{t-1})$, and $f_t$ belongs to a function class $\mathcal{F}$. For this function class, there exists a single function $u_{\mathcal{F}}$ and a constant $m_{u} \in \R$ such that
$\mathbb{E}[u_{\mathcal{F}}(X_t,a,f_{t}(a,X_{t}))^{4}] < m_{u}$ for all $a \in \mathcal{A}$ and,
$$\sup_{\theta \in \Theta} m_{\theta}(x,a,f_{t}(a,x)) \le u_{\mathcal{F}}(x,a,f_{t}(a,x)),$$ 
$$\sup_{\theta \in \Theta} \norm{\dot{m}_{\theta}(x,a,f_{t}(a,x))}_{1} \le u_{\mathcal{F}}(x,a,f_{t}(a,x)),$$
$$\sup_{\theta \in \Theta} \norm{\ddot{m}_{\theta}(x,a,f_{t}(a,x))}_{1} \le u_{\mathcal{F}}(x,a,f_{t}(a,x)),$$
$$\sup_{\theta \in \Theta} \norm{\dddot{m}_{\theta}(x,a,f_{t}(a,x))}_{1} \le u_{\mathcal{F}}(x,a,f_{t}(a,x)),$$
for all $a \in \mathcal{A}$, $x$ in the support of $\mathcal{P}$ and $f_{t} \in \mathcal{F}$. Furthermore, we assume that there exists a constant $C_{2}$ not dependent on $t$ such thate
$$\E\left[ \norm{m_{\theta}(X_{t},a,f_{t}(X_t,a)) - m_{\theta'}(X_{t},a,f_{t}(X_t,a))}^{2}_{2} |\mathcal{H}_{t-1} \right] \le C_{2} \norm{ \theta - \theta'}_{2}^{2}$$ almost surely.

\end{assumption}
Note that a primary difference between \cref{assumption:bracket3} and \cref{assumption:bracket_complexity} is that \cref{assumption:bracket_complexity} is written in terms of expectations relative to the actual policy $\pi_t$ rather than $\pi_e$. This is necessary because \cref{assumption:bracket_complexity} bounds the variance of a term related to $f_t$, which is dependent on history whereas \cref{assumption:bracket3} is bounding the complexity of a quantity that is \emph{independent} of history. 
These assumptions allow us to prove a key lemma needed to ensure the consistency of $\hat{\theta}_{T}$ and $\tilde{\theta}_{T}$.

\begin{lemma} \label{lemma:uniform_lln}

Assume $\mathcal{G}_{\Theta}=\{ g_{\theta}(X_t,A_{t},Y_{t}) : \theta \in \Theta  \}$ is a class of functions such that for any $\epsilon >0$, the bracketing number  $N_{[]}(\epsilon, \mathcal{G}_{\Theta}, L_{2}(\mathcal{P},\pi_{e})) < \infty$. Define
\begin{equation} \label{eqn:r_t_theta}
R_{t}(\theta) =  \sum_{a=1}^{K} \pi_{e}(A_{t} =a | X_{t}) \left( g_{\theta}(a,X_{t},f_{t}(a,X_{t}))   + \mathbbm{1}_{A_{t} = a} \frac{ g_{\theta}(X_{t},A_{t},Y_{t})  -  g_{\theta}(X_{t},A_{t},f_{t}(X_{t},A_{t})) }{ \mathbb{P} \left(A_{t} =a | X_{t},\mathcal{H}_{t-1} \right)  }\right).
\end{equation}
Assume that there exists a constant $L$ not dependent on $t$ such that
$$\E\left[ \norm{m_{\theta}(X_{t},a,f_{t}(X_t,a)) - m_{\theta'}(X_{t},a,f_{t}(X_t,a))}^{2}_{2} |\mathcal{H}_{t-1} \right] \le L \norm{ \theta - \theta'}_{2}^{2}$$ almost surely.
Then under Assumptions \ref{assumption:potential_outcomes}-\ref{assumption:bounded_importance},
$$\sup_{\theta \in \Theta }\frac{1}{T} \sum_{t=1}^{T} \left(R_t(\theta) - \mathbb{E}[R_{t}(\theta) |\mathcal{H}_{t-1}]\right) \overset{p}{\to} 0.$$
\end{lemma}
We can think of \cref{lemma:uniform_lln} as a uniform version of the martingale weak law of large numbers, specifically adjusted for our setting where machine learning methods are used to augment data collection.

\begin{assumption} \label{assumption:domination}
There exists a fixed, integrable function  $u_{m}(X_{t},A_{t},Y_{t})$ such that for some $\delta >0$, $$\sup_{\theta \in \Theta:\norm{\theta - \theta^{\star}} < \delta} \norm{\dddot{m}_{\theta}(X_{t},A_{t},Y_{t})}_{1} \le u_{m}(X_{t},A_{t},Y_{t}),$$
and $\E_{\mathcal{P},\pi_e}\left[ u_{m}(X_{t},A_{t},Y_{t})^{2} \right]$ is bounded. 
\end{assumption}

This assumption mirrors that of classical approaches to proving the normality of $M$-estimators \citep{van2000asymptotic} and is used to make sure that certain quantities related to the third derivative of the objective function remain bounded when using Taylor expansion to construct the confidence set. 

We are now ready to state the theorem. 

\begin{theorem} \label{thm:clt} 
Let $\theta^{\star}$ be defined as \cref{eqn:theta_star} and $s_{t,\theta}$ defined as in \cref{eqn:score}. Assume $V_{t,\theta^{\star}} := \mathbb{E}_{\mathcal{P},\pi_{t}}\left[ s_{t,\theta^{\star}}s_{t,\theta^{\star}}^{T} | \mathcal{H}_{t-1}\right]$ is almost surely invertible and that there exists a sequence of random matrices $\{\hat{V}_{t}\}_{t=1}^{T}$ adapted to the filtration $\sigma(\mathcal{H}_{t-1})$ such that $\norm{\hat{V}_{t}^{-1/2} - V_{t,\theta^{\star}}^{-1/2}}_{\text{op}} \overset{p}{\to} 0$. Assume that the eigenvalues of both $\hat{V}_{t}$ and $V_{t,\theta^{\star}}$ are bounded above and below by constants $\delta_{\text{min}}, \delta_{\text{max}}$.  Let $\hat{\theta}_{T}$ be any estimator such that  $\frac{1}{T}\sum_{t=1}^{T}\hat{V}_{t}^{-1/2} s_{t,\hat{\theta}_{T}} = o_{p}(1/\sqrt{T})$. 
Then under Assumptions \ref{assumption:potential_outcomes}-\ref{assumption:bracket_complexity},
$$\frac{1}{\sqrt{T}} \sum_{t=1}^{T}\hat{V}_{t}^{-1/2} \dot{s}_{t,\hat{\theta}_{T}} \left(\hat{\theta}_{T}- \theta^{\star} \right) \overset{d}{\to} N(0,I_{d}).$$
\end{theorem}
Note that letting $\frac{1}{T}\sum_{t=1}^{T}\hat{V}_{t}^{-1/2} s_{t,\hat{\theta}_{T}} = o_{p}(1/\sqrt{T})$ instead of exactly zero ensures that the results are still valid when using an approximate root finding algorithm instead of an exact root. As an illustrative example, we first demonstrate the application of \cref{thm:clt} in the case of a linear model. 

\subsection*{An illustration: correctly specified models}

Let us first consider the case of correct specification, where $Y_{t} = Z_{t}^{T}\theta^{*}  +\epsilon_{t}$ and $\epsilon_{t}$ are independent Gauassian noise with $\epsilon_{t} \sim N(0,\sigma_{t}^{2})$. We assume $Z_{t} \in \sigma(\mathcal{H}_{t-1})$ is dependent on user choices based on a shared history $\mathcal{H}_{t-1} := \{ Z_{1},Y_{1},...,Z_{t-1},Y_{t-1}\}$. In our framework, we can think of $Z_{t} = \phi(A_t)$ as some deterministic function of the observed actions (e.g., one-hot encoding). 

Letting $\Sigma_{T} := \sum_{t=1}^{T} Z_{t}Z_{t}^{T}$, $\lambda_{\text{min}}(\Sigma_{T})$ represent the minimum eigenvalue of $\Sigma_T$, and $\lambda_{\text{max}}(\Sigma_{T})$ represent the maximum eigenvalue of $\Sigma_T$, \cite{1982laiwei} show that a sufficient condition for the OLS estimator $\tilde{\theta}_{T}:= \Sigma_{T}^{-1} \sum_{t=1}^{T} Z_{t}Y_{t} $ to converge to $\theta$ almost surely is that $\lambda_{\text{min}}(\Sigma_{T}) \overset{\text{a.s.}}{\to} \infty $ and $\frac{\log\lambda_{\text{max}} (\Sigma_{T})}{ \lambda_{\text{min}}(\Sigma_{T})} \overset{\text{a.s.}}{\to} 0 $. However, the requirement for asymptotic normality is much stronger. They require that a deterministic series of matrices $\{B_{T} \}_{T=1}^{\infty}$ exist such that $B_{T}^{-1} (\sum_{t=1}^{T} Z_{t}Z_{t}^{T})^{1/2} \overset{p}{\to} I_{p}$. 

Suppose that we apply \cref{thm:clt} in the case where $\tilde{\theta}_{T}$ is consistent but the asymptotic normality condition does not apply. We can still form a confidence interval for $\theta^*$ as follows. For simplicity, we assume that $f_{t} = c$ for some constant $c$ so the score function corresponds to the ordinary least squares loss without augmented components. In this case, we have $s_{t,\theta^{\star}} = - 2w_{t}(Y_{t} - Z_{t}^{T}\theta^{*} )Z_{t}$ with $w_{t} =\frac{\pi_{e}(A_{t})}{ \mathbb{P} (A_{t}  | \mathcal{H}_{t-1}) }$. Moreover, $$V_{t,\theta^{\star}} = \E\left[s_{t,\theta^{\star}}s_{t,\theta^{\star}}^{T} | \mathcal{H}_{t-1} \right] =\Var\left(Z_{t}^{T}(Y_{t} - \theta^{*}Z_{t}^{T})|\mathcal{H}_{t-1} \right) = \sigma_{t}^{2} \E \left[w_{t}^{2}Z_{t}Z_{t}^{T} |\mathcal{H}_{t-1} \right],$$ noting that $\text{Var} \left(Y_{t} | Z_{t} \right) =  \text{Var} \left(\epsilon_{t} |\mathcal{H}_{t-1}\right) =\sigma_{t}^{2}$ and $\dot{s}_{t,\theta} = w_{t}Z_{t}Z_{t}^{T}$.

Note that we can obtain a closed form solution by solving for $\hat{\theta}_{T}$ in the equation $$0=\frac{1}{T}\sum_{t=1}^{T} V_{t,\theta^{\star}}^{-1/2} s_{t,\hat{\theta}} =\frac{1}{T}\sum_{t=1}^{T}-\left(w_{t}^{2}\sigma_{t}^{2} \E \left[Z_{t}Z_{t}^{T} |\mathcal{H}_{t-1} \right]\right)^{-1/2}w_{t}(Y_{t} - Z_{t}^{T}\theta^{*} )Z_{t}.$$ Rearranging terms yields
$$\hat{\theta}_{T}= \left(\sum_{t=1}^{T} \sigma_{t}^{-1}\E \left[w_{t}^{2}Z_{t}Z_{t}^{T} |\mathcal{H}_{t-1} \right]^{-1/2} Z_{t}Z_{t}^{T}\right)^{-1} \left(\sum_{t=1}^{T}\sigma_{t}^{-1}w_t^{-1}\E \left[w_{t}^{2}Z_{t}Z_{t}^{T} |\mathcal{H}_{t-1} \right]^{-1/2} Z_{t}Y_{t}\right).$$
The formation of confidence intervals using \cref{thm:clt} yields, 
$$\frac{1}{\sqrt{T}}\sum_{t=1}^{T} w_{t} \sigma_{t}^{-1} \E \left[w_{t}^{2}Z_{t}Z_{t}^{T} |\mathcal{H}_{t-1} \right]^{-1/2} Z_{t}Z_{t}^{T} \left(  \hat{\theta}_{T}- \theta^{\star}\right) \overset{d}{\to} N(0,I_{p}).$$
In the case where the model is specified correctly, $\hat{\theta}_{T}$ is a consistent estimate of $\theta^{\star}$ for any choices of $w_{t}$ that follow \cref{assumption:bounded_importance}. Therefore, the choice of $w_{t}$ becomes strictly a question of asymptotic efficiency. We can illustrate this with a few special cases
\begin{itemize}
  \item Suppose $\E \left[Z_{t}Z_{t}^{T} |\mathcal{H}_{t-1} \right]$ is independent of $\mathcal{H}_{t-1}$, then $\frac{1}{T}\sum_{t=1}^{T} Z_{t}Z_{t}^{T}$ becomes a consistent estimate for $\E \left[Z_{t}Z_{t}^{T}\right]$. 
  \begin{itemize}
      \item When $\sigma_{t}$ is constant across $t$, it is optimal to choose $w_{t} = 1$. This reduces to the ordinary least squares solution (i.e., $\tilde{\theta}_{T}= \hat{\theta}_{T}$) and the confidence intervals become standard.
      \item  When $\sigma_{t}$ is not constant, it is optimal to choose $w_{t} \propto \frac{1}{\sigma_t^{2}}$. This reduces to the case of weighted least squares, and the confidence intervals again become standard. 
  \end{itemize}
  \item When each component of $Z_{t}$ corresponds to an indicator variable (i.e. $Z_{t,j} = 1_{A_t =j}$), then $\E \left[w_{t}^{2}Z_{t}Z_{t}^{T} |\mathcal{H}_{t-1} \right] = \diag\left(w_{t}^{2}\mathbb{P}\left(Z_{t,j} = 1 | \mathcal{H}_{t-1} \right)\right)$. In this case, the CLT simplifies to:
  $$\frac{1}{\sqrt{T}}\sum_{t=1}^{T} \sigma_{t}^{-1/2}\text{diag}\left(\{\mathbb{P}(A_{t} = a|\mathcal{H}_{t=1})^{-1/2}\}_{a \in \mathcal{A}}\right)\left(  \hat{\theta}_{T}- \theta^{\star}\right)  \overset{d}{\to} N(0,I_{p}).$$

\end{itemize}

Of course, this example is unrealistic because it assumes oracle knowledge of $V_{t,\theta^{\star}} = \E[w_{t}^{2}Z_{t}Z_{t}^{T}]$. In practice, the key challenge in applying \cref{thm:clt} is to find a reliable method to estimate this quantity. In the next section, we tackle the problem of finding practical methods to construct a sequence of estimators $\{\hat{V}_{t}\}_{t=1}^{T}$ that  can be used as plug-ins.  



\section{Practical Strategies for Covariance Estimation} \label{sec:covariance_estimation}
Our strategy for building an estimator of the variance is to decompose $\Var(s_{t,\theta^{\star}} \mid \mathcal{H}_{t-1})$ into terms that are either independent of the history or known to the experimenter and then find empirical estimators of these quantities from external data, similar in spirit to \cite{kato2020standardized}.

\begin{proposition} \label{prop:variance_decomp}
Recalling definitions for $\mathcal{H}_{t-1}$, $V_{t,\theta^{\star}}$, and $\theta^{\star}$ defined previously,
\begin{align*}
\Var(s_{t,\theta^{\star}} \mid \mathcal{H}_{t-1})
&=\Var\Big(\E_{A_{t} \sim \pi_{e}}  \left[\dot{m}_{\theta^{\star}}(X_t,A_{t},Y_t) \mid X_t\right]\Big) + \\
&\qquad \E \Bigg[\sum_{a=1}^{K} \frac{\pi_e(a \mid X_t)^2}{\pi_t(a|X_t)} 
\E[\dot{m}_{\theta^{\star}}(X_t,a,Y_t(a))\dot{m}_{\theta^{\star}}(X_t,a,Y_t(a))^{T} \mid X_t] \mid \mathcal{H}_{t-1}\Bigg] - \\
 & \qquad \E\left[\E_{A_{t} \sim \pi_{e}}  \left[\dot{m}_{\theta^{\star}}(X_t,A_{t},Y_t) \mid X_t\right]\E_{A_{t} \sim \pi_{e}}  \left[\dot{m}_{\theta^{\star}}(X_t,A_{t},Y_t) \mid X_t\right]^{T} \right].
\end{align*}
\end{proposition}

This decomposition allows the dependence on $\mathcal{H}_{t-1}$ to be explicitly decomposed into pieces that are known conditional on $\mathcal{H}_{t-1}$ (action-selection probabilities) and terms that are independent of history. It is still challenging to use this decomposition without knowing the potential outcome distribution $(X_{t}, Y_{t}(a))$. To bridge this gap, we assume that we can sample freely from the marginal distribution of $X_{t}$ and then estimate the first and second moments of the conditional distribution $Y_{t}(a) | X_{t}$ using flexible machine learning approaches. In many applications, we may have access to an independent data set of features that can be used for this purpose. However, if this is unavailable, we can also use sample splitting to use the features from the experiment itself. We briefly discuss both of these strategies below.

\subsection{Using external data} \label{subsec:external}
We first assume access to an external dataset composed of observations of $X$ independent of the history but with the same marginal distribution as $X_{t}$. 

\begin{assumption} \label{assumption:external}
There exists an external data set $\tilde{X} := \{ \tilde{X}_{i}\}_{i=1}^{n}$, independent of $\mathcal{H}_{t}$ for all $t \in T$, where $\tilde{X}_{i} \overset{\mathrm{iid}}{\sim} p(x)$ and $n = [rT]$ for some fixed $r \in (0,\infty)$. 
\end{assumption}

We also assume access to models targeting the conditional means and variance that can be used to create a plug-in estimate of $\Var \left(s_{t,\theta^{\star}} |\Hist_{t-1} \right)$.

\begin{assumption} \label{assumption:converges_var} We can construct functions $g_{t} : \mathbb{R}^{p} \times \mathcal{A} \to \mathbb{R}^{d}$ and $h_{t} : \mathbb{R}^{p}\times \mathcal{A}  \to \mathbb{R}^{d \times d}$ adapted to $\sigma(\mathcal{H}_{t-1})$ such that 
$g_{t}(X_{t},a) - \E\left[\dot{m}_{\theta^{\star}}(X_t,A_t,Y_t) |X_{t},A_t=a \right]\overset{p}{\to} 0$ for all $a \in \mathcal{A}$ and $$h_{t}(X_{t},a) - \E \left[  \dot{m}_{\theta^{\star}}(X_t,A_t,Y_{t} )\dot{m}_{\theta^{\star}}(X_t,A_t,Y_{t} )^{T} | X_{t},A_t=a \right] \overset{p}{\to} 0$$ for all $a \in \mathcal{A}$. 
\end{assumption}
We postpone discussion of how to construct the functions $g_t$ and $h_t$ to \cref{prop:glm}. In this section, we assume the existence of these quantities and use them to construct an estimate for $\Var \left(s_{t,\theta^{\star}} |\Hist_{t-1} \right)$. 
\begin{proposition} \label{prop:estimator}
Define the functions
\begin{equation} \label{equation:nu_t}
  \hat{\nu}_{t}(X_i) = \sum_{a=1}^{K} \pi_{e}(A_{t} = a \mid X_{t}) g_{t}(a,X_i),
\end{equation}
and $\bar{\nu} = \frac{1}{n} \sum_{X_{i} \in \tilde{X}} \hat{\nu}_{t}(X_i)$. Define: 
\begin{align*}
\hat{V}_{t} &= \frac{1}{n - 1} 
\Biggl( \sum_{X_i \in \mathcal{\tilde{X}}} \hat{\nu}_t(X_i)
- \bar{\nu} \Biggr)
\Biggl( \sum_{X_i \in \mathcal{\tilde{X}}} \hat{\nu}_t(X_i) 
- \bar{\nu} \Biggr)^{T} \\
&\quad + \frac{1}{n} \sum_{X_i \in \mathcal{\tilde{X}}} \sum_{a=1}^{K} 
\frac{\pi_e(A_t = a \mid X_i)^2 \, h_t(X_i,a)}{\mathbb{P}(A_t = a \mid X_i, \mathcal{H}_{t-1})} 
- \frac{1}{n} \sum_{X_i \in \mathcal{\tilde{X}}}  \hat{\nu}_{t}(X_i)\hat{\nu}_{t}(X_i)^{T}.
\end{align*}
Then, under Assumptions \ref{assumption:potential_outcomes}-\ref{assumption:converges_var}, 
$\norm{\hat{V}_{t} - V_{t,\theta^{\star}}}_{\text{op}} \overset{p}{\to} 0$.
\end{proposition}

Of course, it is far from clear how to construct $g_t$ and $h_t$ so that \cref{assumption:converges_var} is satisfied. In the remainder of this section, we provide a construction in the case of certain classes of objective functions $m_{\theta}$ and demonstrate that generalized linear models are a special case of this. 

\begin{proposition} \label{prop:glm}
Assume you have access to an estimator $\bar{\theta}_{T} \in \sigma(\mathcal{H}_{T})$ such that $ \bar{\theta}_{T} \overset{p}{\to} \theta^{\star}$. Assume that you can construct a sequence of continuous functions $f_{t} : \mathbb{R} \times \mathcal{A} \to \mathbb{R}$ and $j_{t} : \mathbb{R}\times \mathcal{A}  \to \mathbb{R}$ adapted to $\sigma(\mathcal{H}_{t-1})$ such that $f_{t}(X_{t},a) - \E\left[Y_t |X_{t},A_t = a \right]\overset{p}{\to} 0$  and  $j_{t}(X_{t},a) - \text{Var}\left[Y_t |X_{t},A_t = a \right]\overset{p}{\to} 0$. Furthermore, assume that $\dot{m}_{\theta}$ is linear in $y$ such that it can be decomposed as  $\dot{m}_{\theta}(x,a,y) = z_{\theta}(x,a)y + \nu_{\theta}(x,a)$ for some continuous functions $z,\nu$. Then, defining
$$g_{t}(x,a) = z_{\bar{\theta}_{T}}(x,a) f_t(x,a) + \nu_{\bar{\theta}_{T}}(x,a) $$
$$h_{t}(x,a) = z_{\bar{\theta}_{T}}(x,a)z_{\bar{\theta}_{T}}(x,a)^{T} j_t(x,a),$$
will satisfy the conditions of \cref{assumption:converges_var}.
\end{proposition}
Although this may seem restrictive, GLMs are a notable class of models that have a score function that satisfies these conditions. In our settings, GLMs will satisfy a first order condition of the form 
$$\sum_{t=1}^{T}(Y_{t} - \psi(\theta^{T}z(X_t,A_t)))z(X_t,A_t) = 0 ,$$
where $z : \R^{d} \times |\mathcal{A}| \rightarrow \R^{p}$ denotes some deterministic transformations of $(X_{t},A_{t})$ (e.g., one-hot encoding of $A_{t}$, interaction terms between $A_t$ and $X_t$) and  $\psi(\cdot)$ is the inverse-link function that maps $\theta^{T}z(X_t,A_t)$ to the mean under the (possibly misspecified) working model. In this case, we have:
\begin{align*}
    &g_{t}(X_t,a ) = \left(f_{t}(X_{t},a) - \psi(\theta^{T}z(X_t,A_t))\right)z(X_t,A_t) = f_{t}(X_{t},a)z(X_t,A_t) - \psi(\theta^{T}z(X_t,A_t))z(X_t,A_t),\\
   & h_{t}(X_t,a) = z(X_t,a)z(X_t,a)^{T} j_t(X_t,a) .
\end{align*}

\cref{prop:glm} simplifies the variance estimation task, as it converts the problem of estimation of potentially very high-dimensional covariance matrices into a simpler question of estimating conditional means and variances of $Y_{t}$. Note that $f_t$ can be the same model that was used in the definition of the MAIPWM estimator (\cref{eqn:hat_theta}) so this procedure only requires the management of an additional second model targeting $\E\left[Y_t^{2} |X_{t},A_t = a \right]$. In the case where $m_{\theta}$ is \emph{not} linear in $y$, then Jensen's inequality implies that a straightforward plug-in estimate is not available since $\E\left[\dot{m}_{\theta}(X_t,A_t,Y_t) \right] \ne \dot{m}_{\theta}(X_t,A_t,\E[Y_t])$. In these cases, more complicated methods of density estimation will be required. 

Some care must also be taken when choosing $\bar{\theta}_{T}$ in \cref{prop:glm}. Although $\hat{\theta}_{T}$ is consistent, it is defined as an estimator that satisfies $\frac{1}{T} \sum_{t=1}^{T} \hat{V}_{t}^{-1/2}s_{t,\hat{\theta}} = o_{p}(1)$ and therefore assumes that $\hat{V}_{t}$ has already been constructed in its definition, so it is not available when we need to construct $\hat{V}_{t}$. Fortunately, $\tilde{\theta}_{T}$ can be shown to be consistent.
\begin{proposition} \label{prop:consistency_tilde}
Let $\tilde{\theta}_{T}$ be defined as in \cref{eqn:hat_theta}. Under the conditions of \cref{thm:clt}, $\norm{\tilde{\theta}_{T}-  \theta^{\star}}_{1} = o_{p}(1)$.
\end{proposition}
This results in a two-step procedure in which $\tilde{\theta}_{T}$ is first estimated and used to construct $\{\hat{V}_{t}\}_{t=1}^{T}$. We then use these matrices to construct a corrected estimator $\hat{\theta}_{T}$ that allows us to form confidence sets. We summarize the end-to-end procedure for producing confidence sets under adaptive sampling for GLMs in \cref{alg:glm_confint}. 

\begin{algorithm}[H] 
\caption{Construction of Confidence Intervals for GLMs Under Adaptive Sampling}
\begin{algorithmic}[1]
\Require Data $\{(X_t, A_t, Y_t)\}_{t=1}^T$; evaluation dataset $\tilde{X} := \{\tilde{X}_{i}\}_{i=1}^{n}$; target policy $\pi_e(a \mid X_t)$; model class $\{m_\theta : \theta \in \Theta\}$ for the GLM with transformation function $z(x,a)$ and link function $\psi(\cdot)$.

\For{each time step $t$}
    \State Construct predictive models $f_t, j_t$ using $\mathcal{H}_{t-1}$ targeting the conditional mean and variances:
    \[
        f_t(X_i,a) \approx E[Y_i \mid X_i, A_i = a], \quad j_t(X_i,a) \approx \text{Var}[Y_{i} \mid X_i, A_i = a].
    \]
\EndFor
\State $\tilde{\theta}_{T} \leftarrow \argmax_{\theta \in \Theta}  \sum_{t=1}^{T} \sum_{a=1}^{K} \pi_{e}(a \mid X_{t}) \left( m_{\theta}(a,X_{t},f_{t}(a,X_{t}))   + \mathbbm{1}_{\{A_{t} = a\}} \frac{ m_{\theta}(A_{t},X_{t},Y_{t})  -  m_{\theta}(X_{t},A_{t},f_{t}(X_{t},A_{t})) }{ \mathbb{P}(A_{t} =a \mid X_{t},\mathcal{H}_{t-1})  }\right)$.
\State Compute 
\begin{align*}
    g_{t}(x,a) &\leftarrow f_{t}(x,a)z(x,a) - \psi(\tilde{\theta}^{T}_ {T}z(x,a))z(x,a),\\
   h_{t}(x,a) &\leftarrow z(x,a)z(x,a)^{T}j_t(x,a)  .
\end{align*}
\State Use $\tilde{X}$ and \cref{prop:estimator} to construct $\hat{V}_{t}$ for every $t \in [T]$. 
\State Use $\{\hat{V}_{t}\}_{t=1}^{T}$ and apply \cref{thm:clt} to construct $\hat{\theta}_{T}$ and corresponding confidence intervals. 

\State \textbf{Output:} Estimated parameter $\hat{\theta}_{T}$ and asymptotic $(1-\alpha)$ confidence region
\[
    \mathcal{C}_{\theta^{\star}} = \left\{ \theta \in \Theta : 
    \left(\frac{1}{\sqrt{T}} \sum_{t=1}^{T}\hat{V}_{t}^{-1/2} \dot{s}_{t,\hat{\theta}_{T}} (\hat{\theta}_{T}- \theta ) \right)^{\!\top}
    \left(\frac{1}{\sqrt{T}} \sum_{t=1}^{T}\hat{V}_{t}^{-1/2} \dot{s}_{t,\hat{\theta}_{T}} (\hat{\theta}_{T}- \theta ) \right)
    \le  \chi^{2}_{d,1-\alpha} \right\}.
\]
\end{algorithmic}\label{alg:glm_confint}
\end{algorithm}

\subsection{Using sequential sample splitting} \label{subsec:sequential}
Alternatively, when an external dataset is not available, we can use sample splitting to create an independent data set to estimate $\Var \left(s_{t,\theta^{\star}} |\Hist_{t-1} \right)$. At each $t$, we draw an independent random variable $\zeta_{t} \sim \text{Ber}(r)$ and then assign $X_t$ into one of two parallel histories $\mathcal{H}_{t}^{(1)}$ and $\mathcal{H}_{t}^{(2)}$: 
$$ X_{t} \in 
\begin{cases}
    \mathcal{H}_{t}^{(1)} & \text{when } \zeta_t = 0 \\
    \mathcal{H}_{t}^{(2)} & \text{when } \zeta_t = 1\\
\end{cases} \text{ and }
\pi_{t}(a |X_{t})  \in 
\begin{cases}
    \sigma(\mathcal{H}_{t-1}^{(1)},X_{t}) & \text{when } \zeta_t = 0 \\
    \sigma(\mathcal{H}_{t-1}^{(2)},X_{t})  & \text{when } \zeta_t = 1\\
\end{cases}.
$$
Proceeding in this fashion constructs two separate histories with the same distribution that are independent of each other. Note that for this particular application, we are only trying to sample from the marginal distribution of $X_{t}$ so the assignment of $A_{t}$ is only required for one of the two histories. An alternative experimental setup might only track the triples $(X_t,A_t,Y_t)$ for $\mathcal{H}_{t}^{(1)}$ and then refrain from treating or tracking outcomes for individuals in $\mathcal{H}_{t}^{(2)}$. In contexts where there is a cost to administering the treatment, the latter design will be preferable, but in other situations it may not be possible or desirable to refrain from treatment in which case the former is required.

Regardless of the particular approach, the only requirement is that $\zeta_{t}$ is independent of both histories and the predictive models are trained on only one of the tracked histories.  
\begin{assumption} \label{assumption:independence_zeta}
    $\zeta_{t} \sim \text{Ber}(p)$ for some fixed $p \in (0,1)$ and is independent of $\mathcal{H}_{t-1}^{(1)}$ and $\mathcal{H}_{t-1}^{(2)}$ for all $t\in [T]$.
\end{assumption}
\cref{assumption:independence_zeta} can be used in place of \cref{assumption:external} to construct an independent dataset of features when one is not available. 
\begin{assumption}
 \label{assumption:converges_var_sequential} There exist functions $g_{t} : \mathbb{R}^{p} \times \mathcal{A} \to \mathbb{R}$ and $h_{t} : \mathbb{R}^{p}\times \mathcal{A}  \to \mathbb{R}^{d \times d}$ adapted to $\sigma(\mathcal{H}_{t-1}^{(1)})$ such that 
$g_{t}(X_{t},a) - \E[Y_{t} | X_{t}, A_{t} = a] \overset{p}{\to} 0$ for all $a \in \mathcal{A}$ and $$h_{t}(X_{t},a) - \E \left[  \dot{m}_{\theta^{\star}}(X_t,A_t,Y_{t} )\dot{m}_{\theta^{\star}}(X_t,A_t,Y_{t} )^{T} | X_{t},A_t=a \right] \overset{p}{\to} 0$$ for all $a \in \mathcal{A}$. 
\end{assumption}
\cref{assumption:converges_var_sequential} is almost the same as \cref{assumption:converges_var} but we require that the models are only trained on $\mathcal{H}_{t-1}^{(1)}$ to ensure that $\mathcal{H}_{t-1}^{(2)}$ can be preserved as an independent data set for calculating $\hat{V}_{t}^{-1/2}$. At this point, the procedure can proceed the same as in \cref{subsec:external}. For completeness, we present a slightly modified version of \cref{prop:estimator} adapted to the case of sample splitting.

\begin{proposition} \label{prop:estimator_sequential}
Under Assumptions \ref{assumption:potential_outcomes}- \ref{assumption:bracket_complexity}, \ref{assumption:independence_zeta}, and \ref{assumption:converges_var_sequential}, define $\hat{\nu}_t(X_i)$ as in \cref{equation:nu_t} and  $n_p := \sum_{t=1}^{T} \zeta_t$. Define:
\begin{align*}
\hat{V}_{t} &= \frac{1}{n_p - 1} 
\Biggl( \sum_{X_i \in \mathcal{H}_{t}^{(2)}} \hat{\nu}_t(X_i) 
- \frac{1}{n_p} \sum_{X_i \in \mathcal{H}_{t}^{(2)}} \hat{\nu}_t(X_i) \Biggr)
\Biggl( \sum_{X_i \in \mathcal{H}_{t}^{(2)}} \hat{\nu}_t(X_i) 
- \frac{1}{n_p} \sum_{X_i \in \mathcal{H}_{t}^{(2)}} \hat{\nu}_t(X_i) \Biggr)^{T} + \\
&\qquad \frac{1}{n_p} \sum_{X_i \in \mathcal{H}_{t}^{(2)}} \sum_{a=1}^{K} 
\Biggl[ 
\frac{\pi_e(A_t = a \mid X_i)^2 \, h_t(X_i,a)}{\mathbb{P}(A_t = a \mid X_i, \mathcal{H}_{t-1}^{(2)})} 
\Biggr]-\frac{1}{n_p} \sum_{X_i \in \mathcal{H}_{t}^{(2)}}  \hat{\nu}_{t}(X_i)\hat{\nu}_{t}(X_i)^{T}.
\end{align*}
Then $\norm{\hat{V}_{t} - V_{t,\theta^{\star}}}_{\mathrm{op}} \overset{p}{\to} 0$.
\end{proposition}
\cref{alg:glm_confint} remains the same when used in this setting, with the only modification being that \cref{prop:estimator_sequential} should be used at the sixth step instead of \cref{prop:estimator}.

\section{Empirical Results} \label{sec:empirical_results}

We deploy these methods in the contextual bandit problem, where the goal is to choose a sequence of actions $\{A_{t}\}_{t=1}^{T}$ that minimize \emph{regret},
$\sum_{t=1}^{T} \left(\mu^{\star}_{t} - Y_{t}\right)$, where $\mu^{\star}_{t} := \max \{  \E[Y_{t}(a)] : a \in \mathcal{A} \}$ is the arm with the highest expected reward. This is a well-studied problem with a variety of proposed solutions; we focus on strategies for choosing $P(A_{t} = a| X_{t},\mathcal{H}_{t-1})$ that simultaneously leverage flexible predictive models $f_{t}$ and $j_{t}$ described in \cref{assumption:converges_var} targeting $\E[Y_{t}|X_{t},A_{t}]$ and $\E[Y_{t}^{2}|X_{t},A_{t}]$, respectively. 

\paragraph{Strategies for selecting $A_t$:}
\begin{itemize}
    \item A \textbf{uniform} strategy where $\mathbb{P}(A_{t} = a | X_{t}, \mathcal{H}_{t-1}) = \frac{1}{|\mathcal{A}|}$ for all $t$. Classical statistical guarantees will apply in this setting since the decision-making is stationary and non-adaptive.
    \item An \textbf{epsilon greedy} strategy that lets $$A_{t} = \begin{cases} \argmax_{a \in \mathcal{A}} f_{t}(X_{t},a) &\text{ with probability } 1- \epsilon  \\ \text{ any other } a \in \mathcal{A} &\text{ with probability } \frac{ \epsilon}{|\mathcal{A}| - 1}.
\end{cases}$$
    \item An \textbf{upper confidence bound (UCB)} strategy that constructs confidence intervals using $f_{t}$ to define the centerpoint and $ j_{t}$ to control the width. Here, 
    $$A_{t} = \argmax_{a\in \mathcal{A}} \left[ f_{t}(X_{t},a) - q_{\alpha/2}\sqrt{j_{t}(X_{t},a)}\right],$$ where $q_{\alpha/2}$ denotes the $\alpha/2$ quantile of a standard normal distribution. 
    \item An approach inspired by \textbf{Thompson Sampling} which draws from a distribution of $ Y_{t} | X_{t}, a$ for each $a$ to form an estimate $\hat{Y}_{t}(a)$ and then lets $A_{t} := \argmax \hat{Y}_{t}(a)$. We apply this algorithm in the current setting using the working assumption that $Y_{t}(a) \sim N(f_{t}(X_t,a), j_{t}(X_t,a) - f_{t}^{2}(X_t,a)) $. We also consider versions of this method with a clipping constraint so all probabilities are constrained to lie in $[0.05,0.95]$. 
\end{itemize}
We use a real data example to test the confidence intervals constructed using \cref{thm:clt} under the above sampling strategies. The data set is sourced from the Osteoarthritis Initiative (OAI) \citep{nevitt2006osteoarthritis}, a publicly available data set provided by the NIH. The OAI is a ten-year longitudinal observational study of men and women affected by osteoarthritis. The study collects several baseline measurements about each participant's knee health, such as self-reported measurements of pain, disability status, and flexion contracture, in addition to demographic measurements such as age, BMI, and sex. After enrollment in the study, measurements of knee health, such as Kellgren and Lawrence (KL) grade, are taken at yearly intervals. 

We consider the outcome ($Y_t$) to be the four year change in KL grade for the affected knee of an osteoarthritic patient. Although KL grade is a discrete numeric variable, we treat KL grade as a continuous variable for the purposes of this example. The chosen features ($X_{t}$) include the aforementioned baseline measurements of knee health and demographic variables as well as additional risk factors identified in \cite{DUNN20201020} such as self-reported quality of life scores and use of pain medication. Since OAI was an observational study to characterize disease progression rather than a trial to compare treatments, to evaluate methodology, we construct a semi-synthetic data set that applies hypothetical treatments ($A_{t}$) and then enforces a synthetic relationship between $A_{t}$ and $Y_{t}$ while preserving the conditional distribution of $Y_{t} | X_{t}$. This is done using machine learning methods to first learn the distribution of $Y_{t} | X_{t}$ and then to create synthetic outcomes that combine the output of the machine learning model with hypothetical treatment effects. For a detailed description of the semi-synthetic dataset creation process, see \cref{appendix:dataset_construction}. The end result is a dataset with similar complexity to the original longitudinal study but with the option to sample patients sequentially with ground-truth knowledge of the treatment effect and target parameters. Note that to ensure that this setting would align with the hardest case where treatment policies would not be likely to converge, we picked treatments so that there would be the same expected reward across multiple arms and there would be no unique optimal policy. 

\begin{figure*}  \label{fig:confint}
    \centering
    \begin{subfigure}[t]{\linewidth}
         \centering
         \includegraphics[width=\textwidth]{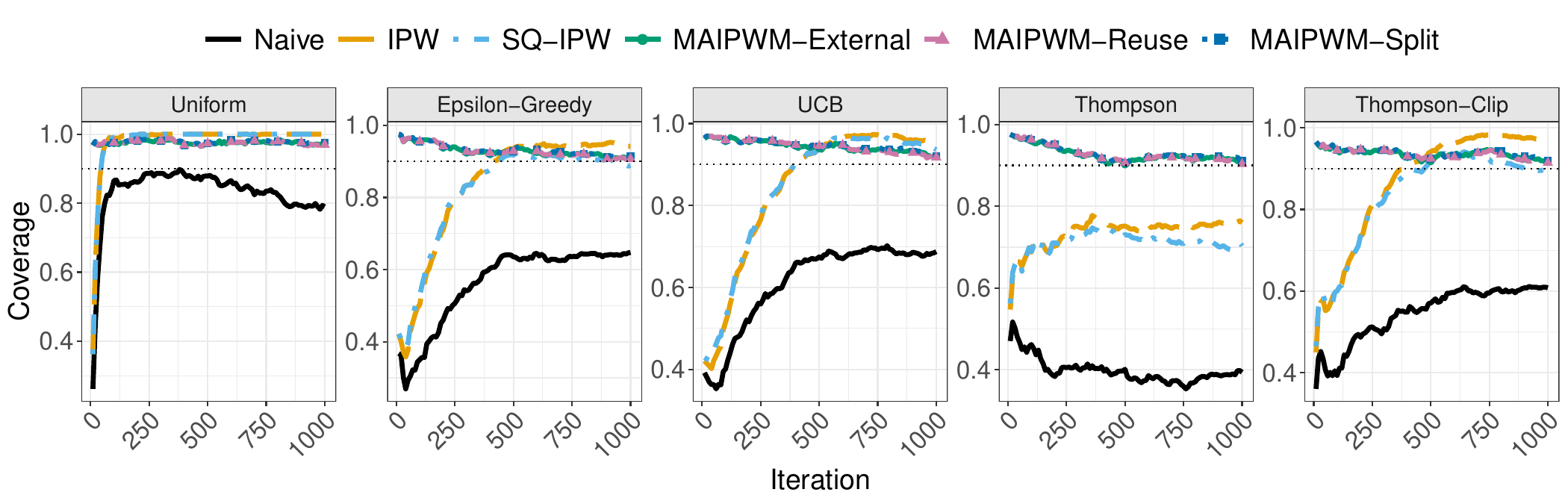}
        \caption{Coverage for target $1-\alpha=0.9$}
     \end{subfigure}
    \begin{subfigure}[t]{\linewidth}
         \centering
         \includegraphics[width=\textwidth]{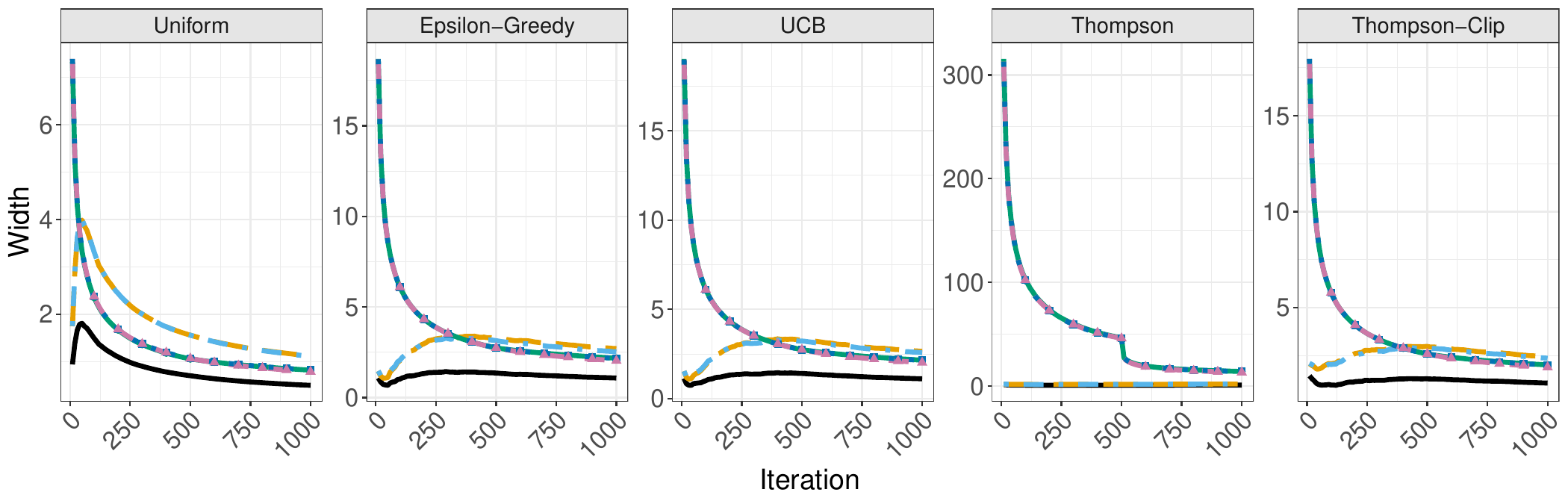}
        \caption{Confidence interval width for confidence intervals shown above}
     \end{subfigure}
\caption{Confidence intervals constructed using \cref{thm:clt} using ML-based estimates of the variance cover in all scenarios. GLMs using naive inverse propensity weighting often undercover, especially in situations where the assignment probabilities vary substantially over time. We note that both sample splitting and external data reuse for covariance estimation are valid, but sample splitting has significantly wider confidence intervals. Reusing the same data for variance estimation and parameter estimation performs similarly to using external data, though we lack theoretical guarantees for this method.} 
    \label{fig:simulation_results}
\end{figure*}

In simulations, patients come online one at a time and $\hat{\theta}_{t}$ is recomputed for each $t$ with \cref{thm:clt} used to form confidence intervals. As discussed in \cref{sec:covariance_estimation}, online model-based estimates of the conditional mean and variance are trained sequentially (using random forests) as inputs into the MAIPWM estimator. To increase computational efficiency, we choose to re-train the predictive models after every $100$ data points, though in principle they could be re-trained after every sample. For continuous outcomes, we consider a linear model with $m_{\theta} :=- (Y_{t} - \sum_{a\in \mathcal{A}}\theta_{a} \mathbbm{1}_{A_t = a})^{2}$. 

We track the length of the confidence intervals and their nominal coverage rate (given target coverage of $0.9$) in \cref{fig:simulation_results}. The methods we consider are:
\paragraph{Methods for confidence interval construction:}
\begin{itemize}
    \item \textbf{Naive} Use off-the shelf maximum likelihood estimation with equal weighting of outcomes and sandwich estimates of the variance. 
    \item \textbf{IPW} Uses the methodology of \cite{guo2025statisticalinferencemisspecifiedcontextual} to construct confidence intervals, which relies on inverse propensity weighting of the score function and sandwich estimates of the variance. 
    \item \textbf{SQ-IPW} Use the methodology of \cite{zhang2021mestimators} to construct confidence intervals, which weight observations by the \emph{square root} of the propensity scores, setting $w_{t} = \sqrt{\frac{\pi_{e}(A_t|X_t) }{P(A_t|X_t,\mathcal{H}_{t-1}) }}$ and using sandwich estimates of the variances. 
    \item \textbf{MAIPWM--External} Using the method described in \cref{thm:clt} with plug-in estimates $V_{t,\theta}$ constructed using external data as described in \cref{subsec:external}.
    \item \textbf{MAIPWM--Sample-Splitting} Using the method described in \cref{thm:clt} with plug-in estimates $V_{t,\theta}$ constructed using sequential sample splitting as described in \cref{subsec:sequential}.
    \item \textbf{MAIPWM--Reuse}  Using the method described in \cref{thm:clt} with plug-in estimates $V_{t,\theta^{\star}}$ constructed by naively reusing all $X_{i} \in \mathcal{H}_{t-1}$ at time step $t$.
\end{itemize}
We tested the MAIPWM estimator in all combinations of strategies used to estimate $V_{t,\theta^{\star}}$ and select actions. The empirical results shown in \cref{fig:simulation_results} align with the theory --- naive maximum likelihood estimates undercover in all situations other than the uniform (nonadaptive case). IPW and SQ-IPW methods require a significantly larger number of samples before they achieve the nominal coverage rate and do not cover at all in the case of Thompson sampling. The MAIPW estimator using sample splitting and external data both cover correctly, with the sample splitting method paying a price in terms of notably wider confidence intervals. Although we have no theoretical results supporting MAIPWM-Reuse, which reuses the historical data to estimate the covariance, this method has nearly identical results to the two MAIPWM methods that use external samples. Future work can potentially justify this formally. 

\section{Conclusion} \label{sec:conclusion}
We present a method for inference in $M$ estimation problems under a potentially misspecified working model. The method requires the estimation of time-varying covariance matrices, for which we provide an algorithm for generalized linear models using flexible machine learning approaches. Empirical results demonstrate that the method has correct nominal coverage, while existing approaches often undercover in situations where action-selection probabilities vary considerably and the variance of the score function does not converge. Our work suggests several potential avenues for follow-up work, which we detail below:

\paragraph{Estimation of action-selection probabilities} Throughout, we assume that the probabilities of each action being selected in the experiment are determined by the experimenter and known, which sometimes cannot be guaranteed in practical settings. Weakening this assumption first requires modifying the asymptotic arguments for \cref{thm:clt} and ensuring the rate of estimation of these probabilities to converge sufficiently fast. It also complicates the estimation of the time-varying variance as explained in \cref{sec:covariance_estimation} since \cref{prop:variance_decomp} crucially relies on knowledge of the action-selection probabilities to decompose the variance into a tractable form for plug-in estimation.

\paragraph{Data reuse}
In estimating the variance of the score function, we either assume access to an external dataset of features with the same marginal distribution as in the sequential experiment or are forced to split the data (and consequently sacrifice power). Empirical results suggest that simply reusing the collected sequential data for both variance estimation and parameter estimation provides nearly identical results compared to when an external dataset is available, but we have not provided a rigorous guarantee.

\paragraph{Efficiency of estimators} Although we have demonstrated that the proposed estimator results in valid confidence intervals, we have not shown it is optimal in any sense or demonstrated results about semi-parametric efficiency, which is a logical next step for future work.

%% file: appendix.tex
\section{Preliminaries }
We introduce several key facts about martingales that are used throughout. For a more thorough discussion, consult \cite{hall2014martingale}.

\begin{fact}[Martingale Weak Law of Large Numbers, \cite{hall2014martingale}]  \label{fact:martingale_wlln}
 
Let $\left\{S_n=\sum_{i=1}^n X_i, \mathcal{H}_{n}, n \geq 1\right\}$ be a martingale with respect to a filtration $\mathcal{H}_{n}$ and $\left\{b_n\right\}$ a sequence of positive constants with $b_n \rightarrow \infty$ as $n \rightarrow \infty$. Then, writing $X_{n i}=X_i \mathbbm{1}\left[\left|X_i\right| \leq b_n\right]$, $1 \leq i \leq n$, we have that $b_n^{-1} S_n \xrightarrow{\mathrm{p}} 0$ as $n \rightarrow \infty$ if
\begin{enumerate}
    \item $\sum_{i=1}^n P\left(\left|X_i\right|>b_n\right) \rightarrow 0$
    \item $b_n^{-1} \sum_{i=1}^n \mathbb{E}\left[X_{n i} \mid \mathcal{H}_{n-1}\right] \xrightarrow{\mathrm{p}} 0$
    \item $b_n^{-2} \sum_{i=1}^n\left\{\mathbb{E}\left[X_{n i}^2\right]-\mathbb{E}\left[\mathbb{E}\left[X_{n i} \mid \mathcal{H}_{n-1}\right]\right]^2\right\} \rightarrow 0$
\end{enumerate}
\end{fact}

\begin{remark} A sufficient condition for \cref{fact:martingale_wlln} to hold is when $X_i$ is bounded by a constant.
\end{remark}

\begin{remark} \label{remark:wlln_integrable} A sufficient condition for \cref{fact:martingale_wlln} to hold is when $X_{i}$ is square integrable and $\frac{1}{T^{2}} \sum_{t=1}^{T} \E[X_{i}^{2}] \rightarrow 0$.
\end{remark}

\section{Deferred Proofs}

\subsection{Proof of \cref{thm:clt}}
The overall structure of the proof is similar to \cite{zhang2021mestimators}, which in turn is based on classical proofs of the asymptotic normality of $M$-estimators such as the ones described in \cite{van2000asymptotic}. The major differences in the approach are that we now need to deal with the variance stabilizing matrices $\hat{V}_{t}$ and bound the complexity of the predictive model $f_t$. 

The proof relies on two Lemmas:
\begin{enumerate}
    \item \cref{lemma:score_clt} uses the first order conditions to show that an appropriately rescaled version of the score function $\frac{1}{\sqrt{T}}\sum_{t=1}^{T} \hat{V}_{t}^{-1/2} s_{t,\theta^{\star}}$ is asymptotically normal using a Martingale Central Limit Theorem \citep{hall2014martingale}. 
    \item \cref{lemma:consistency} shows that $\norm{\hat{\theta}_{T}- \theta^{\star}}_{1}= o_{p}(1)$.
\end{enumerate}
Given these two lemma, we can now use Taylor expansion around the score function to prove \cref{thm:clt}.

First, by Taylor's theorem we can write for each $t$ and some $\tilde{\theta}_{t}$ on the line segment between $\hat{\theta}_{T}$ and $\theta^{\star}$ that
\begin{align*}
s_{t,\theta^{\star}} &= s_{t,\hat{\theta}_{T}}  + \dot{s}_{t,\hat{\theta}_{T}}(\theta^{\star} - \hat{\theta}_{T}) + \frac{1}{2} \left( \theta^{\star} - \hat{\theta}_{T}\right)^{T} \ddot{s}_{t,\tilde{\theta}_{t}} \left( \theta^{\star} -\hat{\theta}_{T}\right).
\end{align*}
For each $t$, we multiply by $\frac{1}{\sqrt{T}} \hat{V}_{t}^{-1/2}$ and sum over $t \in [T]$ to arrive at 
\begin{align*}
    -\frac{1}{\sqrt{T}} \sum_{t=1}^{T} \hat{V}_{t}^{-1/2} s_{t,\theta^{\star}} &= 
    \frac{1}{\sqrt{T}} \sum_{t=1}^{T} \hat{V}_{t}^{-1/2} \left(  s_{t,\hat{\theta}_{T}} - \dot{s}_{t,\hat{\theta}_{T}}(\hat{\theta}_{T}-\theta^{\star} )  -  \frac{1}{2}\left(  \hat{\theta}_{T}-\theta^{\star} \right)^{T} \ddot{s}_{t,\tilde{\theta}_{t}} \left( \hat{\theta}_{T}-\theta^{\star} \right) \right).
\end{align*}
Recall that by assumption $\frac{1}{T}\sum_{t=1}^{T} \hat{V}_{t}^{-1/2}s_{t,\hat{\theta}_{T}} = o_{p}(1/\sqrt{T})$, this simplifies the expression to 
\begin{align*}
 -\frac{1}{\sqrt{T}} \sum_{t=1}^{T} \hat{V}_{t}^{-1/2} s_{t,\theta^{\star}} &= o_{p}(1) + 
    \frac{1}{\sqrt{T}} \sum_{t=1}^{T} \hat{V}_{t}^{-1/2} \left(  \dot{s}_{t,\hat{\theta}_{T}}(\hat{\theta}_{T}-\theta^{\star} )  +  \frac{1}{2}\left(  \hat{\theta}_{T}-\theta^{\star} \right)^{T} \ddot{s}_{t,\tilde{\theta}_{t}} \left( \hat{\theta}_{T}-\theta^{\star} \right) \right)   \\
    &= o_{p}(1)
     + \frac{1}{\sqrt{T}}
     \Bigg(
         I_{d}
         + \frac{1}{2}
         \sum_{t=1}^{T}
         \hat{V}_{t}^{-1/2}
         (\hat{\theta}_{T}-\theta^{\star})^{T}
         \ddot{s}_{t,\tilde{\theta}_{t}}   \\
     &\qquad\qquad\qquad\qquad
         \times
         \Bigg(
            \sum_{t=1}^{T}
            \hat{V}_{t}^{-1/2}
            \dot{s}_{t,\hat{\theta}_{T}}
         \Bigg)^{-1}
     \Bigg)
     \sum_{t=1}^{T}
     \hat{V}_{t}^{-1/2}
     \dot{s}_{t,\hat{\theta}_{T}}
     (\hat{\theta}_{T}-\theta^{\star}) .
\end{align*}

Recall that  $ -\frac{1}{\sqrt{T}} \sum_{t=1}^{T} \hat{V}_{t}^{-1/2} s_{t,\theta^{\star}} \overset{d}{\to} N(0,I_d)$. By Slutsky's Theorem, it is therefore sufficient to to show that 
$$\norm{\sum_{t=1}^{T}  \hat{V}_{t}^{-1/2} \left(  \hat{\theta}_{T}-\theta^{\star} \right)^{T} \ddot{s}_{t,\tilde{\theta}_{t}} \left(\sum_{t=1}^{T} \hat{V}_{t}^{-1/2} \dot{s}_{t,\hat{\theta}_{T}}\right)^{-1}}_{2} = o_{p}(1),$$
in order to conclude that $\frac{1}{\sqrt{T}} \sum_{t=1}^{T} \hat{V}_{t}^{-1/2} \dot{s}_{t,\hat{\theta}_{T}} \left(\hat{\theta}_{T}-\theta^{\star} \right) \overset{d}{\to} N(0,I_d)$.
We write this as follows. 
\begin{align*}
\norm{\sum_{t=1}^{T}  \hat{V}_{t}^{-1/2} \left(  \hat{\theta}_{T}-\theta^{\star} \right)^{T} \ddot{s}_{t,\tilde{\theta}_{t}} \left(\sum_{t=1}^{T} \hat{V}_{t}^{-1/2} \dot{s}_{t,\hat{\theta}_{T}}\right)^{-1}}_{2} &= \norm{ \frac{1}{T}  \sum_{t=1}^{T}  \hat{V}_{t}^{-1/2} \left(  \hat{\theta}_{T}-\theta^{\star} \right)^{T} \ddot{s}_{t,\tilde{\theta}_{t}} \left( \frac{1}{T}  \sum_{t=1}^{T} \hat{V}_{t}^{-1/2} \dot{s}_{t,\hat{\theta}_{T}}\right)^{-1}}_{2} \\
&\le \norm{\frac{1}{T}  \sum_{t=1}^{T}  \hat{V}_{t}^{-1/2} \left(  \hat{\theta}_{T}-\theta^{\star} \right)^{T} \ddot{s}_{t,\tilde{\theta}_{t}}}_{2} \norm{\left( \frac{1}{T}  \sum_{t=1}^{T} \hat{V}_{t}^{-1/2} \dot{s}_{t,\hat{\theta}_{T}}\right)^{-1}}_{2} \\
&\le\frac{1}{T} \norm{\hat{\theta}_{T}-\theta^{\star}} \sum_{t=1}^{T} \norm{ \hat{V}_{t}^{-1/2}\ddot{s}_{t,\tilde{\theta}_{t}}}_{2} \norm{\left( \frac{1}{T}  \sum_{t=1}^{T} \hat{V}_{t}^{-1/2} \dot{s}_{t,\hat{\theta}_{T}}\right)^{-1}}_{2} \\
&\le\norm{\hat{\theta}_{T}-\theta^{\star}}_{2}\frac{1}{\delta_{\text{min}}}\frac{1}{T}  \sum_{t=1}^{T} \norm{\ddot{s}_{t,\tilde{\theta}_{t}}}_{2} \norm{\left( \frac{1}{T}  \sum_{t=1}^{T} \hat{V}_{t}^{-1/2} \dot{s}_{t,\hat{\theta}_{T}}\right)^{-1}}_{2}.
\end{align*}
By \cref{lemma:consistency}, $\norm{\hat{\theta}_{T}-\theta^{\star}}_{2} =o_{p}(1)$ and by \cref{lemma:invertible}, $\norm{\left( \frac{1}{T}  \sum_{t=1}^{T} \hat{V}_{t}^{-1/2} \dot{s}_{t,\hat{\theta}_{T}}\right)^{-1}}_{2}  = O_{p}(1)$ . Therefore, it is sufficient to show that the remaining terms are $O_p(1)$ to conclude that the entire term is $o_{p}(1)$. We know that $1/\delta_{\text{min}}$ is bounded because $\delta_{\text{min}} > 0$ by assumption. We tackle each of these terms individually (note that by equivalence of matrix norms, it is sufficient to show convergence in $1$-norm). 

\paragraph{Showing $\frac{1}{T}  \sum_{t=1}^{T} \norm{\ddot{s}_{t,\tilde{\theta}_{t}}}_{1} = O_{p}(1)$.}
We for some
\begin{align*}
\norm{\ddot{s}_{t,\tilde{\theta}_{t}}}_{1} &=  \norm{\sum_{a=1}^{K} \pi_{e}(A_{t} =a | X_{t}) \left( \dddot{m}_{\tilde{\theta}_{t}}(a,X_{t},f_{t}(a,X_{t}))   + \mathbbm{1}_{A_{t} = a} \frac{ \dddot{m}_{\tilde{\theta}_{t}}(X_{t},A_{t},Y_{t})  -  \dddot{m}_{\tilde{\theta}_{t}}(X_{t},A_{t},f_{t}(X_{t},A_{t})) }{ \mathbb{P} \left(A_{t} =a | X_{t},\mathcal{H}_{t-1} \right)  }\right)}_1 \\
&\le \frac{\pi_{e}(A_{t} | X_{t})}{\mathbb{P} \left(A_{t} | X_{t},\mathcal{H}_{t-1} \right) }
\norm{\dddot{m}_{\tilde{\theta}_{t}}(X_{t},A_{t},Y_{t})}_{1}+ \pi_{e}(A_{t} =a | X_{t})  \sum_{a=1}^{K} \left(1 - \frac{\mathbbm{1}_{A_{t} = a} }{\mathbb{P} \left(A_{t} =a | X_{t},\mathcal{H}_{t-1} \right)}\right) \norm{\dddot{m}_{\tilde{\theta}_{t}}(X_{t},a,f_{t}(X_{t},a))}_{1}\\
&\le \frac{\pi_{e}(A_{t} | X_{t})}{\mathbb{P} \left(A_{t} | X_{t},\mathcal{H}_{t-1} \right) }
u_{m}(X_{t},A_{t},Y_{t}) + \pi_{e}(A_{t} =a | X_{t})  \sum_{a=1}^{K} \left(1 - \frac{\mathbbm{1}_{A_{t} = a} }{\mathbb{P} \left(A_{t} =a | X_{t},\mathcal{H}_{t-1} \right)}\right) u_{\mathcal{F}}(X_{t},a,f_{t}(X_{t},a)),
\end{align*}
where the last inequalities follow from \cref{assumption:domination} and \cref{assumption:bracket_complexity} after noting that $\norm{\tilde{\theta}_{t} - \theta^{\star}} < \delta$ for any fixed $\delta$ since $\tilde{\theta}_{t}$ must lie on the line segment between $\hat{\theta}_{T}$ and $\theta^{\star}$. 
Therefore, 
\begin{align*}
\sum_{t=1}^{T} \norm{\ddot{s}_{t,\tilde{\theta}_{t}}}
&\le
\sum_{t=1}^{T}
\frac{\pi_{e}(A_{t} \mid X_{t})}
     {\mathbb{P}\!\left(A_{t} \mid X_{t},\mathcal{H}_{t-1}\right)}
\, u_{m}(X_{t},A_{t},Y_{t})  \\
&\quad
+ \sum_{t=1}^{T}
\pi_{e}(A_{t}=a \mid X_{t})
\sum_{a=1}^{K}
\Bigg(
    1
    - \frac{\mathbbm{1}_{A_{t}=a}}
           {\mathbb{P}\!\left(A_{t}=a \mid X_{t},\mathcal{H}_{t-1}\right)}
\Bigg)
\, u_{\mathcal{F}}\!\left(
    X_{t},a,f_{t}(X_{t},a)
\right).
\end{align*}

We will now show that both the first and second terms converge to their expected values by \cref{fact:martingale_wlln} to conclude the proof. For the first term, 
$ \E\left[\frac{\pi_{e}(A_{t} | X_{t})}{\mathbb{P} \left(A_{t} | X_{t},\mathcal{H}_{t-1} \right) }
u_{m}(X_{t},A_{t},Y_{t}) |\mathcal{H}_{t-1}\right] =  \E_{\mathcal{P},\pi_{e}}\left[u_{m}(X_t,A_t,Y_t)\right]$. On the other hand, letting $w_{t} = \frac{\pi_{e}(A_{t} | X_{t})}{\mathbb{P} \left(A_{t} | X_{t},\mathcal{H}_{t-1} \right) }$ implies that $w_{t} \le C_{1}$ by \cref{assumption:bounded_importance}. Therefore, 
\begin{align*}
    \mathbb{E}[w_{t}^{2} u_{m}(X_t,A_t,Y_t)^{2}] &\le \mathbb{E}[C_{1} w_{t} u_{m}(X_t,A_t,Y_t)^{2}] \\
    &= C_{1}\mathbb{E}[ w_{t} u_{m}(X_t,A_t,Y_t)^{2}]\\
    &= C_{1} \mathbb{E}[\mathbb{E}[ w_{t} u_{m}(X_t,A_t,Y_t)^{2} |\mathcal{H}_{t-1}]] \\
    &= C_{1} \mathbb{E}[\mathbb{E}_{\mathcal{P},\pi_e}[ u_{m}(X_t,A_t,Y_t)^{2}]] \\
    &= C_{1} \mathbb{E}_{\mathcal{P},\pi_e}[ u_{m}(X_t,A_t,Y_t)^{2}],
\end{align*}
where the last term is bounded by a constant due to \cref{assumption:domination}.

Similarly, we have 

\begin{align*}
    &\E\left[\sum_{t=1}^{T} \pi_{e}(A_{t} =a | X_{t})  \sum_{a=1}^{K} \left(1 - \frac{\mathbbm{1}_{A_{t} = a} }{\mathbb{P} \left(A_{t} =a | X_{t},\mathcal{H}_{t-1} \right)}\right) u_{\mathcal{F}}(X_{t},a,f_{t}(X_{t},a)|\mathcal{H}_{t-1}\right] \\
    &\le \sum_{a=1}^{K} \E\left[ \E \left[\left(1 - \frac{\mathbbm{1}_{A_{t} = a} }{\mathbb{P} \left(A_{t} =a | X_{t},\mathcal{H}_{t-1} \right)}\right) u_{\mathcal{F}}(X_{t},a,f_{t}(X_{t},a) |X_t, \mathcal{H}_{t-1} \right]|, \mathcal{H}_{t-1} \right]\\
    &= \sum_{a=1}^{K} \E\left[ \E \left[\left(1 - \frac{\mathbbm{1}_{A_{t} = a} }{\mathbb{P} \left(A_{t} =a | X_{t},\mathcal{H}_{t-1} \right)}\right) |X_t, \mathcal{H}_{t-1} \right] u_{\mathcal{F}}(X_{t},a,f_{t}(X_{t},a) | \mathcal{H}_{t-1} \right]\\
    &= 0,
\end{align*}
where the last line follows from the fact that  $\E \left[\left(1 - \frac{\mathbbm{1}_{A_{t} = a} }{\mathbb{P} \left(A_{t} =a | X_{t},\mathcal{H}_{t-1} \right)}\right) |X_t, \mathcal{H}_{t-1} \right] = 0$.

On the other hand, the variance can also be bounded. By \cref{assumption:bounded_importance} and the fact that $\pi_{e}(A_{t} =a | X_{t}) < 1$, we have that $\left(1 - \frac{\mathbbm{1}_{A_{t} = a} }{\mathbb{P} \left(A_{t} =a | X_{t},\mathcal{H}_{t-1} \right)}\right)^{2} < (1-C_{2})^{2} $. Therefore, 
$$\mathbb{E} \left[\left(1 - \frac{\mathbbm{1}_{A_{t} = a} }{\mathbb{P} \left(A_{t} =a | X_{t},\mathcal{H}_{t-1} \right)}\right)^{2} u_{\mathcal{F}}(X_{t},a,f_{t}(X_{t},a))^{2}\right] \le (1-C_{2})^{2} \mathbb{E}\left[ u_{\mathcal{F}}(X_{t},a,f_{t}(X_{t},a))^{2} \right],$$
which is finite by \cref{assumption:bracket_complexity}. 

\begin{lemma} \label{lemma:invertible}
Under the conditions of \cref{thm:clt}, $\norm{\left( \frac{1}{T}  \sum_{t=1}^{T} \hat{V}_{t}^{-1/2} \dot{s}_{t,\hat{\theta}_{T}}\right)^{-1}}_{2}  = O_{p}(1)$.
\begin{proof}
Let us decompose

$$\dot{s}_{t,\hat{\theta}_{T}} = \frac{\pi_{e}(A_{t} | X_{t})}{\mathbb{P} \left(A_{t} | X_{t},\mathcal{H}_{t-1} \right) }
\ddot{m}_{\hat{\theta}_{T}}(X_{t},A_{t},Y_{t}) + \pi_{e}(A_{t} =a | X_{t})  \sum_{a=1}^{K} \left(1 - \frac{\mathbbm{1}_{A_{t} = a} }{\mathbb{P} \left(A_{t} =a | X_{t},\mathcal{H}_{t-1} \right)}\right) \ddot{m}_{\hat{\theta}_{T}}(X_{t},a,f_{t}(X_{t},a)).$$
It is sufficient to show the following statements:
\begin{enumerate}
    \item $\frac{1}{T} \sum_{t=1}^{T}\hat{V}_{t}^{-1/2}\frac{\pi_{e}(A_{t} | X_{t})}{\mathbb{P} \left(A_{t} | X_{t},\mathcal{H}_{t-1},\right) }
\ddot{m}_{\theta^{\star}}(X_{t},A_{t},Y_{t})$ is invertible and has bounded eigenvalues with probability tending to 1.
    \item $\norm{\frac{1}{T} \sum_{t=1}^{T}\hat{V}_{t}^{-1/2}\frac{\pi_{e}(A_{t} | X_{t})}{\mathbb{P} \left(A_{t} | X_{t},\mathcal{H}_{t-1},\right) }
\left(\ddot{m}_{\hat{\theta}_{T}}(X_{t},A_{t},Y_{t}) - \ddot{m}_{\theta^{\star}}(X_{t},A_{t},Y_{t}) \right)} = o_{p}(1)$ 
    \item $\norm{\frac{1}{T} \sum_{t=1}^{T}\pi_{e}(A_{t} =a | X_{t})  \sum_{a=1}^{K} \left(1 - \frac{\mathbbm{1}_{A_{t} = a} }{\mathbb{P} \left(A_{t} =a | X_{t},\mathcal{H}_{t-1} \right)}\right) \ddot{m}(X_{t},a,f_{t}(X_{t},a))} =o_{p}(1).$
\end{enumerate}
\paragraph{Demonstrating (1)}

First, note that
\begin{align*}
\E\left[\hat{V}_{t}^{-1/2}\frac{\pi_{e}(A_{t} | X_{t})}{\mathbb{P} \left(A_{t} | X_{t},\mathcal{H}_{t-1},\right) }\ddot{m}_{\theta^{\star}}(X_{t},A_{t},Y_{t}) |\mathcal{H}_{t-1}\right]  &=\hat{V}_{t}^{-1/2}\E\left[\frac{\pi_{e}(A_{t} | X_{t})}{\mathbb{P} \left(A_{t} | X_{t},\mathcal{H}_{t-1},\right) }\ddot{m}_{\theta^{\star}}(X_{t},A_{t},Y_{t})|\mathcal{H}_{t-1}\right] \\
&=\hat{V}_{t}^{-1/2}\E_{\mathcal{P},p_{e}}\left[\ddot{m}_{\theta^{\star}}(X_{t},A_{t},Y_{t})\right] \\
\end{align*}
By the weak law of large numbers, we can therefore conclude that 

\begin{align*}
   \frac{1}{T} \sum_{t=1}^{T}\hat{V}_{t}^{-1/2}\frac{\pi_{e}(A_{t} | X_{t})}{\mathbb{P} \left(A_{t} | X_{t},\mathcal{H}_{t-1},\right) }
\ddot{m}_{\theta^{\star}}(X_{t},A_{t},Y_{t}) &= o_p(1) + \frac{1}{T}\sum_{t=1}^{T}\hat{V}_{t}^{-1/2} \E_{\mathcal{P},p_{e}}\left[\ddot{m}_{\theta^{\star}}(X_{t},A_{t},Y_{t})\right].\\
\end{align*}
Rearranging terms, we can write this as 

\begin{align*}
 o_p(1) +  \frac{1}{T}\sum_{t=1}^{T}(\hat{V}_{t}^{-1/2} - V_{t,\theta^{\star}}^{1/2})\E_{\mathcal{P},p_{e}}\left[\ddot{m}_{\theta^{\star}}(X_{t},A_{t},Y_{t})\right] + \frac{1}{T}\sum_{t=1}^{T}V_{t,\theta^{\star}}^{-1/2} \E_{\mathcal{P},p_{e}}\left[\ddot{m}_{\theta^{\star}}(X_{t},A_{t},Y_{t})\right].
\end{align*}

For the second term, we note that:

\begin{align*}
\norm{  \frac{1}{T}\sum_{t=1}^{T}(\hat{V}_{t}^{-1/2} - V_{t,\theta^{\star}}^{1/2})\E_{\mathcal{P},p_{e}}\left[\ddot{m}_{\theta^{\star}}(X_{t},A_{t},Y_{t})\right] }_{\text{op}} &\le 
\frac{1}{T} \sum_{t=1}^{T} \norm{\hat{V}_{t}^{-1/2} - V_{t,\theta^{\star}}^{1/2} }_{\text{op}}  \norm{\E_{\mathcal{P},p_{e}}\left[\ddot{m}_{\theta^{\star}}(X_{t},A_{t},Y_{t})\right]}_{\text{op}} \\
&= o_p(1).\\ 
\end{align*}
The last line follows because we know that  $\norm{\hat{V}_{t}^{-1/2} - V_{t,\theta^{\star}}^{1/2} }_{\text{op}} = o_p(1)$ and $\norm{\E_{\mathcal{P},p_{e}}\left[\ddot{m}_{\theta^{\star}}(X_{t},A_{t},Y_{t})\right]}_{\text{op}}$ is bounded by \cref{assumption:maximum}. 

For the final term, we have that $-\E_{\mathcal{P},p_{e}}\left[\ddot{m}_{\theta^{\star}}(X_{t},A_{t},Y_{t})\right] \succeq H$ for some positive definite matrix $H$ and therefore it has bounded eigenvalues. Similarly, $\norm{\frac{1}{T}\sum_{t=1}^{T}V_{t,\theta^{\star}}^{-1/2}}_{\text{op}} \le \delta_{\text{min}}^{-1/2}$ and $\norm{\left(\frac{1}{T}\sum_{t=1}^{T}V_{t,\theta^{\star}}^{-1/2}\right)^{-1}}_{\text{op}} \ge\delta_{\text{max}}^{-1/2}$. Therefore, the product of both terms has bounded eigenvalues and therefore the third term is invertible with bounded eigenvalues.

\paragraph{Demonstrating (2)} We have by Taylor's theorem that for some $\tilde{\theta}$ on the line segment connecting $\hat{\theta}_{T}$ and $\tilde{\theta}_{T}$ that,
\begin{align*}
&\frac{1}{T} \sum_{t=1}^{T}\hat{V}_{t}^{-1/2}\frac{\pi_{e}(A_{t} | X_{t})}{\mathbb{P} \left(A_{t} | X_{t},\mathcal{H}_{t-1},\right) }
\left(\ddot{m}_{\hat{\theta}_{T}}(X_{t},A_{t},Y_{t}) - \ddot{m}_{\theta^{\star}}(X_{t},A_{t},Y_{t}) \right) \\
\qquad &= \frac{1}{T} \sum_{t=1}^{T}\hat{V}_{t}^{-1/2}\frac{\pi_{e}(A_{t} | X_{t})}{\mathbb{P} \left(A_{t} | X_{t},\mathcal{H}_{t-1},\right) } \dddot{m}_{\tilde{\theta}}(X_t,A_t,Y_t).
\end{align*}
Taking expectations, we have:

\begin{align*}
\norm{\E\left[\hat{V}_{t}^{-1/2}\frac{\pi_{e}(A_{t} | X_{t})}{\mathbb{P} \left(A_{t} | X_{t},\mathcal{H}_{t-1},\right) } \dddot{m}_{\tilde{\theta}}(X_t,A_t,Y_t) | \mathcal{H}_{t-1}\right]}_{1} &= \norm{\hat{V}_{t}^{-1/2} \E_{\mathcal{P},\pi_e}\left[\dddot{m}_{\tilde{\theta}}(X_t,A_t,Y_t) \right]}_{1} \\
&\le \left(\norm{\hat{V}_{t}^{-1/2} - V_{t,\theta^{\star}}}_{1} + \norm{V_{t,\theta_{\star}}}_{1}\right) \norm{ \E_{\mathcal{P},\pi_e}\left[u_{m}(X_t,A_t,Y_t) \right]}_{1}.\\
\end{align*}

Therefore, by \cref{fact:martingale_wlln} we have that
\begin{align*}
&\norm{\frac{1}{T} \sum_{t=1}^{T}\hat{V}_{t}^{-1/2}\frac{\pi_{e}(A_{t} | X_{t})}{\mathbb{P} \left(A_{t} | X_{t},\mathcal{H}_{t-1},\right) }
\left(\ddot{m}_{\hat{\theta}_{T}}(X_{t},A_{t},Y_{t}) - \ddot{m}_{\theta^{\star}}(X_{t},A_{t},Y_{t}) \right)} \\
&\le o_p(1) + \frac{1}{T}\sum_{t=1}^{T}\left(\norm{\hat{V}_{t}^{-1/2} - V_{t,\theta^{\star}}^{-1/2}}_{1} + \norm{V_{t,\theta_{\star}}}_{1}\right) \norm{ \E_{\mathcal{P},\pi_e}\left[u_{m}(X_t,A_t,Y_t) \right]}_{1} \\
&= o_p(1) +  \frac{1}{T}\sum_{t=1}^{T}\norm{V_{t,\theta_{\star}}}_{1} \norm{ \E_{\mathcal{P},\pi_e}\left[u_{m}(X_t,A_t,Y_t) \right]}_{1}.
\end{align*}
However, the last term is also $o_p(1) $ because $V_{t,\theta^{\star}}$ has bounded eigenvalues and $\E_{\mathcal{P},\pi_e}\left[u_{m}(X_t,A_t,Y_t) \right]$ is bounded by \cref{assumption:domination}.

\paragraph{Demonstrating (3)}
Note that for all $a \in \mathcal{A}$, we have that 
\begin{align}
&\E\Bigg[
\left(1 - \frac{\mathbbm{1}_{A_{t} = a}}
{\mathbb{P}\!\left(A_{t} = a \mid X_{t}, \mathcal{H}_{t-1}\right)}\right)
\ddot{m}_{\hat{\theta}_{T}}\!\left(X_{t}, a, f_{t}(X_{t}, a)\right)
\,\Bigg|\, \mathcal{H}_{t-1}, X_{t}
\Bigg] \nonumber \\
&\qquad =
\E\Bigg[
\left(1 - \frac{\mathbbm{1}_{A_{t} = a}}
{\mathbb{P}\!\left(A_{t} = a \mid X_{t}, \mathcal{H}_{t-1}\right)}\right)
\,\Bigg|\, \mathcal{H}_{t-1}, X_{t}
\Bigg]
\E\Big[
\ddot{m}_{\hat{\theta}_{T}}\!\left(X_{t}, a, f_{t}(X_{t}, a)\right)
\,\Big|\, \mathcal{H}_{t-1}, X_{t}
\Big] \nonumber =0 .
\end{align}
Therefore,  
$$\E\left[\pi_{e}(A_{t} =a | X_{t})  \sum_{a=1}^{K} \left(1 - \frac{\mathbbm{1}_{A_{t} = a} }{\mathbb{P} \left(A_{t} =a | X_{t},\mathcal{H}_{t-1} \right)}\right) \ddot{m}_{\hat{\theta}_{T}}(X_{t},a,f_{t}(X_{t},a)) | \mathcal{H}_{t-1}\right] = 0 $$
By the martingale WLLN, we then have that 
$$\frac{1}{T} \sum_{t=1}^{T}\pi_{e}(A_{t} =a | X_{t})  \sum_{a=1}^{K} \left(1 - \frac{\mathbbm{1}_{A_{t} = a} }{\mathbb{P} \left(A_{t} =a | X_{t},\mathcal{H}_{t-1} \right)}\right) \ddot{m}(X_{t},a,f_{t}(X_{t},a)) =o_{p}(1).$$

\end{proof}

\end{lemma}

\begin{lemma} \label{lemma:score_clt} Under the conditions of \cref{thm:clt},
$\frac{1}{\sqrt{T}}\sum_{t=1}^{T} \hat{V}_{t}^{-1/2} s_{t,\theta^{\star}} \overset{d}{\to} N(0,I_{d})$.
\end{lemma}
\begin{proof} Consider the term $Z_{t} := c^{T}\hat{V}_{t}^{-1/2} s_{t,\theta^{\star}}$, for any $c\in \mathbb{R}^{d}$. By the Cramer-Wold device, it suffices to show that $\frac{1}{\sqrt{T}} \sum_{t=1}^{T} Z_{t} \overset{d}{\to} N(0,c^{2})$. To do this, we demonstrate that $Z_{t}$ is a martingale difference sequence and then apply the martingale CLT of \cite{dvoretzky1972asymptotic}. In order to apply this fact, we need the following to be true:
\begin{itemize}
    \item \textbf{Conditional Expectation} $\E[Z_{t} | \mathcal{H}_{t-1}] = 0$ for all $t\in [T]$;
    \item \textbf{Conditional Variance} $\frac{1}{T}\sum_{t=1}^{T}\E[Z_{t}^{2} | \mathcal{H}_{t-1}] = \norm{c}^{2}$;
    \item \textbf{Lindeberg Condition} $\frac{1}{T}\sum_{t=1}^{T}\E[Z_{t}^{2}1_{Z_{t} > \epsilon} | \mathcal{H}_{t-1}] \overset{p}{\to} 0$.
    
\end{itemize}

Checking each condition separately:
\paragraph{Conditional Expectation}  By construction, we have that 
\begin{align*}
\E[Z_{t} | \mathcal{H}_{t-1}] = \E[c^{T}\hat{V}_{t}^{-1/2} s_{t,\theta^{\star}}| \mathcal{H}_{t-1}] =c^{T}\hat{V}_{t}^{-1/2}\E[s_{t,\theta^{\star}}| \mathcal{H}_{t-1}],
\end{align*}
based on the fact that  $\hat{V}_{t}^{-1/2} \in \sigma( \mathcal{H}_{t-1}) $. So it suffices to demonstrate that $\E[s_{t,\theta^{\star}}| \mathcal{H}_{t-1}] = 0 $. 

First, we will rewrite
\begin{align*}
    \E[s_{t,\theta^{\star}}| \mathcal{H}_{t-1}] = & \E_{\mathcal{P},\pi_{t}}\left[\frac{  \pi_{e}(A_{t} | X_{t})}{\mathbb{P} \left(A_{t}| X_{t}, \mathcal{H}_{t-1},X_t\right) } \dot{m}_{\theta^{\star}}(A_{t},X_{t},Y_{t}) |\mathcal{H}_{t-1}\right] \\
    & + \sum_{a =1}^{K} \E_{\mathcal{P},\pi_{t}} \left[ 
    \left( 1 - \frac{ \mathbbm{1}_{A_t=a}  }{  \mathbb{P} \left(A_{t}| X_{t}, \mathcal{H}_{t-1} \right)} \right)\dot{m}_{\theta^{\star}}(a,X_{t},f_{t}(a,X_{t})) | \mathcal{H}_{t-1} \right].  \\
\end{align*}
For the first term, note that

\begin{align*}\E_{\mathcal{P},\pi_{t}}\left[\frac{  \pi_{e}(A_{t} =a | X_{t})}{\mathbb{P} \left(A_{t}| X_{t}, \mathcal{H}_{t-1}\right) } \dot{m}_{\theta^{\star}}(A_{t},X_{t},Y_{t}) | \mathcal{H}_{t-1}\right] 
&=\E_{\mathcal{P},\pi_{e}}\left[ \dot{m}_{\theta^{\star}}(A_{t},X_{t},Y_{t}) | \mathcal{H}_{t-1}\right] = 0,
\end{align*}
where the first equality follows because of a change of measure and the last line follows because $\theta^{\star}$ is the minimizer of a score equation. 

For the second term, we note that 
\begin{align*}
     &\sum_{a =1}^{K} \E_{\mathcal{P},\pi_{t}} \left[ 
    \left( 1 - \frac{ \mathbbm{1}_{A_t=a}  }{\mathbb{P} \left(A_{t}=a | X_{t} , \mathcal{H}_{t-1}  \right) } \right)\dot{m}_{\theta^{\star}}(a,X_{t},f_{t}(a,X_{t}))  | \mathcal{H}_{t-1}\right]  \\ &= \sum_{a =1}^{K} \E_{\mathcal{P},\pi_{t}} \left[ \E_{\mathcal{P},\pi_{t}} \left[ 
    \left( 1 - \frac{\mathbbm{1}_{A_t=a}  }{\mathbb{P} \left(A_{t}=a | X_{t} , \mathcal{H}_{t-1}  \right)  } \right)\dot{m}_{\theta^{\star}}(a,X_{t},f_{t}(a,X_{t}))  | \mathcal{H}_{t-1}, X_{t} \right]| \mathcal{H}_{t-1}\right] \\
    &= \sum_{a =1}^{K}\E_{\mathcal{P},\pi_{t}}\left[ \dot{m}_{\theta^{\star}}(a,X_{t},f_{t}(a,X_{t})) \E_{\mathcal{P},\pi_{t}} \left[  
    \left( 1 - \frac{ \mathbbm{1}_{A_t=a}  }{\mathbb{P} \left(A_{t}=a| X_{t}, \mathcal{H}_{t-1} \right) } \right)  | \mathcal{H}_{t-1}, X_{t} \right]| \mathcal{H}_{t-1}\right].
\end{align*}
However, $\E_{\mathcal{P},\pi_{t}} \left[  
    1 - \frac{\mathbbm{1}_{A_t=a}  }{\mathbb{P} \left(A_{t}| X_{t}, \mathcal{H}_{t-1} \right) }  | \mathcal{H}_{t-1}, X_{t} \right] = 0$.
Therefore, both the first and second terms are $0$.

\paragraph{Conditional Variance} Let us rewrite $\hat{V}_{t}^{-1/2} = \hat{V}_{t}^{-1/2}+V_{t,\theta^{\star}}^{-1/2}-V_{t,\theta^{\star}}^{-1/2} $. We then have that

\begin{align*}\frac{1}{T}  &\sum_{t=1}^{T} \E\left[c^{T}\hat{V}_{t}^{-1/2} s_{t,\theta^{\star}} s_{t,\theta^{\star}}^{T}    \hat{V}_{t}^{-1/2}c |\mathcal{H}_{t-1}\right] \\
&= \frac{1}{T}  \sum_{t=1}^{T} \E\left[c^{T}( \hat{V}_{t}^{-1/2} + V_{t,\theta^{\star}}^{-1/2} - V_{t,\theta^{\star}}^{-1/2}   )s_{t,\theta^{\star}} s_{t,\theta^{\star}}^{T}  ( \hat{V}_{t}^{-1/2} + V_{t,\theta^{\star}}^{-1/2} - V_{t,\theta^{\star}}^{-1/2}   )c  |\mathcal{H}_{t-1}\right]
\\
&= \frac{1}{T}  \sum_{t=1}^{T} c^{T}( \hat{V}_{t}^{-1/2} + V_{t,\theta^{\star}}^{-1/2} - V_{t,\theta^{\star}}^{-1/2}   )\E\left[s_{t,\theta^{\star}}s_{t,\theta^{\star}}  |\mathcal{H}_{t-1}  ^{T}  \right] ( \hat{V}_{t}^{-1/2} + V_{t,\theta^{\star}}^{-1/2} - V_{t,\theta^{\star}}^{-1/2}   )c \\  
&=c^{T} \frac{1}{T} \sum_{t=1}^{T} \Big[V_{t,\theta^{\star}}^{-1/2} V_{t,\theta^{\star}} V_{t,\theta^{\star}}^{-1/2} - \left(V_{t,\theta^{\star}}^{-1/2} - \hat{V}_{t}^{-1/2}\right)  V_{t,\theta^{\star}}^{1/2} - V_{t,\theta^{\star}}^{1/2} \left(V_{t,\theta^{\star}}^{-1/2} - \hat{V}_{t}^{-1/2}\right)  & \\
&\hspace*{6em} +\left(V_{t,\theta^{\star}}^{-1/2} - \hat{V}_{t}^{-1/2}\right)  V_{t,\theta^{\star}} \left(V_{t,\theta^{\star}}^{-1/2} - \hat{V}_{t}^{-1/2}\right)\Big] c \\
& = \norm{c}^{2} + c^{T}\left( \frac{1}{T}\sum_{t=1}^{T}A_{t}\right)c,
\end{align*}
where $A_{t}  =\left(V_{t,\theta^{\star}}^{-1/2} - \hat{V}_{t}^{-1/2}\right)  V_{t,\theta^{\star}}^{1/2} + V_{t,\theta^{\star}}^{1/2} \left(V_{t,\theta^{\star}}^{-1/2} - \hat{V}_{t}^{-1/2}\right)  +\left(V_{t,\theta^{\star}}^{-1/2} - \hat{V}_{t}^{-1/2}\right)  V_{t,\theta^{\star}} \left(V_{t,\theta^{\star}}^{-1/2} - \hat{V}_{t}^{-1/2}\right)$.
By Cesaro summation, it suffices to show that $\norm{A_{t}}_{\text{op}} \overset{p}{\to} 0$. However, this is true because $\norm{V_{t,\theta^{\star}}^{-1/2} - \hat{V}_{t}^{-1/2}}_{\text{op}} \overset{p}{\to} 0 $ by assumption, and $V_{t,\theta^{\star}}^{1/2}$ has bounded eigenvalues. 
\paragraph{Lindeberg Condition} Fix any $\epsilon > 0$. Note that in general $\mathbbm{1}_{|Z_{t}|\ge \epsilon} \le \frac{Z_{t}^{2} }{\epsilon^{2}} $. We therefore have that,
\begin{align*}\frac{1}{T}  \sum_{t=1}^{T} \E\left[c^{T}\hat{V}_{t}^{-1/2} s_{t,\theta^{\star}} s_{t,\theta^{\star}}^{T}    \hat{V}_{t}^{-1/2}c \mathbbm{1}_{ |c^{T}\frac{1}{\sqrt{T}}\hat{V}_{t}^{-1/2} s_{t,\theta^{\star}} |\ge \epsilon    } |\mathcal{H}_{t-1}\right] &\le \frac{1}{\epsilon^{2}T^{2}}  \sum_{t=1}^{T} \E\left[\left(c^{T}\hat{V}_{t}^{-1/2} s_{t,\theta^{\star}} s_{t,\theta^{\star}}^{T}    \hat{V}_{t}^{-1/2}c \right)^{2}|\mathcal{H}_{t-1}\right]\\
&=\frac{1}{\epsilon^{2}T^{2}}\sum_{t=1}^{T}c^{T}\hat{V}_{t}^{-1/2}\E\left[ (c^{T}s_{t,\theta^{\star}})^{2} (c^{T} \hat{V}_{t}^{-1/2}s_{t,\theta^{\star}})^{2} |\mathcal{H}_{t-1} \right]. \\ 
\end{align*}
It is sufficient to show that $\E\left[ (c^{T}s_{t,\theta^{\star}})^{2} (c^{T} \hat{V}_{t}^{-1/2}s_{t,\theta^{\star}})^{2} |\mathcal{H}_{t-1} \right]$ is bounded. By Cauchy-Schwarz, we have:
\begin{align*}
\E\left[ (c^{T}s_{t,\theta^{\star}})^{2} (c^{T} \hat{V}_{t}^{-1/2}s_{t,\theta^{\star}})^{2} |\mathcal{H}_{t-1} \right] &\le \E\left[ (c^{T}s_{t,\theta^{\star}})^{4}|\mathcal{H}_{t-1}\right]^{1/2}\E\left[ (c^{T} \hat{V}_{t}^{-1/2}s_{t,\theta^{\star}})^{4} |\mathcal{H}_{t-1}\right]^{1/2} \\
&\le 2c^{4} \E\left[ \norm{s_{t,\theta^{\star}}}^{4} |\mathcal{H}_{t-1} \right] \norm{\hat{V}_{t}^{-1/2}}_{\text{op}}^{2}.
\end{align*}
By assumption, we know that $\|\widehat{V}_{t}^{-1/2}\|_{\mathrm{op}}^{2}$ is bounded. Therefore, it is sufficient to demonstrate that
$\E\!\left[ \|s_{t,\theta^{\star}}\|^{4} \mid \mathcal{H}_{t-1} \right] = O_{p}(1)$ to conclude the proof. Recall that
\[
s_{t,\theta^\star}
=\sum_{a=1}^{K}\pi_e(a\mid X_t)\left\{
\dot m_{\theta^\star}(X_t,a,f_t(X_t,a))
+\mathbbm{1}_{\{A_t=a\}}
\frac{\dot m_{\theta^\star}(X_t,a,Y_t)-\dot m_{\theta^\star}(X_t,a,f_t(X_t,a))}{\mathbb P(A_t=a\mid X_t,\mathcal H_{t-1})}
\right\}.
\]
Define
\[
U_{t,a}:=\dot m_{\theta^\star}(X_t,a,f_t(X_t,a)),\qquad
V_{t,a}:=\mathbbm{1}_{\{A_t=a\}}
\frac{\dot m_{\theta^\star}(X_t,a,Y_t)-\dot m_{\theta^\star}(X_t,a,f_t(X_t,a))}{\mathbb P(A_t=a\mid X_t,\mathcal H_{t-1})}.
\]
Since $\pi_e(a\mid X_t)\in[0,1]$ and $K<\infty$, the inequality
\[
\Big\|\sum_{a=1}^{K} w_a z_a\Big\|^{4}
\le K^{3}\sum_{a=1}^{K}\|z_a\|^{4}
\quad\text{for } |w_a|\le 1,
\]
applied with $w_a=\pi_e(a\mid X_t)$ and $z_a=U_{t,a}+V_{t,a}$ yields
\begin{align*}
\|s_{t,\theta^\star}\|^{4}
&\le K^{3}\sum_{a=1}^{K}\|U_{t,a}+V_{t,a}\|^{4} \\
&\le 8K^{3}\sum_{a=1}^{K}\big(\|U_{t,a}\|^{4}+\|V_{t,a}\|^{4}\big),
\end{align*}
where the second inequality follows from $(x+y)^4\le 8(x^4+y^4)$.

We now bound the two terms separately. For $U_{t,a}$, Assumption~6 implies that
\[
\|U_{t,a}\|^{4}
\le u_F(X_t,a,f_t(X_t,a))^{4},
\]
where $u_F$ is an envelope satisfying $\E[u_F(X_t,a,f_t(X_t,a))^{4}]<\infty$. Consequently,
\[
\E\!\left[\|U_{t,a}\|^{4}\mid \mathcal H_{t-1}\right]
\le
\E\!\left[u_F(X_t,a,f_t(X_t,a))^{4}\mid \mathcal H_{t-1}\right]
=O_p(1).
\]
For $V_{t,a}$, bounded importance ratios (\cref{assumption:bounded_importance}) imply that
$\mathbb P(A_t=a\mid X_t,\mathcal H_{t-1})^{-1}\le C_1$, and therefore
\begin{align*}
\|V_{t,a}\|^{4}
&\le C_1^{4}
\big\|\dot m_{\theta^\star}(X_t,a,Y_t)
-\dot m_{\theta^\star}(X_t,a,f_t(X_t,a))\big\|^{4} \\
&\le 8C_1^{4}
\Big(
\|\dot m_{\theta^\star}(X_t,a,Y_t)\|^{4}
+\|U_{t,a}\|^{4}
\Big),
\end{align*}
where the second line again uses $(x+y)^4\le 8(x^4+y^4)$. Taking conditional expectations and
using the fourth-moment bound on $\dot m_{\theta^\star}(X_t,a,Y_t)$ from \cref{assumption:maximum},
together with the bound on $\E[\|U_{t,a}\|^{4}\mid\mathcal H_{t-1}]$, yields
\[
\E\!\left[\|V_{t,a}\|^{4}\mid \mathcal H_{t-1}\right]=O_p(1).
\]
Combining these bounds and summing over the fixed number of actions $K$ gives
\[
\E\!\left[\|s_{t,\theta^\star}\|^{4}\mid \mathcal H_{t-1}\right]
\le
8K^{3}\sum_{a=1}^{K}
\left(
\E[\|U_{t,a}\|^{4}\mid\mathcal H_{t-1}]
+
\E[\|V_{t,a}\|^{4}\mid\mathcal H_{t-1}]
\right)
=O_p(1),
\]
as required.

\end{proof}
\begin{lemma} \label{lemma:consistency}
Under the conditions of \cref{thm:clt}, $\norm{\hat{\theta}_{T}-  \theta^{\star}}_{1} = o_{p}(1)$.
\end{lemma}
\begin{proof}
We note that by \cref{assumption:maximum}, there exists some $\delta_{2} > 0$ such that $\norm{ \E\left[ \dot{m}_{\hat{\theta}_{T} }(X_{t},A_{t},Y_{t}) \right]}_{1} > \delta_{2}$ will imply that $\norm{\hat{\theta}_{T} - \theta^{\star}} \ge \epsilon$. Therefore, 
\begin{align*}
    \mathbb{P} \left(\norm{\hat{\theta}_{T} - \theta^{\star}} \ge \epsilon \right) &\le \mathbb{P} \left( \norm{ \E\left[ \dot{m}_{\hat{\theta}_{T} }(X_{t},A_{t},Y_{t}) \right]}_{1} > \delta_{2} \right).
\end{align*}
However, we know that for all $\theta \in \Theta$, 
\begin{align*}
\E\!\left[s_{t,\theta}\mid \mathcal{H}_{t-1}\right]
&=
\sum_{a=1}^{K}
\E\!\Bigg[
\Bigg.
\pi_{e}(a \mid X_{t})
\Bigg(
m_{\theta}(a,X_{t},f_{t}(a,X_{t}))  \\
&\qquad\qquad
+\, \mathbbm{1}_{\{A_{t}=a\}}
\frac{
m_{\theta}(X_{t},A_{t},Y_{t})
-
m_{\theta}(X_{t},A_{t},f_{t}(X_{t},A_{t}))
}{
\mathbb{P}(A_{t}=a \mid X_{t},\mathcal{H}_{t-1})
}
\Bigg)
\Bigg| \mathcal{H}_{t-1}
\Bigg] \\
&=
\E_{\mathcal{P},\pi_{t}}\!\left[
\frac{\pi_{e}(A_{t}\mid X_{t})}{\mathbb{P}(A_{t}\mid X_{t},\mathcal{H}_{t-1})}
\, m_{\theta}(X_{t},A_{t},Y_{t})
\right] \\
&\quad
+ \sum_{a =1}^{K}
\E_{\mathcal{P},\pi_{t}}\!\left[
\left(
1 -
\frac{\mathbbm{1}_{\{A_t=a\}}}{\mathbb{P}(A_{t}=a\mid X_{t}, \mathcal{H}_{t-1})}
\right)
m_{\theta}(a,X_{t},f_{t}(a,X_{t}))
\right] \\
&=
\E_{\pi_{e}}\!\left[ m_{\theta}(X_{t},A_{t},Y_{t}) \right].
\end{align*}

Note that, 
$$\norm{\mathbb{E}\left[\hat{V}_{t}^{-1/2}s_{t,\theta} | \mathcal{H}_{t-1}\right]}_{1} = \norm{\hat{V}_{t}^{-1/2} \E\left[s_{t,\theta} |\mathcal{H}_{t-1}\right]}_{1} =  \norm{ \hat{V}_{t}^{-1/2} \E\left[\dot{m}_{\theta} (X_{t},A_{t},Y_{t})\right]}_{1} \le \delta_{\text{max}}  \norm{\E\left[ \dot{m}_{\hat{\theta}_{T}}(X_{t},A_{t},Y_{t}) \right]}_{1}.$$
Therefore,

\begin{align*}
    \mathbb{P} \left( \norm{ \E\left[ \dot{m}_{\hat{\theta}_{T}}(X_{t},A_{t},Y_{t}) \right]}_{1} > \delta_{2} \right) &=     \mathbb{P} \left( \delta_{\text{max}} \norm{ \E\left[ \dot{m}_{\hat{\theta}_{T}}(X_{t},A_{t},Y_{t}) \right]}_{1} > \delta_{\text{max}} \delta_{2} \right)\\
    &\le  \mathbb{P} \left( \norm{\mathbb{E}\left[\hat{V}_{t}^{-1/2}s_{t,\theta} | \mathcal{H}_{t-1} \right]}_{1} > \delta_{\text{max}} \delta_{2} \right).
\end{align*}
By assumption, $\frac{1}{T}\sum_{t=1}^{T} \hat{V}_{t}^{-1/2} s_{t,\hat{\theta}_{T}} = o_{p}(1/\sqrt{T})$. Therefore, 
$$ \mathbb{P} \left( \norm{ \E\left[ \dot{m}_{\hat{\theta}_{T}}(X_{t},A_{t},Y_{t}) \right]}_{1} > \delta_{2} \right) \le \mathbb{P} \left( \norm{\frac{1}{T} \sum_{t=1}^{T} \left(\hat{V}_{t}^{-1/2} s_{t,\hat{\theta}_{T}} - \mathbb{E}\left[\hat{V}_{t}^{-1/2}s_{t,\hat{\theta}_{T}} | \mathcal{H}_{t-1}\right]\right)}_{1} > \delta_{\text{max}} \delta_{2} + o_{p}(1) \right).$$
Therefore, it is sufficient to show that 
$$\sup_{\theta \in \Theta }\norm{\frac{1}{T} \left(\sum_{t=1}^{T} \hat{V}_{t}^{-1/2} s_{t,\theta} - \mathbb{E}\left[\hat{V}_{t}^{-1/2}s_{t,\theta} |\mathcal{H}_{t-1}\right]\right)}_{1} \overset{p}{\rightarrow}  0.$$

We show this is true for each component individually. Define $g_{\theta}(x,a,y) := e_{j}^{T}\hat{V}_{t}^{-1/2} \dot{m}_{\theta}(x,a,y)$. By \cref{assumption:bracket3}, for any $\epsilon >0$, $N_{[]}(\epsilon, \dot{\mathcal{M}}_{\Theta}, L_{2}(\mathcal{P},\pi_{e})) < \infty$. However, note that $\hat{V}_{t}^{-1/2}$ has bounded eigenvalues by assumption and $e_j$ is a unit vector, $N_{[]}(\epsilon, \{g_{\theta}(x,a,y) : \theta \in \Theta\}, L_{2}(\mathcal{P},\pi_{e})) < \infty$. Now note that
$$ e_{j}^{T}\hat{V}_{t}^{-1/2} s_{t,\theta} =\sum_{a=1}^{K} \pi_{e}(A_{t} =a | X_{t}) \left( g_{\theta}(a,X_{t},f_{t}(a,X_{t}))   + \mathbbm{1}_{A_{t} = a} \frac{ g_{\theta}(X_{t},A_{t},Y_{t})  -  g_{\theta}(X_{t},A_{t},f_{t}(X_{t},A_{t})) }{ \mathbb{P} \left(A_{t} =a | X_{t},\mathcal{H}_{t-1} \right)  }\right),$$
and we can immediately apply \cref{lemma:uniform_lln} to conclude the proof. 
\end{proof}

\subsection{Proof of \cref{lemma:uniform_lln}}
\begin{proof}
We will decompose $R_{t}(\theta)$ into two terms and consider them separately.
\begin{align*}
R_{t}(\theta) &=  \sum_{a=1}^{K} \pi_{e}(A_{t} =a | X_{t}) \left( g_{\theta}(a,X_{t},f_{t}(a,X_{t}))   + \mathbbm{1}_{A_{t} = a} \frac{ g_{\theta}(X_{t},A_{t},Y_{t})  -  g_{\theta}(X_{t},A_{t},f_{t}(X_{t},A_{t})) }{ \mathbb{P} \left(A_{t} =a | X_{t},\mathcal{H}_{t-1} \right)  }\right) \\
&= \underbrace{\frac{\pi_{e}(A_{t} | X_{t})}{\mathbb{P} \left(A_{t} | X_{t},\mathcal{H}_{t-1} \right) }
g_{\theta}(X_{t},A_{t},Y_{t})}_{_{R_t^{(1)}(\theta) }} +   \underbrace{\sum_{a=1}^{K} \pi_{e}(A_{t} =a | X_{t}) \left(1 - \frac{\mathbbm{1}_{A_{t} = a} }{\mathbb{P} \left(A_{t} =a | X_{t},\mathcal{H}_{t-1} \right)}\right) g_{\theta}(X_{t},a,f_{t}(X_{t},a))}_{R_t^{(2)}(\theta) }.
\end{align*}
\paragraph{Showing first term converges} We have by assume that for any $\delta >0$, there exists a set of brackets $B_{\delta}$ of finite size. This means that for every $\theta$, there exists a pair of functions $(l,u)\in B_{\delta}$ such that:
\begin{enumerate}
    \item $l(x,a,y) \le g_{\theta}(x,a,y) \le u(x,a,y)$ for every $a \in \mathcal{A}$, $(x,y)$ in the support of $\mathcal{P}$;
    \item $\E_{\pi_e} \left[ u(X_{t},A_{t},Y_{t}) - l(X_{t},A_{t},Y_{t})\right]\le \delta $;
    \item $\E \left[ u(X_{t},A_{t},Y_{t})^{2}\right] < \infty$ and $\E \left[ l(X_{t},A_{t},Y_{t})^{2}\right] < \infty$.
\end{enumerate}
Fixing $\delta >0$, we then have the following:
\begin{align*}
\sup_{\theta \in \Theta} R^{(1)}_t(\theta) - \mathbb{E}[R^{(1)}_t(\theta) | \mathcal{H}_{t-1}] &= \sup_{\theta \in \Theta} \frac{\pi_{e}(A_{t} | X_{t})}{\mathbb{P} \left(A_{t} | X_{t},\mathcal{H}_{t-1} \right) }
g_{\theta}(X_{t},A_{t},Y_{t}) - \mathbb{E} \left[\frac{\pi_{e}(A_{t} | X_{t})}{\mathbb{P} \left(A_{t} | X_{t},\mathcal{H}_{t-1} \right) }
g_{\theta}(X_{t},A_{t},Y_{t}) | \mathcal{H}_{t-1}\right] \\
&\le \max_{(l,u) \in B_{\delta}} \bigg\{ \frac{\pi_{e}(A_{t} | X_{t})}{\mathbb{P} \left(A_{t} | X_{t},\mathcal{H}_{t-1} \right) }
u(X_{t},A_{t},Y_{t}) -  \mathbb{E} \left[\frac{\pi_{e}(A_{t} | X_{t})}{\mathbb{P} \left(A_{t} | X_{t},\mathcal{H}_{t-1} \right) }
l(X_{t},A_{t},Y_{t}) | \mathcal{H}_{t-1}\right]\bigg\} \\
&= \max_{(l,u) \in B_{\delta}} \bigg\{\frac{\pi_{e}(A_{t} | X_{t})}{\mathbb{P} \left(A_{t} | X_{t},\mathcal{H}_{t-1} \right) }
u(X_{t},A_{t},Y_{t}) -  \mathbb{E} \left[\frac{\pi_{e}(A_{t} | X_{t})}{\mathbb{P} \left(A_{t} | X_{t},\mathcal{H}_{t-1} \right) }
l(X_{t},A_{t},Y_{t}) | \mathcal{H}_{t-1}\right] + \\
&\qquad \mathbb{E} \left[\frac{\pi_{e}(A_{t} | X_{t})}{\mathbb{P} \left(A_{t} | X_{t},\mathcal{H}_{t-1} \right) }
u(X_{t},A_{t},Y_{t}) | \mathcal{H}_{t-1}\right]
- \mathbb{E} \left[\frac{\pi_{e}(A_{t} | X_{t})}{\mathbb{P} \left(A_{t} | X_{t},\mathcal{H}_{t-1} \right) }
u(X_{t},A_{t},Y_{t}) | \mathcal{H}_{t-1}\right]\bigg\}\\
&= \max_{(l,u) \in B_{\delta}}  \bigg\{\frac{\pi_{e}(A_{t} | X_{t})}{\mathbb{P} \left(A_{t} | X_{t},\mathcal{H}_{t-1} \right) }
u(X_{t},A_{t},Y_{t}) - \mathbb{E} \left[\frac{\pi_{e}(A_{t} | X_{t})}{\mathbb{P} \left(A_{t} | X_{t},\mathcal{H}_{t-1} \right) }
u(X_{t},A_{t},Y_{t})\right] + \\
&\qquad  \mathbb{E}_{\mathcal{P},\pi_{e}} \left[ u(X_t,A_t,Y_t) - l(X_t,A_t,Y_t)\right]| \mathcal{H}_{t-1} \bigg\}.
\end{align*}

First, note that $\mathbb{E}_{\mathcal{P},\pi_{e}} \left[ u(X_t,A_t,Y_t) - l(X_t,A_t,Y_t)\right] \le \delta$ by assumption. For the other term, let $w_{t} =\frac{\pi_{e}(A_{t} | X_{t})}{\mathbb{P} \left(A_{t} | X_{t},\mathcal{H}_{t-1} \right) } $. Now, note that 
$$\max_{(l,u) \in B_{\delta}}  w_{t} u(X_t,A_t,Y_t)  - \mathbb{E}\left[w_{t} u(X_t,A_t,Y_t) | \mathcal{H}_{t-1} \right] \le \sum_{(l,u) \in B_{\delta}} | w_{t} u(X_t,A_t,Y_t)  - \mathbb{E}\left[w_{t} u(X_t,A_t,Y_t) | \mathcal{H}_{t-1} \right].$$
Putting this all together, we have that
$$\frac{1}{T} \sum_{t=1}^{T} R^{(1)}_t(\theta) - \mathbb{E}[R^{(1)}_t(\theta) | \mathcal{H}_{t-1}]  \le \delta + \frac{1}{T}\sum_{t=1}^{T} \sum_{(l,u) \in B_{\delta}}  | w_{t} u(X_t,A_t,Y_t)  - \mathbb{E}\left[w_{t} u(X_t,A_t,Y_t) | \mathcal{H}_{t-1} \right].$$
Since $|B_{\delta}|$ is finite, it is sufficient to show that $w_{t} u(X_t,A_t,Y_t)  - \mathbb{E}\left[w_{t} u(X_t,A_t,Y_t) |\mathcal{H}_{t-1}\right] = o_{p}(1)$ for all $(l,u) \in B_{\delta}$. To do this, we invoke the martingale weak law of large numbers from \cref{fact:martingale_wlln}. Following \cref{remark:wlln_integrable},it is sufficient to show that $E_{\mathcal{P},\pi_{t}}[w_{t}^{2} u(X_t,A_t,Y_t)^{2}]$ is bounded by a constant for all $t$ to demonstrate the third criterion. For this, note that because of \cref{assumption:bounded_importance}, there exists a constant $C_1 > 0$ such that $\frac{\pi_{e}(A_{t} | X_{t})}{\mathbb{P} \left(A_{t} | X_{t},\mathcal{H}_{t-1} \right) } \le \frac{1}{\mathbb{P} \left(A_{t} | X_{t},\mathcal{H}_{t-1} \right) } < C_{1}$. Then
\begin{align*}
    \mathbb{E}[w_{t}^{2} u(X_t,A_t,Y_t)^{2}] &\le \mathbb{E}[C_{1} w_{t} u(X_t,A_t,Y_t)^{2}] \\
    &= C_{1}\mathbb{E}[ w_{t} u(X_t,A_t,Y_t)^{2}]\\
    &= C_{1} \mathbb{E}[\mathbb{E}[ w_{t} u(X_t,A_t,Y_t)^{2} |\mathcal{H}_{t-1}]] \\
    &= C_{1} \mathbb{E}[\mathbb{E}_{\mathcal{P},\pi_e}[ u(X_t,A_t,Y_t)^{2}]] \\
    &= C_{1} \mathbb{E}_{\mathcal{P},\pi_e}[ u(X_t,A_t,Y_t)^{2}] < \infty.
\end{align*}
Therefore, \cref{fact:martingale_wlln} applies, and we can conclude that 
$\frac{1}{T} \sum_{t=1}^{T} R^{(1)}_t(\theta) - \mathbb{E}[R^{(1)}_t(\theta) | \mathcal{H}_{t-1}]  \le \delta + o_{p}(1)$. Taking $\delta \rightarrow 0$ concludes the proof. 

\paragraph{Showing second term converges}
We have that $\pi_e{}(A_t|X_t) < 1$ and $|\mathcal{A}| <\infty$ so it is sufficient to show that for all $a \in \mathcal{A}$
 $$\sup_{\theta \in\Theta}\frac{1}{T}\sum_{t=1}^{T}\left(1 - \frac{\mathbbm{1}_{A_{t} = a} }{\mathbb{P} \left(A_{t} =a | X_{t},\mathcal{H}_{t-1} \right)}\right) g_{\theta}(X_{t},a,f_{t}(X_{t},a)) =o_{p}(1).$$
We will first show pointwise convergence and then use a covering argument to get the uniform result. Define $Z_{t,\theta} = \left(1 - \frac{\mathbbm{1}_{A_{t} = a} }{\mathbb{P} \left(A_{t} =a | X_{t},\mathcal{H}_{t-1} \right)}\right) g_{\theta}(X_{t},a,f_{t}(X_{t},a))$. Note that for all $\theta$, 

\begin{align*}
 \E[Z_{t,\theta}|\mathcal{H}_{t-1}] &= \E \left[  \left(1 - \frac{\mathbbm{1}_{A_{t} = a} }{\mathbb{P} \left(A_{t} =a | X_{t},\mathcal{H}_{t-1} \right)}\right) g_{\theta}(X_{t},a,f_{t}
 (X_{t},a))|\mathcal{H}_{t-1}\right] = 0.
\end{align*}
By \cref{assumption:bracket_complexity}, 
 $\sup_{\theta} \norm{g_{\theta}(X_t,a,f_{t}(X_t))}_{1} \le u_{\mathcal{F}}(X_{t},a,f_{t}(X_{t},a))$ for all $X_t$ and $a$. Note that it is sufficient to show that $\mathbb{E} \left[\left(1 - \frac{\mathbbm{1}_{A_{t} = a} }{\mathbb{P} \left(A_{t} =a | X_{t},\mathcal{H}_{t-1} \right)}\right)^{2} u_{\mathcal{F}}(X_{t},a,f_{t}(X_{t},a))^{2}\right] $
is bounded by a constant for all $t$ in order for \cref{fact:martingale_wlln} to be applied.  However, we note that by \cref{assumption:bounded_importance} and the fact that $\pi_{e}(A_{t} =a | X_{t}) < 1$ that $\left(1 - \frac{\mathbbm{1}_{A_{t} = a} }{\mathbb{P} \left(A_{t} =a | X_{t},\mathcal{H}_{t-1} \right)}\right)^{2} < \max((1-C_{1})^{2},1) $. Therefore, 
$$\mathbb{E} \left[\left(1 - \frac{\mathbbm{1}_{A_{t} = a} }{\mathbb{P} \left(A_{t} =a | X_{t},\mathcal{H}_{t-1} \right)}\right)^{2} u_{\mathcal{F}}(X_{t},a,f_{t}(X_{t},a))^{2}\right] \le \max((1-C_{1})^{2},1) \mathbb{E}\left[ u_{\mathcal{F}}(X_{t},a,f_{t}(X_{t},a))^{2} \right],$$
which is finite by assumption. This implies that for each $\theta$, $\frac{1}{T} \sum_{t=1}^{T} Z_{t,\theta} = 0 $.

Because $\Theta$ is a bounded parameter space, for any $\epsilon>0$, we can cover $\Theta \subseteq \bigcup_{j=1}^{M} \Theta_{j}$. Here, each $\Theta_j$ is an $\epsilon$-ball with centers denoted $\theta_1,...,\theta_k$. For any $\theta$, we can decompose $Z_{t,\theta} = Z_{t,\theta} + Z_{t,\theta_j}$ for some $\theta_j$ such that $|\theta - \theta_j| < \epsilon$. Now, let us write out
$$\sup_{\theta \in\Theta}\frac{1}{T}\sum_{t=1}^{T}Z_{t,\theta} \le \max_{1\le j \le M} \lvert\frac{1}{T}\sum_{t=1}^{T}Z_{t,\theta_j} \rvert +  \lvert\frac{1}{T}\sum_{t=1}^{T}Z_{t,\theta} - Z_{t,\theta_j}\rvert.$$
For the first term, we know that $ \max_{1\le j \le M}\lvert\frac{1}{T}\sum_{t=1}^{T}Z_{t,\theta} \rvert = o_{p}(1)$ because $M$ is finite for each $\epsilon$. We therefore just need to show that $\sup_{\theta \in \Theta_{j}} \lvert\frac{1}{T}\sum_{t=1}^{T}Z_{t,\theta}\rvert =o_{p}(1)$ to finish the proof.

First, note that $Z_{t,\theta} - Z_{t,\theta_j}$ is a martingale difference sequence. Therefore, it is sufficient to show that $ \E \left[\norm{Z_{t,\theta} - Z_{t,\theta_j}}^{2}\right]$ is bounded to conclude the proof. This is true becaue
\begin{align*}
    \E \left[\norm{Z_{t,\theta} - Z_{t,\theta_j}}^{2}\right] &\le \E \left[\left(1 - \frac{\mathbbm{1}_{A_{t} = a} }{\mathbb{P} \left(A_{t} =a | X_{t},\mathcal{H}_{t-1} \right)}\right)^{2} \left( g_{\theta}(X_{t},a,f_{t}(X_{t},a))  - g_{\theta_{j}}(X_{t},a,f_{t}(X_{t},a)) \right)^{2} \right] \\
    &\le \max((1-C_{1})^{2},1)\epsilon^{2}.\\
\end{align*}

\end{proof}

\subsection{Proof of \cref{prop:variance_decomp}}

\begin{proof}
For ease of notation, we denote the quantities $\pi_t(a \mid X_t) :=\mathbb{P}(A_{t}\mid X_{t},\mathcal{H}_{t-1})$, $\mu^\star(X_t,a) := \E[\dot{m}_{\theta^{\star}}(X_t,a,Y_t(a)) \mid X_t]$, and $v^\star(X_t,a) := \text{Var}[\dot{m}_{\theta^{\star}}(X_t,a,Y_t(a)) \mid X_t]$.

Let us rewrite using the law of total variance:
\begin{align*}
\Var \left(s_{t,\theta^{\star}} \mid \mathcal{H}_{t-1} \right) 
&= \Var\Big(\E \left[ s_{t,\theta^{\star}} \mid \mathcal{H}_{t-1}, X_{t}\right] \mid \mathcal{H}_{t-1}\Big) 
+ \E\Big[\Var \left( s_{t,\theta^{\star}} \mid \mathcal{H}_{t-1}, X_{t} \right) \mid \mathcal{H}_{t-1}\Big].
\end{align*}
Now, analyzing each of these terms separately: 
\begin{align*}
\E \left[s_{t,\theta^{\star}} \mid \mathcal{H}_{t-1}, X_{t} \right] 
&= \E\Bigg[
    \sum_{a=1}^{K} \pi_{e}(A_{t} =a \mid X_{t})
    \Big(
        \dot{m}_{\theta^{\star}}(X_{t},a,f_{t}(X_{t},a)) 
        \\
&\qquad\qquad
        + \mathbbm{1}_{A_{t}=a}
        \frac{
            \dot{m}_{\theta^{\star}}(X_{t},a,Y_{t})  
            - \dot{m}_{\theta^{\star}}(X_{t},a,f_{t}(X_{t},a)
        }{
            \pi_t(a \mid X_t) 
        }
    \Big)
    \,\Big|\, \mathcal{H}_{t-1}, X_{t}
\Bigg] \\
&=  \sum_{a=1}^{K} \pi_{e}(A_{t} = a \mid X_{t} ) 
   \E \left[\frac{\mathbbm{1}_{A_{t}=a}}{ \pi_t(a \mid X_t)  }\dot{m}_{\theta^{\star}}(X_t,a,Y_{t}(a)) \mid X_{t}, \mathcal{H}_{t-1} \right]  \\
&\qquad\qquad + \sum_{a=1}^{K}\pi_{e}(A_{t} = a \mid X_{t} ) \E\left[\left(1 - \frac{\mathbbm{1}_{A_{t}=a}}{ \pi_t(a \mid X_t)  }\right)\dot{m}_{\theta^{\star}}(X_{t},a,f_{t}(X_{t},a)  \mid X_{t}, \mathcal{H}_{t-1}\right] \\
&=\sum_{a=1}^{K} \pi_{e}(A_{t} = a \mid X_{t} ) 
   \E \left[\dot{m}_{\theta^{\star}}(X_t,a,Y_{t}(a)) \mid X_{t} \right] \\
&= \E_{A_{t} \sim \pi_{e}}  
   \left[\dot{m}_{\theta^{\star}}(X_t,A_{t},Y_{t}) \mid X_{t}\right].
\end{align*}
On the other hand, note that $\dot{m}_{\theta}(a,X_{t},f_{t}(a,X_{t})) \in \sigma(X_{t}, \mathcal{H}_{t-1})$. Therefore, we can rewrite the conditional variance given \(X_t\) as:
\begin{align*}
\Var \left( s_{t,\theta^{\star}} \mid \mathcal{H}_{t-1}, X_{t} \right) 
&= \Var \Bigg( \sum_{a=1}^{K} \pi_e(a \mid X_t) \frac{\mathbbm{1}_{A_t = a}}{\pi_t(a \mid X_t)} 
\dot{m}_{\theta^{\star}}(X_t,a,Y_t) \,\Big|\, \mathcal{H}_{t-1}, X_{t} \Bigg) \\
&= \sum_{a=1}^{K} \frac{\pi_e(a \mid X_t)^2}{\pi_t(a|X_t)} 
\E[\dot{m}_{\theta^{\star}}(X_t,a,Y_t(a))\dot{m}_{\theta^{\star}}(X_t,a,Y_t(a))^{T} \mid X_t]\\
&\qquad -\Big(\sum_{a=1}^{K} \pi_e(a \mid X_t) \mu^\star(X_t,a)\Big)\Big(\sum_{a=1}^{K} \pi_e(a \mid X_t) \mu^\star(X_t,a)\Big)^{T} \\
&= \sum_{a=1}^{K} \frac{\pi_e(a \mid X_t)^2}{\pi_t(a|X_t)} 
\E[\dot{m}_{\theta^{\star}}(X_t,a,Y_t(a))\dot{m}_{\theta^{\star}}(X_t,a,Y_t(a))^{T} \mid X_t]\\
&\qquad -\E_{A_{t} \sim \pi_{e}}  \left[\dot{m}_{\theta^{\star}}(X_t,A_{t},Y_t) \mid X_t\right]\E_{A_{t} \sim \pi_{e}}  \left[\dot{m}_{\theta^{\star}}(X_t,A_{t},Y_t) \mid X_t\right]^{T}.
\end{align*} 

Finally, combining the two terms, we obtain the full decomposition:
\begin{align*}
\Var(s_{t,\theta^{\star}} \mid \mathcal{H}_{t-1}) 
&=\Var\Big(\E_{A_{t} \sim \pi_{e}}  \left[\dot{m}_{\theta^{\star}}(X_t,A_{t},Y_t) \mid X_t\right]\Big) \\
\\
&\qquad +\E \Bigg[\sum_{a=1}^{K} \frac{\pi_e(a \mid X_t)^2}{\pi_t(a|X_t)} 
\E[\dot{m}_{\theta^{\star}}(X_t,a,Y_t(a))\dot{m}_{\theta^{\star}}(X_t,a,Y_t(a))^{T} \mid X_t] \mid \mathcal{H}_{t-1}\Bigg] \\
 & \qquad - \E\left[\E_{A_{t} \sim \pi_{e}}  \left[\dot{m}_{\theta^{\star}}(X_t,A_{t},Y_t) \mid X_t\right]\E_{A_{t} \sim \pi_{e}}  \left[\dot{m}_{\theta^{\star}}(X_t,A_{t},Y_t) \mid X_t\right]^{T} \right].
\end{align*}

\end{proof}

\subsection{Proof of \cref{prop:estimator}}
For ease of notation, let us write  $\nu^{\star}_{t}(X_t) := \E_{A_{t} \sim \pi_{e}} \left[ \dot{m}_{\theta}(X_t,A_t,Y_t) |X_{t} \right] $. We will tackle this proof in three parts:
\begin{enumerate}
\item Showing $\frac{1}{n - 1} 
\Biggl( \sum_{X_i \in \mathcal{\tilde{X}}} \hat{\nu}_t(X_i)
- \bar{\nu} \Biggr)
\Biggl( \sum_{X_i \in \mathcal{\tilde{X}}} \hat{\nu}_t(X_i) 
- \bar{\nu} \Biggr)^{T} - \Var\left(\E_{A_{t} \sim \pi_{e}} \left[ \dot{m}_{\theta}(X_t,A_t,Y_t) |X_{t} \right]\right) = o_{p}(1)$;
\item Showing 
\begin{align*}
&\frac{1}{n} \sum_{X_i \in \mathcal{\tilde{X}}} \sum_{a=1}^{K} 
\frac{\pi_e(A_t = a \mid X_i)^2 \, h_t(X_i,a)}{\mathbb{P}(A_t = a \mid X_i, \mathcal{H}_{t-1})} \\
& \qquad - \E \Bigg[\sum_{a=1}^{K} \frac{\pi_e(a \mid X_t)^2}{\pi_t(a|X_t)}\E[\dot{m}_{\theta^{\star}}(X_t,a,Y_t(a))\dot{m}_{\theta^{\star}}(X_t,a,Y_t(a))^{T} \mid X_t] \mid \mathcal{H}_{t-1}\Bigg] = o_{p}(1);
\end{align*}
\item Showing 
\begin{align*}&\frac{1}{n} \sum_{X_i \in \mathcal{\tilde{X}}}  \hat{\nu}_{t}(X_i)\hat{\nu}_{t}(X_i)^{T}-  \E\Bigg[\E_{A_{t} \sim \pi_{e}}  \left[\dot{m}_{\theta^{\star}}(X_t,A_{t},Y_t) \mid X_t\right]\E_{A_{t} \sim \pi_{e}}  \left[\dot{m}_{\theta^{\star}}(X_t,A_{t},Y_t) \mid X_t\right]^{T} \mid \mathcal{H}_{t-1}\Bigg] = o_{p}(1).
\end{align*}
\end{enumerate}
Adding term (1)-(3) together will demonstrate that
$\norm{\hat{V}_{t} - V_{t,\theta^{\star}}}_{\text{op}} \overset{p}{\to} 0$. 

\paragraph{Showing (1)}

First, we will show that $\hat{\nu}_{t}(X_i) - \E_{A_{t} \sim \pi_{e}} \left[ \dot{m}_{\theta}(X_t,A_t,Y_t) |X_{t} \right]\overset{p}{\to} 0$. To see this, write
\begin{align*}
\hat{\nu}_{t}(X_i) - \E_{A_{t} \sim \pi_{e}} \left[ \dot{m}_{\theta}(X_t,A_t,Y_t) |X_{t} \right] &= \sum_{a = 1}^K \pi_{e}(a|X_{t})g_{t}(a,X_i) - \E_{A_{t} \sim \pi_{e}} \left[ \dot{m}_{\theta}(X_t,A_t,Y_t) \right] \\ 
&= \sum_{a=1}^K \pi_{e}(a|X_{t})g_{t}(a,X_i) - \sum_{a=1}^K \pi_{e}(a|X_{t}) \E\left[ \dot{m}_{\theta}(X_t,a,Y_t) | X_t\right] \\
&= \sum_{a=1}^K \pi_{e}(a|X_{t}) \left(g_{t}(a,X_i)  - \E\left[ \dot{m}_{\theta}(X_t,A_t,Y_t) | A_t=a, X_t\right] \right) \\
&\le K \left(g_{t}(a,X_i)  - \E\left[ \dot{m}_{\theta}(X_t,A_t,Y_t) | A_t=a, X_t\right] \right)\\
&=o_{p}(1).
\end{align*}
The final line holds by \cref{assumption:converges_var}.

Let us consider the first term in the expression $    \frac{1}{n}\left( \sum_{X_{i} \in \tilde{X}} \hat{\nu}_{t}(X_i) - \bar{\nu} \right)\left( \sum_{X_{i} \in \tilde{X}} \hat{\nu}_{t}(X_i)^{2} - \bar{\nu} \right)^{T}$. Let us substitute $\hat{\nu}_{t}(X_i) = \hat{\nu}_{t}(X_i) - \nu^{\star}_{t}(X_t) + \nu^{\star}_{t}(X_t)$. We then obtain the result that the above expression is equal to 
\begin{align*}
    \frac{1}{n} \sum_{X_{i} \in \tilde{X}} \nu^{\star}_{t}(X_t)\nu^{\star}_{t}(X_t)^{T} - \left(\frac{1}{n} \sum_{X_{i} \in \tilde{X}} \nu^{\star}_{t}(X_t) \right) \left(\frac{1}{n} \sum_{X_{i} \in \tilde{X}} \nu^{\star}_{t}(X_t) \right)^{T}+ o_{p}(1),
\end{align*}
which converges to  $\Var\left(\E_{A_{t} \sim \pi_{e}} \left[ \dot{m}_{\theta}(X_t,A_t,Y_t) |X_{t} \right]\right)$ by the weak law of large numbers.

\paragraph{Showing (2)} Now, let us consider the expression 

\begin{align*}
&\frac{1}{n} \sum_{X_i \in \tilde{X}} \sum_{a=1}^{K}
\frac{\pi_{e}(A_{t}=a \mid X_{t})^{2}}
{\mathbb{P}\!\left(A_{t}=a \mid X_{t}, \mathcal{H}_{t-1}\right)}
\, h_t(X_i,a) \\
&\quad - \E\Bigg[
\sum_{a=1}^{K}
\frac{\pi_{e}(A_{t}=a \mid X_{t})^{2}}
{\mathbb{P}\!\left(A_{t}=a \mid X_{t}, \mathcal{H}_{t-1}\right)}
\, \E\Big[
\dot{m}_{\theta^{\star}}(X_t,A_t,Y_t)
\dot{m}_{\theta^{\star}}(X_t,A_t,Y_t)^{T}
\;\Big|\; X_t, A_t=a
\Big]
\;\Big|\; \mathcal{H}_{t-1}
\Bigg].
\end{align*}
This is simply equal to 
\begin{align*}
  \frac{1}{n} \sum_{X_i \in \tilde{X}} &\sum_{a=1}^{K} \pi_{e}(A_{t} =a | X_{t})^{2} \frac{ \left( h_t(X_i,a)  -  \E\left[\E \left[  \dot{m}_{\theta^{\star}}(X_t,A_t,Y_{t})\dot{m}_{\theta^{\star}}(X_t,A_t,Y_{t})^{T} | X_{t},A_t=a \right]\right] \right) }{ \mathbb{P} \left(A_{t}=a| X_{t}, \mathcal{H}_{t-1}\right)}  \\
  &\le \frac{C_1}{n} \left[\sum_{X_i \in \tilde{X}}\sum_{a=1}^{K} h_t(X_i,a)\right] -C_1\E\left[\E \left[  \dot{m}_{\theta^{\star}}(X_t,A_t,Y_{t})\dot{m}_{\theta^{\star}}(X_t,A_t,Y_{t})^{T} | X_{t},A_t=a \right]\right].
\end{align*}
To show that this expression is $o_{p}(1)$, it is sufficient to show that $$ \frac{1}{n}\sum_{X_i \in \tilde{X}} h_t(X_i,a) - \E\left[\E \left[  \dot{m}_{\theta^{\star}}(X_t,A_t,Y_{t})\dot{m}_{\theta^{\star}}(X_t,A_t,Y_{t})^{T} | X_{t},A_t=a \right]\right] =o_{p}(1).$$ for all $a\in \mathcal{A}$. Note that by Assumption~\ref{assumption:converges_var},
\begin{align*}
   \frac{1}{n}\sum_{X_i \in \tilde{X}} h_t(X_i,a) &=    \frac{1}{n}\sum_{X_i \in \tilde{X}} h_t(X_i,a)  + \E \left[  \dot{m}_{\theta^{\star}}(X_t,A_t,Y_{t})\dot{m}_{\theta^{\star}}(X_t,A_t,Y_{t})^{T} | X_{t},A_t=a \right] - \\
   &\qquad \E \left[  \dot{m}_{\theta^{\star}}(X_t,A_t,Y_{t})\dot{m}_{\theta^{\star}}(X_t,A_t,Y_{t})^{T} | X_{t},A_t=a \right]\\
   &= o_{p}(1) + \frac{1}{n}\sum_{X_i \in \tilde{X}} \E \left[  \dot{m}_{\theta^{\star}}(X_t,A_t,Y_{t})\dot{m}_{\theta^{\star}}(X_t,A_t,Y_{t})^{T} | X_{t},A_t=a \right].
\end{align*}
Therefore,  $ \frac{1}{n}\sum_{X_i \in \tilde{X}} h_t(X_i,a) - \E\left[\E \left[  \dot{m}_{\theta^{\star}}(X_t,A_t,Y_{t})\dot{m}_{\theta^{\star}}(X_t,A_t,Y_{t})^{T} | X_{t},A_t=a \right]\right] =o_{p}(1)$ by the weak law of large numbers. 

\paragraph{Showing (3)} Finally, consider the expression 
\begin{align} \hat{\nu}_{t}(X_i)\hat{\nu}_{t}(X_i)^{T}
-  \Big(\sum_{a=1}^{K} \pi_e(a \mid X_t) \mu^\star(X_t,a)\Big)\Big(\sum_{a=1}^{K} \pi_e(a \mid X_t) \mu^\star(X_t,a)\Big)^{T} . \label{eqn:third}
\end{align}
Let $\Delta_t(x) := \widehat{\nu}_t(x) - \nu_t^\star(x)$ which we already showed was $o_{p}(1)$. Now, we can write:
\begin{align*}
\widehat{\nu}_t(X_i)\widehat{\nu}_t(X_i)^{\!\top}
&=
\big(\Delta_t(X_i) + \nu_t^\star(X_i)\big)
\big(\Delta_t(X_i) + \nu_t^\star(X_i)\big)^{\!\top} \\
&=
\Delta_t(X_i)\Delta_t(X_i)^{\!\top}
\;+\;
\Delta_t(X_i)\nu_t^\star(X_i)^{\!\top}
\;+\;
\nu_t^\star(X_i)\Delta_t(X_i)^{\!\top}
\;+\;
\nu_t^\star(X_i)\nu_t^\star(X_i)^{\!\top}.
\end{align*}

Subtracting the last term on both sides yields the identity
\[
\widehat{\nu}_t(X_i)\widehat{\nu}_t(X_i)^{\!\top}
-
\nu_t^\star(X_i)\nu_t^\star(X_i)^{\!\top}
=
\Delta_t(X_i)\Delta_t(X_i)^{\!\top}
+
\Delta_t(X_i)\nu_t^\star(X_i)^{\!\top}
+
\nu_t^\star(X_i)\Delta_t(X_i)^{\!\top}.
\]
Each individual term is $o_{p}(1)$ because $\nu_t^\star(X_i)$ is integrable. Therefore, we have that \cref{eqn:third} is $o_p(1)$. Now, consider the expression.

\begin{align*} &\frac{1}{n}\sum_{i=1}^{n}\hat{\nu}_{t}(X_i)\hat{\nu}_{t}(X_i)^{T}
- \E\left[\E_{A_{t} \sim \pi_{e}}  \left[\dot{m}_{\theta^{\star}}(X_t,A_{t},Y_t) \mid X_t\right]\E_{A_{t} \sim \pi_{e}}  \left[\dot{m}_{\theta^{\star}}(X_t,A_{t},Y_t) \mid X_t\right]^{T} \mid \mathcal{H}_{t-1}\right] \\
&=\frac{1}{n}\sum_{i=1}^{n}\left(\hat{\nu}_{t}(X_i)\hat{\nu}_{t}(X_i)^{T}
- \nu_t^\star(X_i)\nu_t^\star(X_i)^{\!\top}\right) \\
&\qquad +\frac{1}{n}\sum_{i=1}^{n}\nu_t^\star(X_i)\nu_t^\star(X_i)^{\!\top}-\E\left[\E_{A_{t} \sim \pi_{e}}  \left[\dot{m}_{\theta^{\star}}(X_t,A_{t},Y_t) \mid X_t\right]\E_{A_{t} \sim \pi_{e}}  \left[\dot{m}_{\theta^{\star}}(X_t,A_{t},Y_t) \mid X_t\right]^{T} \mid \mathcal{H}_{t-1}\right].\\
\end{align*}
The first term is $o_{p}(1)$ by the arguments given above, the second term is $o_{p}(1)$ by the weak law of large numbers. 
\subsection{Proof of \cref{prop:glm}}
We have for all $\theta$ that 
$$ \E\left[\dot{m}_{\theta}(X_i,A_i,Y_i) | X_{i},A_i=a \right] = \E\left[z_{\theta}(X_i,a)Y_{i} + v_{\theta}(X_i,a) |X_i,A_i=a\right] =  z_{\theta}(X_i,a)\E[Y_{i}|X_i,A_i=a] + v_{\theta}(X_i,a).$$
Therefore, 
$$ g_{t}(X_i,a) - \E\left[m_{\theta^{\star}}(X_i,A_i,Y_i) | X_{i},A_i =a\right] = 
z_{\bar{\theta}_{T}}(X_i,a) f_{t}(X_i,a) - z_{\theta^{\star}}(X_i,a) E[Y_{i}|X_i,A_i=a]+ v_{\bar{\theta}_{T}}(X_i,a) - v_{\theta^{\star}}(X_i,a). 
$$
We have that $v_{\bar{\theta}_{T}}(X_i,A_i) - v_{\theta^{\star}}(X_i,A_i) = o_{p}(1)$ by the consistency of $\bar{\theta}_{T}$ combined with the continuous mapping theorem. Note for the remaining terms that
\begin{align*}
&z_{\bar{\theta}_{T}}(X_i,a) f_{t}(X_i,a) - z_{\theta^{\star}}(X_i,a) E[Y_{i}|X_i,A_i=a]  \\ 
&=z_{\bar{\theta}_{T}}(X_i,a) f_{t}(X_i,a) - z_{\theta^{\star}}(X_i,a) E[Y_{i}|X_i,A_i=a] + z_{\bar{\theta}_{T}}(X_i,a) E[Y_{i}|X_i,A_i=a]  -  z_{\bar{\theta}_{T}}(X_i,a) E[Y_{i}|X_i,A_i=a]\\ 
&=(z_{\bar{\theta}_{T}}(X_i,a) - z_{\theta^{\star}}(X_i,a)) E[Y_{i}|X_i,A_i=a] 
+z_{\bar{\theta}_{T}}(X_i,a) \left(f_{t}(X_i,a) - E[Y_{i}|X_i,A_i=a]\right). 
\end{align*}
Applying the assumptions that $\bar{\theta}_{t} \overset{p}{\to} \theta^{\star}$ and $f_{t}(X_i,a) - E[Y_{i}|X_i,A_i=a]\overset{p}{\to} 0$ along with the continuous mapping theorem immediately yields the result that $g_{t}(X_i,a)$ is consistent.

Next, we will prove the second part of the theorem. By the properties of variance and because $z_{\theta}(X_i,A_i)$ is fixed conditional on $X_i$ and $A_i$, we have that
$$ \Var\left[\dot{m}_{\theta}(X_i,A_i,Y_i) | X_{i},A_i =a\right] = z_{\theta}(X_i,a) z_{\theta}(X_i,a)^{T} \Var(Y_{i} | X_i,A_i=a).$$
Analyzing the term, 
\begin{align*}
h_t(X_i,a) &- \Var\big(\dot m_{\theta^\star}(X_i,A_i,Y_i)\mid X_i,A_i\big) \\
&= z_{\bar\theta_T}(X_i,a) z_{\bar\theta_T}(X_i,a)^\top j_{t}(X_i,a)  
        - z_{\theta^\star}(X_i,a) z_{\theta^\star}(X_i,a)^\top\Var(Y_{i} | X_i,A_i=a)  \\
&= z_{\bar\theta_T}(X_i,a) z_{\bar\theta_T}(X_i,a)^\top j_{t}(X_i,a) - z_{\bar\theta_T}(X_i,a) z_{\bar\theta_T}(X_i,a)^\top \Var(Y_{i} | X_i,A_i=a) \\
&\quad \quad +z_{\bar\theta_T}(X_i,a) z_{\bar\theta_T}(X_i,a)^\top \Var(Y_{i} | X_i,A_i=a) - z_{\theta^\star}(X_i,a) z_{\theta^\star}(X_i,a)^\top\Var(Y_{i} | X_i,A_i=a)  \\
&= z_{\bar\theta_T}(X_i,a) z_{\bar\theta_T}(X_i,a)^\top \left( j_{t}(X_i,a) - \Var(Y_{i} | X_i,A_i=a)  \right) \\
& \quad \quad +  \left(z_{\bar\theta_T}(X_i,a) z_{\bar\theta_T}(X_i,a)^\top - z_{\theta^\star}(X_i,a) z_{\theta^\star}(X_i,a)^\top \right)\Var(Y_{i} | X_i,A_i=a).
\end{align*}
Applying the assumptions that $\bar{\theta}_{t} \overset{p}{\to} \theta^{\star}$, $j_{t}(X_i,a) - \text{Var}[Y_{i}|X_i,A_i=a]\overset{p}{\to} 0$ together with the continuous mapping theorem immediately yields the result that $h_{t}(X_i,a)$ is consistent.

\subsection{Proof of \cref{prop:consistency_tilde}}
\begin{proof}
Define 
\begin{equation*} 
R_{t}(\theta) =  \sum_{a=1}^{K} \pi_{e}(A_{t} =a | X_{t}) \left( m_{\theta}(a,X_{t},f_{t}(a,X_{t}))   + \mathbbm{1}_{A_{t} = a} \frac{ m_{\theta}(X_{t},A_{t},Y_{t})  -  m_{\theta}(X_{t},A_{t},f_{t}(X_{t},A_{t})) }{ \mathbb{P} \left(A_{t} =a | X_{t},\mathcal{H}_{t-1} \right)  }\right).
\end{equation*}
By definition of $\tilde{\theta}_{T}$, we have that 
$$ \sum_{t=1}^{T} R_{t}(\tilde{\theta}_{T}) = \sup_{\theta \in \Theta} \sum_{t=1}^{T}  R_{t}(\tilde{\theta}_{T})\ge \sum_{t=1}^{T} R_{t}(\theta^{\star}).$$
This implies that
\begin{align*}
    \mathbb{P} \left( \norm{\tilde{\theta}_{T} - \theta^{\star}} \ge \epsilon \right) 
    &\le\mathbb{P} \left( \sup_{\norm{\theta- \theta^{\star}} \ge \epsilon} 
    \sum_{t=1}^{T} R_{t}(\theta)  \ge \sum_{t=1}^{T} R_{t}(\theta^{\star})    \right) \\
    &= \mathbb{P} \left( \sup_{\norm{\theta- \theta^{\star}} \ge \epsilon} 
\left\{\frac{1}{T}\sum_{t=1}^{T} R_{t}(\theta) \right\} - \frac{1}{T}\sum_{t=1}^{T} R_{t}(\theta^{\star}) \ge 0\right) \\
    &\le \mathbb{P}  \Bigg( \sup_{\norm{\theta- \theta^{\star}} \ge \epsilon} \left\{  
    \frac{1}{T} \sum_{t=1}^{T} R_{t}(\theta) - \mathbb{E}\left[R_{t}(\theta) | \mathcal{H}_{t-1} \right] \right\}   +  \sup_{\norm{\theta- \theta^{\star}} \ge \epsilon} \left\{ \frac{1}{T} \sum_{t=1}^{T} \mathbb{E}\left[ R_{t}(\theta) - R_{t}(\theta^{\star})| \mathcal{H}_{t-1} \right]  \right\} - \\ 
    &\qquad \qquad \frac{1}{T} \sum_{t=1}^{T} \left\{R_{t}(\theta^{\star}) - \mathbb{E} \left[R_{t}(\theta^{\star}) | \mathcal{H}_{t-1} \right]\right\} \ge 0  \Bigg).\\
\end{align*}
We consider each term separately. First, note that in all cases we have that

\begin{align*}
    \E\left[R_{t}(\theta)| \mathcal{H}_{t-1} \right] &=  \E_{\mathcal{P},\pi_{t}}\left[\sum_{a=1}^{K} \pi_{e}(A_{t} =a | X_{t}) \left( m_{\theta}(a,X_{t},f_{t}(a,X_{t}))   + \mathbbm{1}_{A_{t} = a} \frac{ m_{\theta}(X_{t},A_{t},Y_{t})  -  m_{\theta}(X_{t},A_{t},f_{t}(X_{t},A_{t})) }{ \mathbb{P} \left(A_{t} =a | X_{t},\mathcal{H}_{t-1} \right)  }\right) |\mathcal{H}_{t-1}\right] \\
    &= \E_{\mathcal{P},\pi_{t}}\left[\frac{ \pi_{e}(A_{t} =a | X_{t})}{\mathbb{P} \left(A_{t} =a | X_{t},\mathcal{H}_{t-1} \right)} m_{\theta}(X_{t},A_{t},Y_{t}) \right ] + \\
    &\qquad \sum_{a =1}^{K}\pi_{e}(A_{t} =a | X_{t}) \E_{\mathcal{P},\pi_{t}} \left[ 
    \left( 1 - \frac{ \mathbbm{1}_{A_t=a}  }{  \mathbb{P} \left(A_{t}=a| X_{t}, \mathcal{H}_{t-1} \right)} \right)m_{\theta}(a,X_{t},f_{t}(a,X_{t}))  \right] \\
    &= \E_{\pi_{e} }\left[ m_{\theta}(X_{t},A_{t},Y_{t}) \right ]. 
\end{align*}

Therefore,
\begin{align*}
\sup_{\norm{\theta- \theta^{\star}} \ge \epsilon} \left\{ \frac{1}{T} \sum_{t=1}^{T} \mathbb{E}\left[ R_{t}(\theta) - R_{t}(\theta^{\star})| \mathcal{H}_{t-1} \right]  \right\} &= \sup_{\norm{\theta- \theta^{\star}} \ge \epsilon} \left\{\mathbb{E}_{\pi_{e}}\left[ m_{\theta}(X_{t},A_{t},Y_{t}) - m_{\theta^{\star}}(X_{t},A_{t},Y_{t})\right]  \right\} < - \delta_{1}, 
\end{align*}
for some $\delta_1 > 0$ by \cref{assumption:maximum}. Note that \cref{assumption:bracket3} implies that the bracketing number $N_{[]}(\epsilon, \mathcal{M}_{\Theta}, L_{2}(\mathcal{P},\pi_{e})) < \infty$ for $ \mathcal{M}_{\Theta} =  \{ m_{\theta}(X_t,A_{t},Y_{t}) : \theta \in \Theta  \}$. This combined with \cref{assumption:bracket_complexity} allows us to apply \cref{lemma:uniform_lln} conclude that the first and third terms also conerge to $0$. Therefore,  $\mathbb{P} \left( \norm{\tilde{\theta}_{T} - \theta^{\star}}_{1} \ge \epsilon \right) \rightarrow 0 $.

\end{proof}
\subsection{Proof of \cref{prop:estimator_sequential}}
By the law of large numbers, we have that $\frac{n_p}{T} \rightarrow{p} p$ and therefore $\frac{1}{n_p} = \frac{1}{pT} + o_p(T^{-1})$. Define the per-round variance summand
\begin{align*}
\Psi_t:= 
\Big(\widehat{\nu}_t(X_t)-\overline{\nu}_t\Big)
\Big(\widehat{\nu}_t(X_t)-\overline{\nu}_t\Big)^\top + \sum_{a=1}^K
\frac{\pi_e(a\mid X_t)^2}{\mathbb P(A_t=a\mid X_t,\mathcal H_{t-1})}\,
h_t(X_t,a)
\;-\;
\widehat{\nu}_t(X_t)\widehat{\nu}_t(X_t)^\top.
\end{align*}
We can therefore write $\widehat V_t = \frac{1}{n_p}\sum_{t=1}^T \zeta_t\,\Psi_t$. Moreover, 
$\widehat V_t
=
\frac{1}{pT}\sum_{t=1}^T \zeta_t\,\Psi_t
+ o_p(1)$. 

By assumption, $\{\zeta_t\}$ are independent of all histories, therefore $
\E[\zeta_t\Psi_t \mid \mathcal H_{t-1}]
=
p\,\E[\Psi_t \mid \mathcal H_{t-1}]$. 
Consequently, by the law of large numbers and Slutsky’s theorem,
$
\widehat V_t
=
\E[\Psi_t \mid \mathcal H_{t-1}]
+ o_p(1).
$
The remainder of the argument is identical to \cref{prop:estimator} and, therefore,
\[
\big\|\widehat V_t - V_{t,\theta^\star}\big\|_{\mathrm{op}}
\xrightarrow{p} 0,
\]
as claimed.
\section{Procedure for Semi-Synthetic Data Construction and Additional Results} \label{appendix:dataset_construction}
For the semi-synthetic dataset construction, we learn $E[Y_{t} | X_{t}]$ via a machine learning model, and then enforce a linear model to describe $E[Y_{t} | X_{t}, A_{t}]$. The procedure is
 \begin{enumerate}
     \item Train a machine learning model $f$ to predict $E[Y_{t} | X_{t}]$ using available dataset. 
     \item Generate synthetic causal outcomes, where it is assumed that the mean
    $$ E[Y_t | X_{t}, A_{t}] = \sum_{a=1}^{K}\beta^{a}_{1} \mathbbm{1}_{A_{t}=a} + \sum_{a=1}^{K}\beta_{2}^{a} f(X_{t}) \times \mathbbm{1}_{A_{t}=a},$$
     for user chosen parameters $\beta^{a}_{1}$, $\beta^{a}_{2}$.
     \item We observe each $X_{t}$ sequentially (sampled with replacement from the actual OAI data), and then in each round $A_{t}$ are chosen using each of the methods in \cref{sec:empirical_results}. We then draw $Y_{t} \sim N\left(\sum_{a=1}^{k}\beta^{a}_{1} \mathbbm{1}_{A_{t}=k} + \sum_{a=1}^{k}\beta_{2}^{a} f(X_{t}) \times \mathbbm{1}_{A_{t}=k},1 \right)$.
     \item  Alternatively, we can introduce heteroskedasticity into the dataset by learning a secondary model $v$ to predict $E\left[\left(Y_{t} - f(X_{t})\right)^{2} | X_{t}\right]$. We then sample 
     $$\tilde{Y}_{t} \sim N\left(\sum_{a=1}^{k}\beta^{a}_{1} \mathbbm{1}_{A_{t}=a} + \sum_{a=1}^{k}\beta_{2}^{a} f(X_{t}) \times \mathbbm{1}_{A_{t}=a}, \sum_{a=1}^{K} v(X_{t}) \gamma_{a} \mathbbm{1}_{A_{t}=a} \ \right),$$ where $\gamma_{k}$ is a user-chosen multiplier to the variance. 

 \end{enumerate}
For the empirical simulations, we choose six different initializations 
\begin{enumerate}
\item \textbf{Scenario 1: Correct Specification, Homoscedasticity, Unique Optimal Arm} 
$$\beta_{1} = (0,1,2,3,4,5,6,7), \beta_{2} = (0,0,0,0,0,0,0,0), \text{and unit variance uniformly}.$$ In this scenario, all methodologies other than naive unweighted MLE should perform relatively well as bandit algorithms \emph{will concentrate}. Results are shown in \cref{fig:simulation_results_6}.
\item \textbf{Scenario 2: Correct Specification, Homoscedasticity, No Unique Optimal Arm.}
$$\beta_{1} = (0,0,1,2,2,3,5,5), \beta_{2} = (0,0,0,0,0,0,0,0), \text{and unit variance uniformly}.$$ In this scenario, IPW-style estimators may have difficulties but confidence intervals constructed as in \cite{zhang2021mestimators} should still cover correctly. Results are shown in \cref{fig:simulation_results_5}.
\item \textbf{Scenario 3: Misspecification, Homoscedasticity.}
$$\beta_{1} = (0,0,1,2,2,3,4,4), \beta_{2} = (1,-1,1,0,1,1,1,-3), \text{and unit variance uniformly}.$$ Results are shown in \cref{fig:simulation_results_1}.
\item \textbf{Scenario 4: Misspecification, Heteroskedasticity}
$$\beta_{1} = (0,0,1,2,2,3,4,5), \beta_{2} = (1,-1,1,0,1,1,1,-2), \text{and } \gamma = (1,2,3,4,5,5,5,5) \times 0.2. $$
Results are shown in \cref{fig:simulation_results_3}.

\end{enumerate}
We showed the results of only Scenario 2 in the main body due to space constraints, but the other scenarios reveal broadly comparable results. 

\begin{figure*}  
    \centering
    \begin{subfigure}[t]{\linewidth}
         \centering
         
         \includegraphics[width=\textwidth]{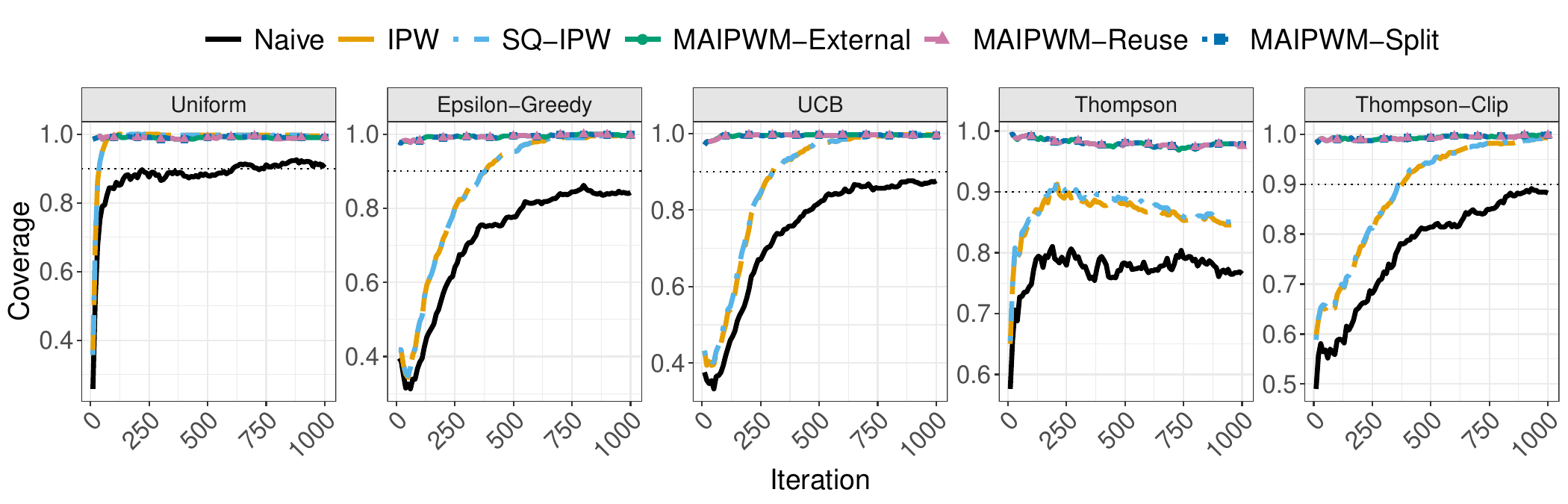}
        \caption{Coverage for target $1-\alpha=0.9$}
     \end{subfigure}
    \begin{subfigure}[t]{\linewidth}
         \centering
         \includegraphics[width=\textwidth]{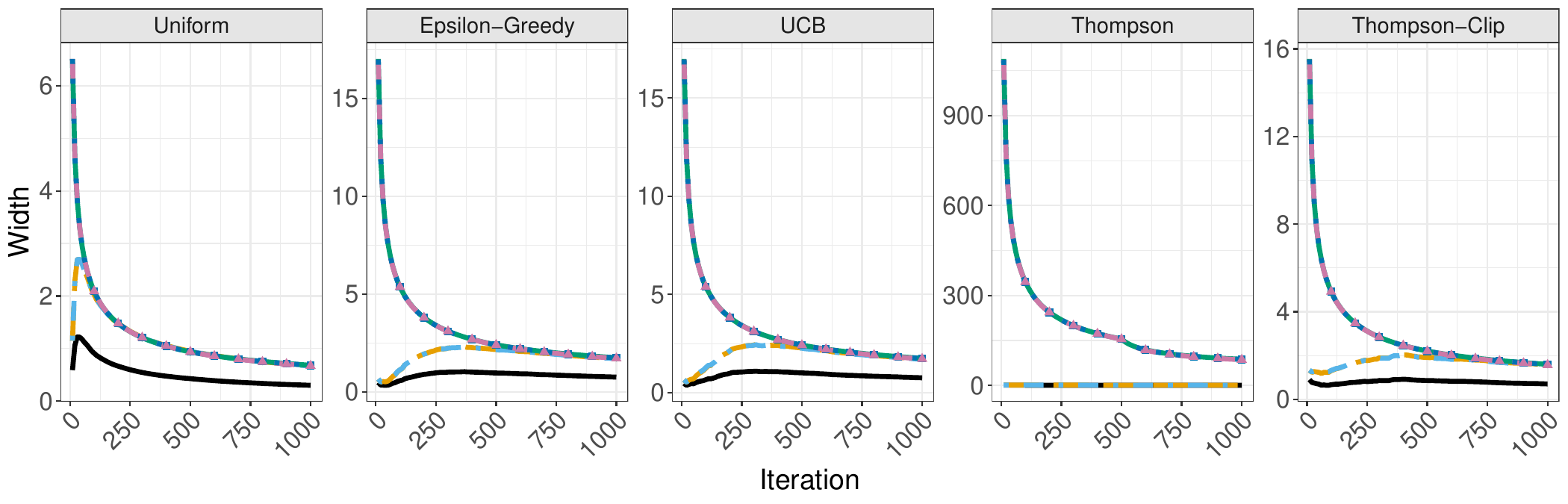}
        \caption{Confidence interval width}
     \end{subfigure}
\caption{Simulation results for scenario 1, where the model is correctly specified. In these settings, we see that all methods cover, but it often takes the naive approaches (IPW and SQ-IPW) estimates significantly more samples to reach the nominal coverage rates. Thompson sampling undercovers for several methods when propensity scores are not clipped, consistent with theory. }
    \label{fig:simulation_results_6}
\end{figure*}

\begin{figure*}  
    \centering
    \begin{subfigure}[t]{\linewidth}
         \centering
         
         \includegraphics[width=\textwidth]{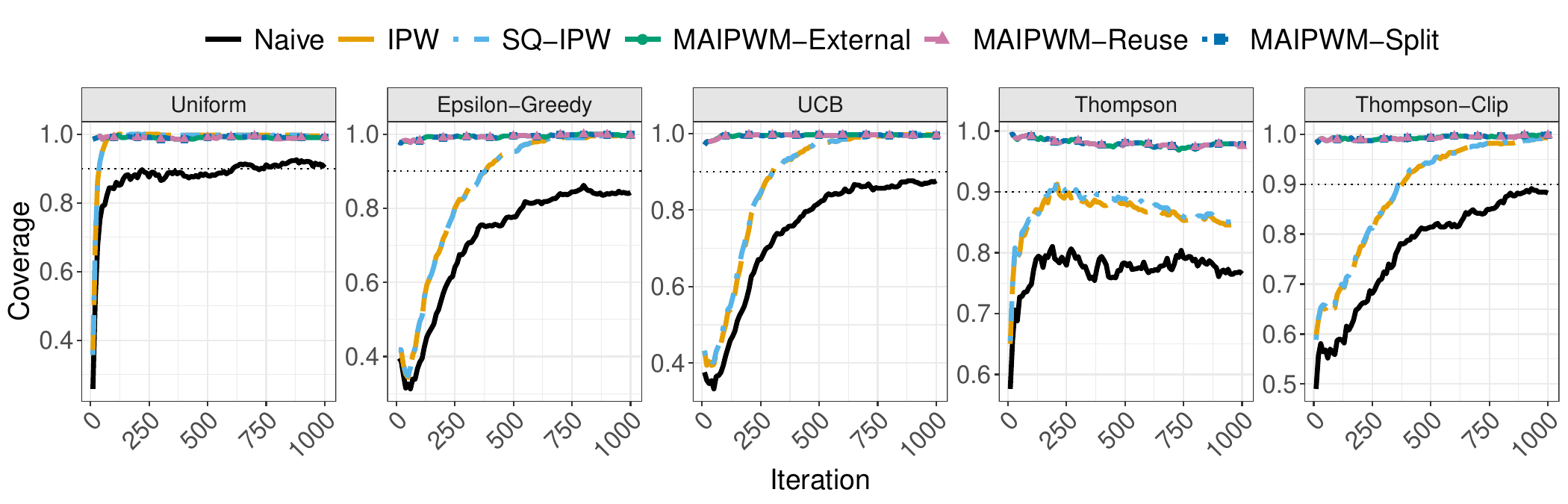}
        \caption{Coverage for target $1-\alpha=0.9$}
     \end{subfigure}
    \begin{subfigure}[t]{\linewidth}
         \centering
         \includegraphics[width=\textwidth]{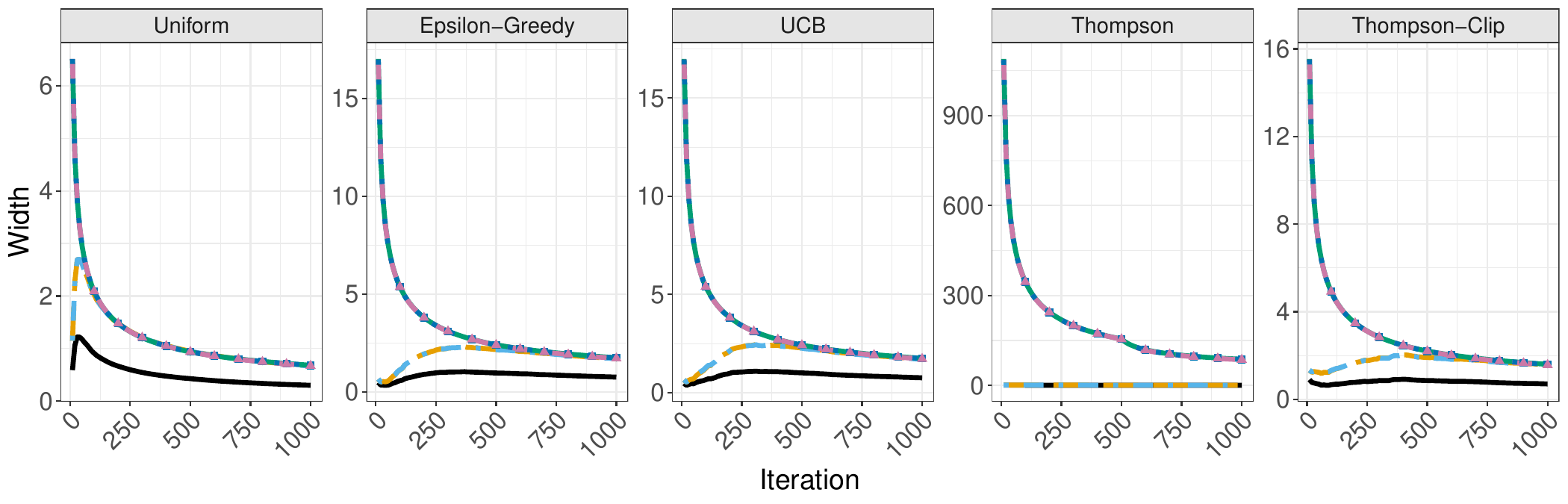}
        \caption{Confidence interval width}
     \end{subfigure}
\caption{Simulation results for scenario 2, where the model is correctly specified but there is not a unique optimal arm. The results are broadly similar as in \cref{fig:simulation_results_5}, though we note there is no theoretical guarantee that IPW estimates should cover in this setting (though they do empirically).}
    \label{fig:simulation_results_5}
\end{figure*}

\begin{figure*}  
    \centering
    \begin{subfigure}[t]{\linewidth}
         \centering
         \includegraphics[width=\textwidth]{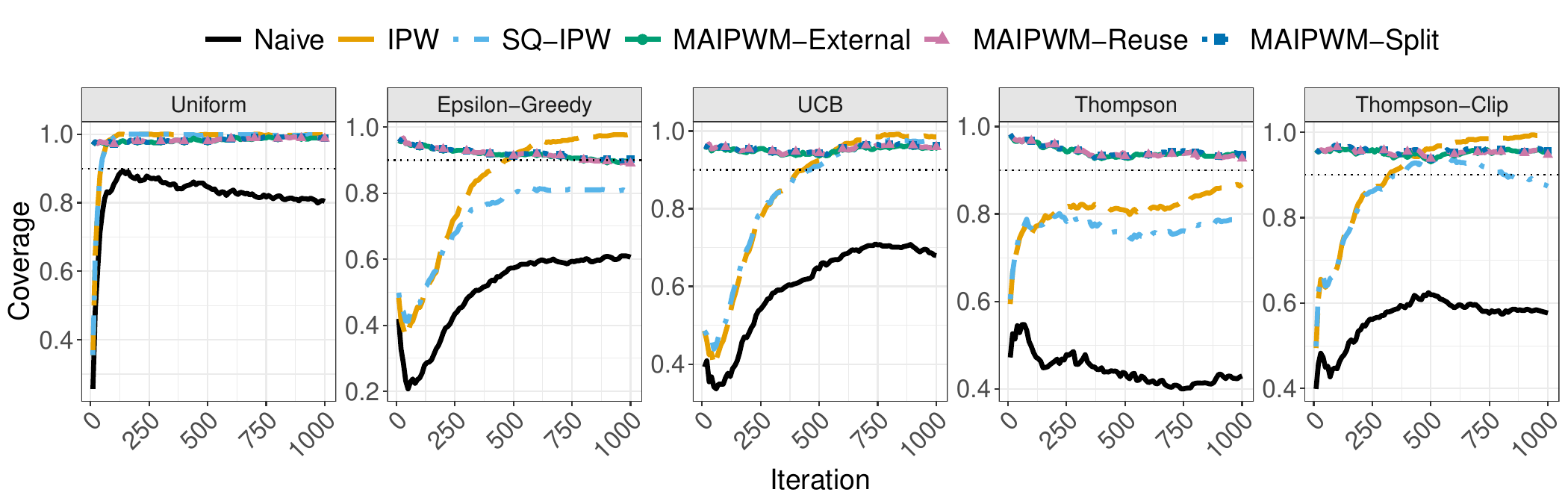}
        \caption{Coverage for target $1-\alpha=0.9$ }
     \end{subfigure}
    \begin{subfigure}[t]{\linewidth}
         \centering
         \includegraphics[width=\textwidth]{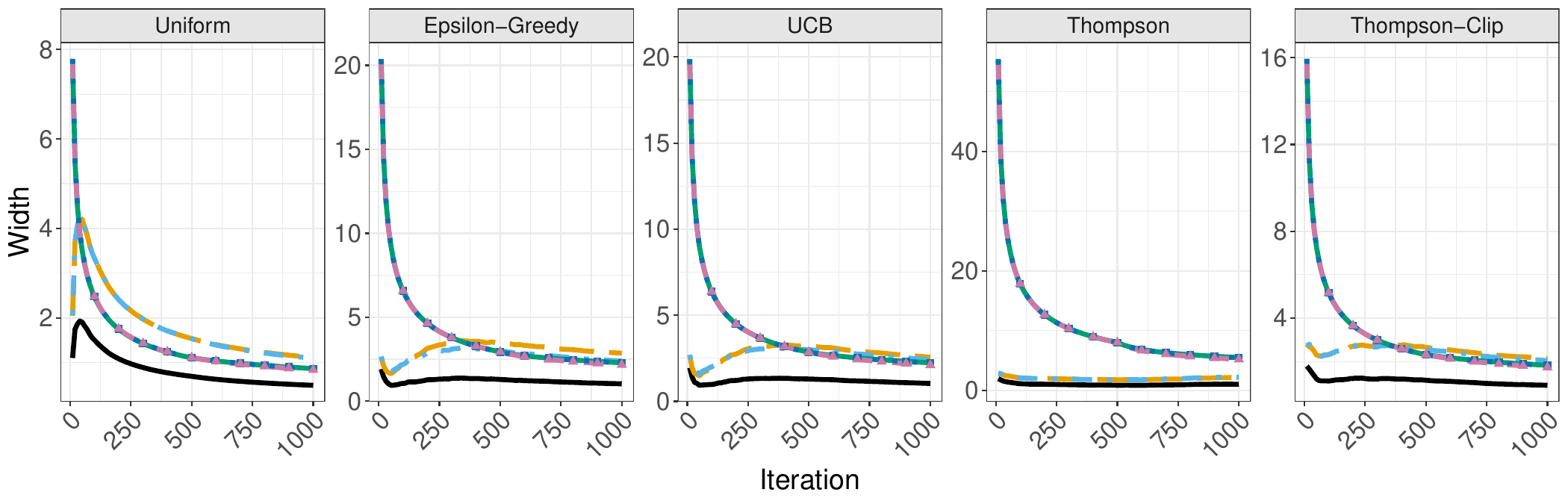}
        \caption{Confidence interval width }
     \end{subfigure}
\caption{Simulation results for scenario 3, when there is model misspecification but homoskedastic errors. IPW and SQ-IPW estimators now do not cover in every scenario, as their theoretical guarantees depend on \emph{correct specification} of the working models.}
    \label{fig:simulation_results_1}
\end{figure*}

\begin{figure*}  
    \centering
    \begin{subfigure}[t]{\linewidth}
         \centering
    
         \includegraphics[width=\textwidth]{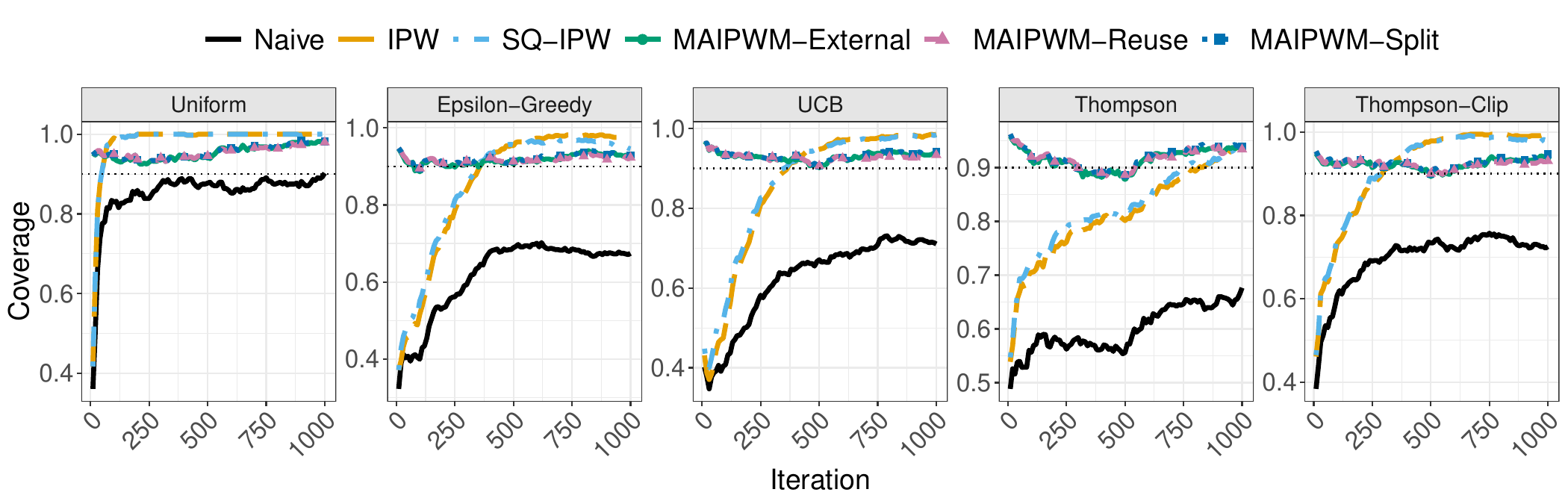}
        \caption{Coverage for target $1-\alpha=0.9$}
     \end{subfigure}
    \begin{subfigure}[t]{\linewidth}
         \centering
         \includegraphics[width=\textwidth]{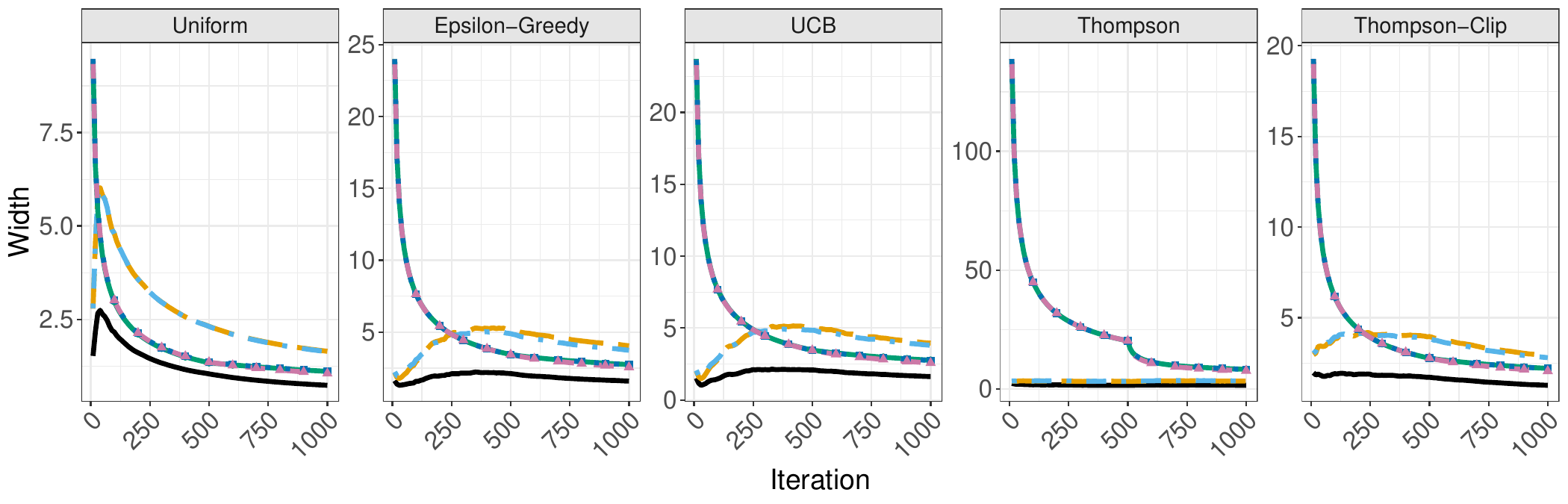}
        \caption{Confidence interval width}
     \end{subfigure}
\caption{Simulation results for scenario 3, when there is model misspecification and heteroskedastic errors. The results are broadly similar to \cref{fig:simulation_results_1}. }
    \label{fig:simulation_results_3}
\end{figure*}

%% file: acknowledgements.tex
Data and/or research tools used in the preparation of this manuscript were obtained and analyzed from the controlled access data sets distributed by the Osteoarthritis Initiative (OAI), a data repository located within the NIMH Data Archive (NDA). OAI is a collaborative informatics system created by the National Institute of Mental Health and the National Institute of Arthritis, Musculoskeletal and Skin Diseases (NIAMS) to provide a worldwide resource to quicken the pace of biomarker identification, scientific investigation and OA drug development. Dataset identifier(s): [outcomes99, kxrsqbu06, kxrsqbu00, enrollees, allclinical00, allclinical06].